\documentclass[a4paper,fleqn]{cas-sc}

\usepackage[numbers]{natbib}
\usepackage{algorithm}
\usepackage{algpseudocode}
\usepackage{amsthm}
\usepackage{amsmath}
\usepackage{amssymb}
\usepackage{gensymb}
\usepackage{ulem}
\usepackage{booktabs} % 用于生成专业的表格线（\toprule, \midrule, \bottomrule）
\usepackage{caption} % 提升图表标题的控制能力
\usepackage{array} % 提供对表格列格式化的更精细控制
\usepackage{xcolor}
\usepackage{multirow}
\usepackage[flushleft]{threeparttable} % ← 关键：左对齐
\usepackage{mathtools}
\usepackage{adjustbox}

\newtheorem{theorem}{Theorem}
\newtheorem{lemma}[theorem]{Lemma}

\def\tsc#1{\csdef{#1}{\textsc{\lowercase{#1}}\xspace}}
\tsc{WGM}
\tsc{QE}
\begin{document}
\let\WriteBookmarks\relax
\def\floatpagepagefraction{1}
\def\textpagefraction{.001}

\shorttitle{Curve resampling based high-quality high-order unstructured quadrilateral mesh generation}    

\shortauthors{Weng et~al.}  

\title[mode = title]{Curve resampling based high-quality high-order unstructured quadrilateral mesh generation}  

\tnotemark[1] 

\tnotetext[1]{This research was supported by the National Key Laboratory of Computational Physics (No. JK2025-06), National Natural Science Foundation of China (Nos. 62272402, 62372389) and Natural Science Foundation of Fujian Province (No. 2024J01513243).} 


\cortext[cor1]{Corresponding author}
\author[label1]{Yongjia Weng}
\ead{wengyongjia@stu.xmu.edu.cn}
\author[label2]{Lufeng Liu}
\ead{liu_lufeng@iapcm.ac.cn}
\author[label3]{Zhonggui Chen}
\ead{chenzhonggui@xmu.edu.cn}
\author[label2]{Xuan Zhou}
\ead{zhou_xuan@iapcm.ac.cn}
\author[label1]{Juan Cao\corref{cor1}}[orcid=0000-0002-8154-4397]
\ead{juancao@xmu.edu.cn}
\affiliation[label1]{organization={School of Mathematical Sciences, Xiamen University},%Department and Organization
            city={Xiamen},
            postcode={361005}, 
            state={Fujian},
            country={China}}
\affiliation[label2]{organization={Institute of Applied Physics and Computational Mathematics},%Department and Organization
            % city={Beijing},
            postcode={100094}, 
            state={Beijing},
            country={China}}
            
\affiliation[label3]{organization={School of Informatics, Xiamen University},%Department and Organization
            city={Xiamen},
            postcode={361005}, 
            state={Fujian},
            country={China}}

% \author[1]{}%[<options>]

% % Corresponding author indication
% \cormark[1]

% % Footnote of the first author

% % Email id of the first author
% \ead{}

% % URL of the first author

% % Credit authorship
% % eg: \credit{Conceptualization of this study, Methodology, Software}

% % Address/affiliation
% \affiliation[1]{organization={},
% %          citysep={}, % Uncomment if no comma needed between city and postcode

% \author[2]{}%[]

% % Footnote of the second author

% % Email id of the second author
% \ead{}

% % URL of the second author

% % Credit authorship

% % Address/affiliation
% \affiliation[2]{organization={},
% %          citysep={}, % Uncomment if no comma needed between city and postcode

% % Corresponding author text
% \cortext[1]{Corresponding author}



\begin{abstract}
High-order quadrilateral meshes offer superior accuracy and computational efficiency in numerical simulations. However, existing methods struggle to simultaneously preserve boundary/interface features, ensure high quality, and achieve efficient generation, particularly for complex geometries where degenerate and inverted elements frequently occur. To address this issue, this paper proposes a high-quality high-order unstructured quadrilateral mesh generation method based on geometric error-bounded curve reconstruction, which employs an indirect approach to enforce interface consistency. By optimization-based curve reconstruction strategies, our method improves mesh quality while maintaining the validity of high-order elements. Compared to direct high-order mesh optimization techniques, our approach reduces the optimization problem to curve reconstruction problem, significantly lowering computational complexity and enhancing efficiency. Experimental results demonstrate that the proposed method efficiently generates high-quality high-order quadrilateral meshes while preserving boundary/interface geometric features, offering improved adaptability and numerical stability in complex geometries.
\end{abstract}


% \begin{highlights}
% \item 
% \item 
% \item 
% \end{highlights}

% Keywords

\maketitle

% Main text
\section{Introduction}
\label{sec1}
In recent years, high-order numerical methods have attracted widespread attention in both academia and industry due to their higher accuracy and faster convergence rates \cite{lopez2020cad,ferguson2023high,le2023second,dobrev2012high}. High-order numerical methods employ high-order elements with more degrees of freedom (DOFs), which not only enhance the approximation capability of numerical schemes but also allow for better geometric representation, thereby improving the overall simulation accuracy.

The generation techniques for high-order triangular or tetrahedral meshes are relatively mature, and include approaches based on shell mapping \cite{jiang2021bijective}, optimization-based strategies \cite{feng2018curved,ji2024evolutionary}, constructive methods \cite{mandad2020bezier,khanteimouri20233d}, and sampling techniques \cite{hu2019triwild}. However, compared to triangular meshes, quadrilateral meshes at the same mesh resolution support higher-order interpolation spaces, thus enhancing the approximation power and convergence order of numerical methods \cite{schneider2022large}. In numerical simulations, quadrilateral meshes often exhibit superior numerical stability and convergence efficiency, making them widely adopted in high-precision simulation fields such as computational fluid dynamics. For example, in large-deformation multi-phase flow problems, the widely used high-order arbitrary Lagrangian–Eulerian (ALE) method~\cite{anderson2018high} critically requires high-order, high-quality quadrilateral meshes to enhance its robustness.
%High-order arbitrary Lagrangian–Eulerian (ALE) framework~\cite{anderson2018high} is widely used in the simulation of large-deformation problems and plays a crucial role in applications such as fluid-structure interaction.

Quadrilateral meshes can be classified into four categories: regular mesh, semi-regular mesh, valence semi-regular mesh, and unstructured mesh \cite{bommes2013quad}. For models with complex geometries, unstructured meshes are particularly suitable for capturing intricate geometric features. In many application domains \cite{liu2023topology}, such as multi-phase flow problems, it is often necessary to perform mesh partitioning on models composed of different materials. These models typically contain material interfaces, and the generated mesh must maintain interface constraints while ensuring consistency across the interface.

Compared to linear meshes, high-order meshes offer stronger geometric representation capabilities and can more accurately conform to model boundaries. This allows for a reduction in the number of required mesh elements while maintaining a desired level of both geometric and numerical accuracy. Additionally, high-order mesh elements possess more internal DOFs, which facilitates the development of algorithms with high computational density and low communication overhead—advantages that are beneficial for parallel acceleration \cite{mittal2025general}.

Moreover, degradation in mesh quality can significantly impact the accuracy of subsequent numerical simulations. Especially when inverted elements appear, the numerical solution may become severely compromised \cite{liu2023topology}. This issue becomes even more critical in moving-mesh computations, where mesh quality directly affects both the stability and efficiency of the numerical solver. High-quality meshes can effectively reduce numerical errors, accelerate algorithmic convergence, and, to some extent, allow for larger time steps, thereby improving computational efficiency.

However, generating high-order quadrilateral meshes that are feature-preserving, efficient, and high-quality remains a significant challenge. Existing methods often struggle to achieve all three objectives simultaneously. For example, quadtree-based templated refinement approaches~\cite{kopriva2024hohqmesh} are relatively efficient but tend to produce low quality elements near input curves, sometimes resulting in inverted elements, and they cannot handle non-closed interface curves.  Block-structured methods~\cite{marcon2019naturally} and optimization-based approaches~\cite{dobrev2019target} can produce higher-quality meshes but at considerable computational cost. In particular, optimization-based methods require managing a large number of variables for high-order elements, making them sensitive to initialization and expensive to compute.

This paper presents a novel method for efficiently generating high-quality, high-order, planar, unstructured quadrilateral meshes based on geometric error-bounded curve reconstruction. The proposed approach leverages the strong correlation between mesh quality and boundary curve segment distribution in high-order quadrilateral meshes, simplifying the high-dimensional optimization problem into a boundary curve reconstruction task. As a result, the method achieves both computational efficiency and high mesh quality, while preserving boundary and interface features with geometric accuracy within the prescribed error tolerance. Extensive numerical experiments demonstrate that the method outperforms state-of-the-art techniques in terms of efficiency, mesh quality, and interface conformity, making it a promising solution for practical applications requiring high-fidelity discretizations. Our specific contributions include:

\begin{itemize} 
\item We propose a geometric error-bounded high-quality high-order quadrilateral mesh generation method. The method effectively preserves complex geometric features embedded in input polynomial curves while ensuring conformity across material interfaces. It also avoids low-quality or degenerate elements, which is crucial for reliable downstream simulations.

\item We develop an efficient high-order quadrilateral mesh generation algorithm by introducing an adaptive boundary curve reconstruction strategy that reformulates high-order mesh optimization as a boundary reconstruction problem. This reduction in optimization DOFs significantly enhances the efficiency and robustness of the entire meshing pipeline.

\item We integrate the proposed meshing approach into a high-order ALE framework, effectively mitigating mesh distortion typical of traditional Lagrangian method. This integration enhances stability and accuracy in numerical simulations, demonstrating the robustness and practical applicability of our method for large-deformation problems, while preserving clear material interfaces compared to existing high-order ALE methods~\cite{anderson2018high}.
\end{itemize}

\section{Related work}

The methods for generating high-order quadrilateral meshes can primarily be categorized into direct methods and indirect methods. Direct methods generate high-order quadrilateral meshes directly, whereas indirect methods typically involve first generating linear quadrilateral meshes and then using these as a basis to construct high-order quadrilateral meshes.

\textbf{Direct methods.} High-order quadrilateral mesh generation via direct methods includes approaches based on block decomposition \cite{marcon2019naturally} and quadtree-based methods \cite{kopriva2024hohqmesh}. The patching-based approach \cite{marcon2019naturally} initially generates high-order triangular meshes, constructs guiding fields from them to determine the positions of irregular nodes, and calculates separation lines to partition the model into high-order quadrilateral elements. While this method can produce high-quality high-order quadrilateral meshes, it incurs significant computational costs and is challenging to apply to complex geometries or problems requiring multiple remeshing iterations. HOHQMesh \cite{kopriva2024hohqmesh}, an open-source automated generator for high-order quadrilateral meshes, employs a quadtree-based template subdivision strategy to directly generate adaptive high-order quadrilateral meshes. Although quadtree-based methods are efficient, they suffer from poor element quality near input curves, may invert, and cannot support non-closed interface curves. Our proposed method efficiently generates high-quality high-order quadrilateral meshes and supports non-closed interface curves.

\textbf{Indirect methods.} Indirect methods for generating high-order quadrilateral meshes include those based on element deformation techniques. These methods start with a linear quadrilateral mesh, extract all high-order points within the mesh, and use optimization \cite{turner2018curvilinear,ruiz2016generation}, solving Winslow equations \cite{fortunato2016high}, or linear elasticity analogies \cite{xie2013generation} to adjust the positions of high-order points so that the curves conform closely to boundaries while maintaining element quality. However, the key challenge lies in generating initial linear quadrilateral meshes such that post-deformation elements remain free from inversion and maintain high quality. Noting the profound relationship between the quality of high-order quadrilateral meshes and boundary curve segment distribution, we propose a geometric error-bounded boundary reconstruction algorithm to provide a robust foundation of linear quadrilateral meshes for high-order mesh generation.

Generating high-quality high-order quadrilateral meshes often necessitates mesh optimization. Prominent among these methods is the target-matrix optimization paradigm (TMOP) \cite{dobrev2019target}. This approach usually starts with a poorly conditioned high-order mesh, defines mesh quality as an energy function, and optimizes the energy to achieve high-quality high-order quadrilateral meshes. However, in such optimization-based methods, all element nodes serve as variables, leading to a large number of DOFs, which increases the risk of getting trapped in local minima. Moreover, the optimization results heavily depend on the initial high-order mesh and are computationally expensive. Our proposed method leverages the intrinsic connection between mesh quality and boundary curve distribution, transforming the mesh quality optimization problem into an optimization of boundary point distribution, significantly reducing the number of optimization variables, lowering optimization complexity, and enhancing efficiency.

In isogeometric analysis, parametrization methods \cite{xu2018constructing,zheng2019boundary,wang2022tcb,khanteimouri2025c} can also yield high-order quadrilateral meshes, focusing on both the quality of parametric mappings and their continuity. For complex geometries, achieving both mapping continuity and low distortion simultaneously poses a significant challenge, making it difficult to attain high-quality parametrization. Such methods generally require maintaining high-order continuity of the mapping, which inevitably leads to degenerate elements when handling input boundaries with high-order continuity. As our method does not require higher-order geometric continuity between elements, maintaining only $C_{0}$ continuity across elements, it enables the generation of high-quality high-order quadrilateral meshes for complex geometries.

\section{Preliminary}
In this paper, both Lagrange and Bernstein polynomials are used interchangeably as basis functions for representing high-order curves and elements \cite{ergatoudis1968curved,farin2001curves}. We employ Bernstein polynomials for geometric analysis and Lagrange polynomials for high-order point sampling and high-order element deformation. The following sections introduce the specific definitions and notations employed.

\textbf{Lagrange polynomial curves.}  
Given $n+1$ interpolation points $\mathbf{P}_0, \dots, \mathbf{P}_n \in \mathbb{R}^2$, an degree-$n$ Lagrange polynomial curve $C(t)$ interpolating these points is defined as
\begin{equation}
C(t) = \sum_{i=0}^{n} \mathbf{P}_i L_i(t), \quad t \in [0, 1],
\end{equation}
where $L_i(t)$ are the Lagrange basis functions given by
\begin{equation}
L_i(t) = \prod_{\substack{0 \le j \le n \\ j \ne i}} \frac{t - t_j}{t_i - t_j}.
\end{equation}
Each basis function satisfies $L_i(t_j) = \delta_{ij}$, where $\delta_{ij}$ is the Kronecker delta. The parameter values $t_i$ correspond to the points $\mathbf{P}_i$, such that $C(t_i) = \mathbf{P}_i$. In particular, $\mathbf{P}_0$ and $\mathbf{P}_n$ are the endpoints of the curve; the remaining interpolation points $\mathbf{P}_i$ ($1 \le i \le n-1$) are called high-order interpolation points. In this work, the input curves are uniformly parameterized with $t_i = \frac{i}{n}$.

\textbf{B\'{e}zier curves.}
Given $n+1$ control points $\mathbf{Q}_0, \dots, \mathbf{Q}_n \in \mathbb{R}^2$, an degree-n B\'{e}zier curve $C_B(t)$ is defined as
\begin{equation}
C_B(t) = \sum_{i=0}^{n} \mathbf{Q}_i B_i^n(t), \quad t \in [0,1],
\end{equation}
where $B_i^n(t)$ are the Bernstein polynomials,
\begin{equation}
B_i^n(t) = \binom{n}{i} t^i (1 - t)^{n-i}, \quad i = 0, 1, \ldots, n.
\end{equation}
The Bernstein polynomials possess the property of nonnegative and partition of unity, i.e.,
\begin{equation}
B_{i}^{n}(t)\geq 0,\quad \sum_{i=0}^{n} B_{i}^{n}(t) \equiv 1,\quad t\in[0,1].
\end{equation}
The curve passes through the first and last control points, i.e., $C_B(0) = \mathbf{Q}_0$ and $C_B(1) = \mathbf{Q}_n$, and lies entirely within the convex hull of its control polygon. Define $\Delta_i = \mathbf{Q}_{i+1} - \mathbf{Q}_i$ as the control vector of the B\'{e}zier curve at control point $\mathbf{Q}_i$.

\textbf{High-order quadrilateral element.}   
A high-order quadrilateral mesh is composed of a collection of high-order quadrilateral elements. Let $Q$ be a degree-$n$ quadrilateral element with curved boundaries, defined by a mapping $\psi: S \to Q$ from the reference square $S = [0,1]^2$ to $Q$, so that $Q = \psi(S)$. 

\begin{figure}
  \centering
\includegraphics[width=.9\textwidth]{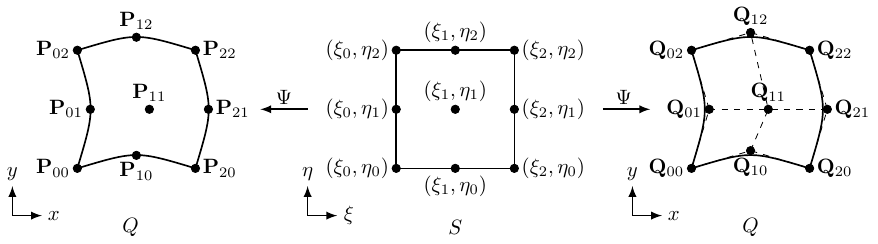}

     \parbox[t]{.31\textwidth}{\centering
            (a) 
           }
  \parbox[t]{.31\textwidth}{\centering
          (b)
           }
     \parbox[t]{.31\textwidth}{\centering
          (c)
           }

  \caption{\label{high_order lagrange element definition}
          Quadratic quadrilateral element representations. (b) Reference square and nodes; (a,c) the same element represented using Lagrange basis functions and Bernstein polynomials, respectively.}
\end{figure}  

Given the nodes $(\xi_0,\eta_0), \dots, (\xi_n,\eta_n)$ of the reference unit square (see Fig.~\ref{high_order lagrange element definition}(b)), the mapping $\psi(\xi,\eta)$ can be expressed using Lagrange basis functions as:  
\begin{equation}  \label{lagrange representation}
\psi(\xi,\eta) = \sum \mathbf{P}_{ij} L_{ij}(\xi,\eta), \quad (\xi,\eta) \in S,
\end{equation}  
where $\mathbf{P}_{00}, \dots, \mathbf{P}_{nn} \in \mathbb{R}^2$ are the Lagrange interpolation nodes of the degree-$n$ quadrilateral, and $L_{ij}(\xi,\eta) = L_i(\xi)L_j(\eta)$ are the tensor-product Lagrange basis functions satisfying $L_{ij}(\xi_i,\eta_j) = 1$ and $L_{ij}(\xi_{\hat{i}},\eta_{\hat{j}}) = 0$ for $(\hat{i},\hat{j}) \neq (i,j)$. Thus, $\psi(\xi_i,\eta_j) = \mathbf{P}_{ij}$, ensuring interpolation at the nodes (see Fig.~\ref{high_order lagrange element definition}(a)). For notational simplicity, we write $\psi$ instead of $\psi(\xi,\eta)$ when no ambiguity arises.

% \begin{figure}
%   \centering
% \includegraphics[width=.5\textwidth]{figures/figure1_bezier_element_show.pdf}

%   \caption{\label{high_order bezier element definition}
%            The definition of a quadratic quadrilateral element based on Bernstein polynomials.}
% \end{figure} 

Alternatively, $\psi$ can equivalently be defined using Bernstein polynomials as:  
\begin{equation}  \label{bezier representation}
\psi(\xi,\eta) = \sum \mathbf{Q}_{ij} B_{ij}^n(\xi,\eta),\quad (\xi,\eta) \in S,  
\end{equation}  
where $\mathbf{Q}_{ij}$ are the control points, and $B_{ij}^n(\xi,\eta) = B_i^n(\xi)B_j^n(\eta)$ are the tensor-product Bernstein polynomials (see Fig.~\ref{high_order lagrange element definition}(c)). Unlike Lagrange representations, Bernstein-based mappings do not interpolate all control points. The two definitions of the mapping $\psi$ (Eq.~(\ref{lagrange representation}) and Eq.~(\ref{bezier representation})) are equivalent, as they can be mutually transformed; the explicit conversion is provided in \cite{ainsworth2016computing}. Bernstein-based mappings facilitate geometric analysis through the control polygon, enabling the establishment of a relationship between element quality and its control polygon, which can guide boundary curve refinement.

The quality of the high-order quadrilateral element is determined by the properties of the mapping $\psi$. Specifically, the Jacobian matrix $J = \partial \psi$ measures local geometric distortion. The mesh quality metric used in this work is based on a variant of the element’s Jacobian; see Section~\ref{quality metrics} for details.

\section{Algorithmic and theoretical foundations} 
The Hausdorff distance is commonly employed to quantify the geometric discrepancy between the model boundary and the mesh boundary~\cite{remacle2014optimizing}. In this section, we present algorithms and theorems for computing an upper bound on the Hausdorff distance between two B\'ezier curves, providing the basis for curve reconstruction within a prescribed error threshold~$\epsilon$ in the subsequent section. We then derive a sufficient condition for the validity of high-order quadrilateral elements, which supports algorithms for generating high-quality high-order quadrilateral meshes. Proofs of all results are deferred to \ref{Proof of Lemma 1}-\ref{proof of theorem valid high-order quad}.

\textbf{Hausdorff distance estimation.} Given two degree-$n$ self-intersection free Bézier curves $\mathbf{l}_1$ and $\mathbf{l}_2$ with control points $\{\mathbf{Q}_{i1}\}_{i=0}^{n}$ and $\{\mathbf{Q}_{i2}\}_{i=0}^{n}$, respectively, the Hausdorff distance ~\cite{surazhsky2004sampling} between them is defined as  
\begin{equation}  
d_{H}(\mathbf{l}_{1},\mathbf{l}_{2})= \max\left\{ \adjustlimits\sup_{\mathbf{x} \in \mathbf{l}_1} \inf_{\mathbf{y} \in \mathbf{l}_2} d(\mathbf{x}, \mathbf{y}), \adjustlimits\sup_{\mathbf{y} \in \mathbf{l}_2} \inf_{\mathbf{x} \in \mathbf{l}_1} d(\mathbf{x}, \mathbf{y}) \right\},  
\end{equation}  
where $d(\mathbf{x}, \mathbf{y}) = \|\mathbf{x} - \mathbf{y}\|_2$. 
In our framework, verifying $d_H(\mathbf{l}_1, \mathbf{l}_2)<\epsilon$ is a recurring task. Sampling-based methods~\cite{surazhsky2004sampling} are sensitive to the number of samples, making it difficult to balance accuracy and efficiency. We instead estimate an upper bound of the Hausdorff distance directly from control point distances to test the condition $d_H(\mathbf{l}_1, \mathbf{l}_2) < \epsilon$ efficiently. Let $d_m(\mathbf{l}_1, \mathbf{l}_2) = \max\limits_{i} d(\mathbf{Q}_{i1}, \mathbf{Q}_{i2})$ denote the maximum distance between corresponding control points of curves $\mathbf{l}_1$ and $\mathbf{l}_2$. 
\begin{lemma}\label{lemma 1} The Hausdorff distance between $ \mathbf{l}_1 $ and $ \mathbf{l}_2 $ satisfies  
\begin{equation}
d_H(\mathbf{l}_1, \mathbf{l}_2) \leq d_{m}(\mathbf{l}_1,\mathbf{l}_2).
\end{equation}
\end{lemma}
% However, $d_m$ alone is too loose as an upper bound for the Hausdorff distance. To tighten this bound, we subdivide the curves using the following algorithm.

% By the subdivision convergence and convex hull properties of B\'ezier curves \cite{farin2001curves}, the bound on $d_H(\mathbf{l}_1, \mathbf{l}_2)$ converges under recursive subdivision, as stated in Lemma~\ref{lemma 2}.
We subdivide the curves using the following algorithm, the bound $b_k$ is monotonically decreasing under recursive subdivision, as stated in Lemma~\ref{lemma 2}.

\begin{algorithm}[htbp!]
    \caption{The upper bound of $d_H(\mathbf{l}_1,\mathbf{l}_2)$ computation.}
    \label{alg:Framwork}
    \begin{algorithmic}[1]
      \Require
        Two self-intersection free B\'{e}zier curves $\mathbf{l}_1,\mathbf{l}_2$; the given geometric error threshold $\epsilon$; the maximum refinement iteration $k_{0}$.
      \Ensure
        The upper bound $b_{k}$ of $d_H(\mathbf{l}_1,\mathbf{l}_2)$.
    \State{Refinement iteration count $k \xleftarrow{} 0$ and $b_0 \xleftarrow{} d_m(\mathbf{l}_1,\mathbf{l}_2)$.}
    \State{Set the subdivided curve sets $\mathcal{L}_{1}=\{\mathbf{l}_{1}\},\mathcal{L}_{2}=\{\mathbf{l}_{2}\}$.}
     \While{$b_k>\epsilon$}
     \State{$k=k+1$.}
     \State{Bisect $\mathbf{l}_1$ at $\mathbf{x}^k = \mathbf{l}_1(0.5)$, replace it in $\mathcal{L}_1$ with the subcurves, and reindex the updated set as $\{\mathbf{l}_1^i\}_{i=0}^{k}$.}
     \State{Find the closest point $\mathbf{y}^k$ on $\mathbf{l}_2$ to $\mathbf{x}^k$, split $\mathbf{l}_2$ at $\mathbf{y}^k$, and update $\mathcal{L}_2$ with the resulting segments, reindexed as $\{\mathbf{l}_2^i\}_{i=0}^{k}$.}
     \State{Set $b_k=\min\{b_{k-1},\max\limits_{i\in\{0,\dots,k\}}\{d_{m}(\mathbf{l}_{1}^{i},\mathbf{l}_{2}^{i})\}\}$, and find $j$ s.t. $d_{m}(\mathbf{l}_{1}^{j},\mathbf{l}_{2}^{j})=b_{k}$.}
     \State{$\mathbf{l}_1\xleftarrow{}\mathbf{l}_1^{j}$,$\mathbf{l}_2\xleftarrow{}\mathbf{l}_2^{j}$.}
     \If{$k>k_{0}$}
     \Return $b_k$.
     \EndIf 
     \EndWhile \\
     \Return $b_k$.
    \end{algorithmic}
\end{algorithm}

\begin{lemma}\label{lemma 2} The Hausdorff distance between $ \mathbf{l}_1 $ and $ \mathbf{l}_2 $ satisfies  
\begin{equation}   d_H(\mathbf{l}_1, \mathbf{l}_2) \leq b_k,  
\end{equation}
where $b_k$ is obtained using Algorithm~\ref{alg:Framwork}, and the bound $b_k$ is monotonically decreasing.

% Moreover, the bound becomes tighter as $k$ increases.
\end{lemma}
%: $b_k \to {\color{blue}d_t}(\mathbf{l}_1, \mathbf{l}_2)$ as $k \to \infty$

Swapping the roles of $\mathbf{l}_1$ and $\mathbf{l}_2$ in Algorithm~\ref{alg:Framwork} yields another bound $b'_k$, i.e., $d_H(\mathbf{l}_1, \mathbf{l}_2) \leq b'_k$. Setting $\hat{b}_k(\mathbf{l}_1, \mathbf{l}_2) = \min\{b_k, b'_k\}$ provides a tighter bound:

\begin{theorem}\label{hausdorff geometric error theorem}
The Hausdorff distance between $ \mathbf{l}_1 $ and $ \mathbf{l}_2 $ satisfies 
\begin{equation}  
d_H(\mathbf{l}_1, \mathbf{l}_2) \leq \hat{b}_k(\mathbf{l}_1, \mathbf{l}_2).
\end{equation}
\end{theorem}

Given a piecewise curve $\mathcal{A} = \{\mathbf{l}_0, \dots, \mathbf{l}_j\}$ and a single curve $\mathbf{l}$, where $\mathbf{l}_i(1) = \mathbf{l}_{i+1}(0)$ for all $i$, we locate on $\mathbf{l}$ the closest point to each endpoint of $\mathcal{A}$, partitioning $\mathbf{l}$ into segments $\mathcal{B} = \{\mathbf{l}^*_0, \dots, \mathbf{l}^*_j\}$. 

\begin{theorem}\label{hausdorff geometric error theorem for multicurve}
The Hausdorff distance between $ \mathcal{A} $ and $\mathbf{l}$ satisfies 
\begin{equation}  
d_H(\mathcal{A}, \mathbf{l}) \leq \max\limits_i \, \hat{b}_k(\mathbf{l}_i, \mathbf{l}^*_i).
\end{equation}
\end{theorem}

\textbf{Validity condition for high-order quadrilateral.} Let $\mathbf{Q}_{ij}$, $0 \leq i,j \leq n$, be the control points of a degree-$n$ B\'ezier quadrilateral $\psi$, and define 0- and 1-vectors as $\Delta_{ij}^{0} = \mathbf{Q}_{(i+1)j} - \mathbf{Q}_{ij}$, $\Delta_{ij}^{1} = \mathbf{Q}_{i(j+1)} - \mathbf{Q}_{ij}$, $0 \leq i,j \leq n-1$, collectively the control vectors of $\psi$. Let $r_k$ the sector spanned by $\{\Delta_{ij}^k \mid 0 \leq i,j \leq n-1\}$, $k = 0,1$, see Fig.~\ref{high_order bezier element valid condition}. By the partition of unity and non-negativity of Bernstein bases, these sectors determine element validity, as formalized in the following theorem.

\begin{figure}
  \centering
\includegraphics[width=.4\textwidth]{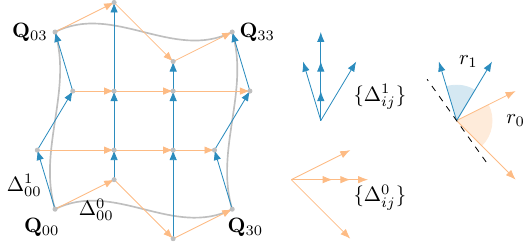}

  \caption{\label{high_order bezier element valid condition}
Illustration of vector sets and their corresponding sectors. The sectors spanned by the orange vectors $\{\Delta_{ij}^{0}|0 \leq i,j \leq n\}$ and the blue vectors $\{\Delta_{ij}^{1}|0 \leq i,j \leq n\}$ are shown as orange and blue convex cones, respectively, on the right.}
\end{figure} 

\begin{theorem}\label{valid high-order quad theorem}
If $r_0 \cap r_1 = \emptyset$ and a line exists such that $r_0$ and $r_1$ lie strictly on the same side (see Fig.~\ref{high_order bezier element valid condition}(right)), then the high-order quadrilateral element is valid, i.e., $\psi$ is injective and its Jacobian satisfies $\det(J) > 0$.
\end{theorem}
Theorem~\ref{valid high-order quad theorem} provides a sufficient condition for generating valid high-order quadrilateral elements, which will be used to guide the curve refinement process for producing high-quality meshes in the next section.

\section{Method}
\subsection{Algorithm overview}

Given $m$ regions $\Omega_1, \dots, \Omega_m$ enclosed by degree-$n$ piecewise Lagrange polynomial curves, denoted collectively as $\mathcal{C}_{\text{ori}}$  (Fig.~\ref{input and output}(a)), any two regions satisfy  
\begin{equation}  
\text{Area}(\Omega_i \cap \Omega_j) = 0, 
\end{equation}  
implying they may intersect only at points, along curves (interface), or remain disjoint. Without loss of generality, we assume throughout that the piecewise Lagrange curve enclosing each region $\Omega_i$ is non-self-intersecting; otherwise, it is split at self-intersection points into self-intersection free curves. Our algorithm takes as input the regions $\Omega_1, \dots, \Omega_m$, a target element size $l_t$, and a geometric error tolerance $\epsilon$, and outputs a high-quality degree-$n$ quadrilateral mesh $\mathcal{M}$ that preserves geometric fidelity, ensures smooth transitions, and maintains maximum element sizes close to $l_t$. Note that enforcing zero geometric deviation of the input curves would require excessive refinement, leading to overly fine meshes. To address this, we allow a geometric error bounded by $\epsilon$ between the mesh and the boundary, thereby balancing accuracy and mesh density. When two regions share an interface (i.e., their intersection is a curve), the mesh must preserve conformity: elements on both sides of the interface must align without hanging nodes (Fig.~\ref{input and output}(b)). 

\begin{figure}
  \centering
\includegraphics[width=.353\textwidth]{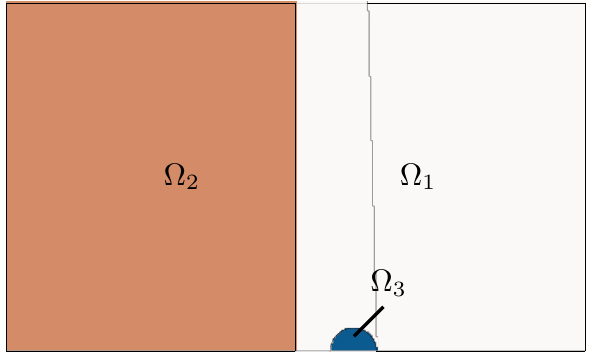}
\includegraphics[width=.353\textwidth]{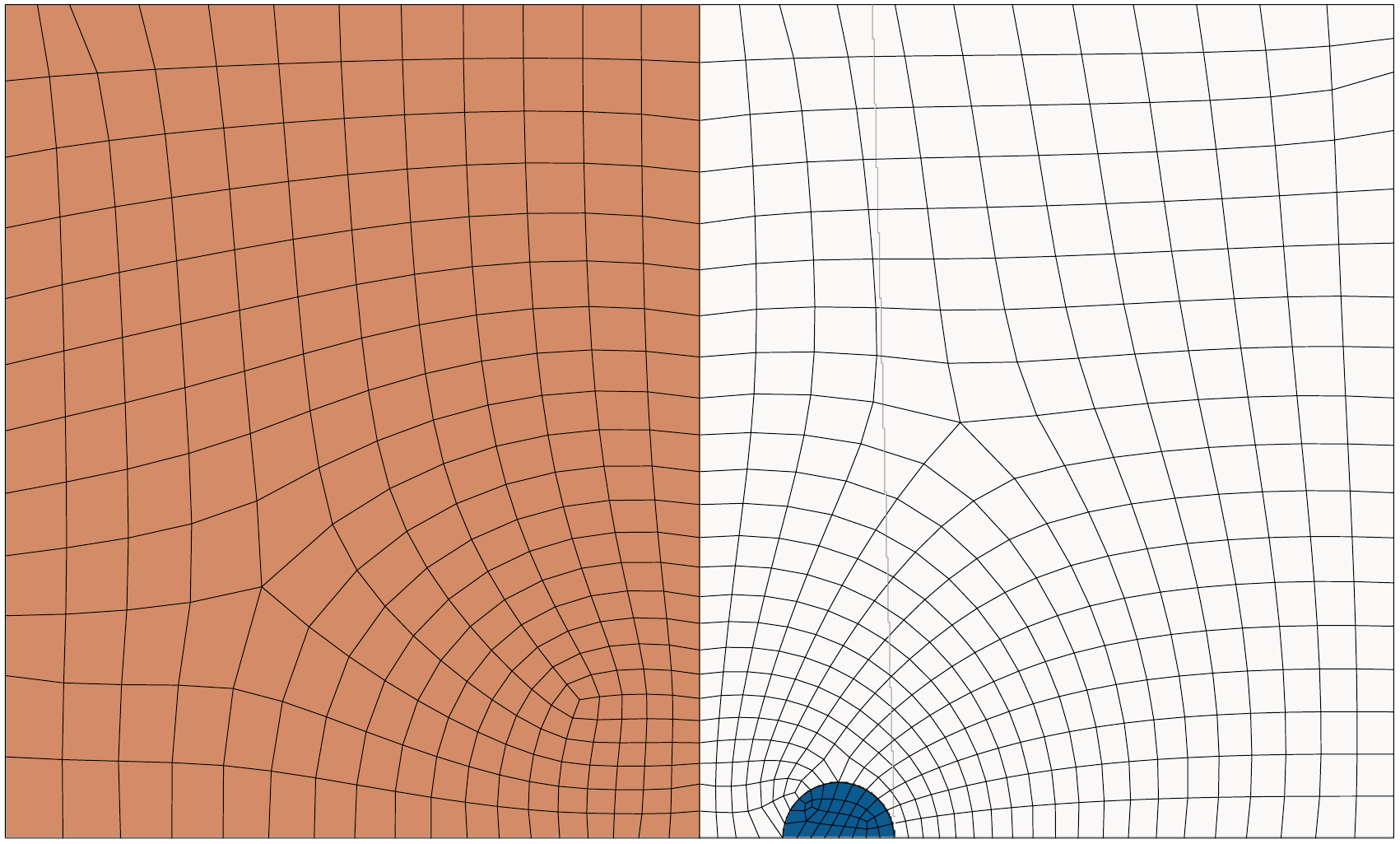}
     \parbox[t]{.35\textwidth}{\centering
            (a) 
           }
  \parbox[t]{.35\textwidth}{\centering
          (b)
           }
           
  \caption{\label{input and output}
           Illustration of input and output. (a) An input region showing interfaces between subregions, and (b) output mesh with consistent interfaces.}
\end{figure}

{\color{black} Conventional ``generate-then-optimize'' strategies heavily relies on mesh initialization, often failing to produce high-quality high-order meshes when the initial mesh mismatches the target geometry \cite{dobrev2019target}, as illustrated in Fig.~\ref{remeshing_using_our}. Furthermore, this approach suffers from a high and fixed DOF during optimization, hindering efficient generation of high-quality meshes. To tackle these challenges, we propose a curve reconstruction-based method for generating high-order quadrilateral meshes. To avoid excessive meshing, we allow the reconstructed boundary to approximate the input curves within a prescribed geometric tolerance, achieving reasonable mesh density (Fig.~\ref{para epsilon}) while ensuring well-distributed segmentation and smooth transitions for high mesh quality. Our method reduces mesh optimization to a geometric error-bounded curve reconstruction problem, lowering complexity and enabling more efficient and effective generation of high-quality high-order meshes.} 

As illustrated in Fig.~\ref{pipeline}, {\color{black} our method first performs optimization-driven, geometric error-bounded reconstruction of closed high-order Lagrangian curves.} High-quality straight-sided triangular meshes are then generated based on the \textcolor{black}{linear segments connecting the approximated curve endpoints}. A modified blossom-quad algorithm merges these triangles into a linear quadrilateral mesh that preserves interfaces and avoids degenerate elements. Finally, the linear mesh is elevated to high-order, with boundary and interface elements deformed via a mean value coordinates strategy to produce the final high-quality high-order quadrilateral mesh. In the rest of this section, we describe in detail the three main components of our algorithm: geometric error-bounded curve reconstruction, linear quadrilateral meshing, and high-order quadrilateral meshing. 

\begin{figure}
  \centering
    \includegraphics[width=.23\textwidth]{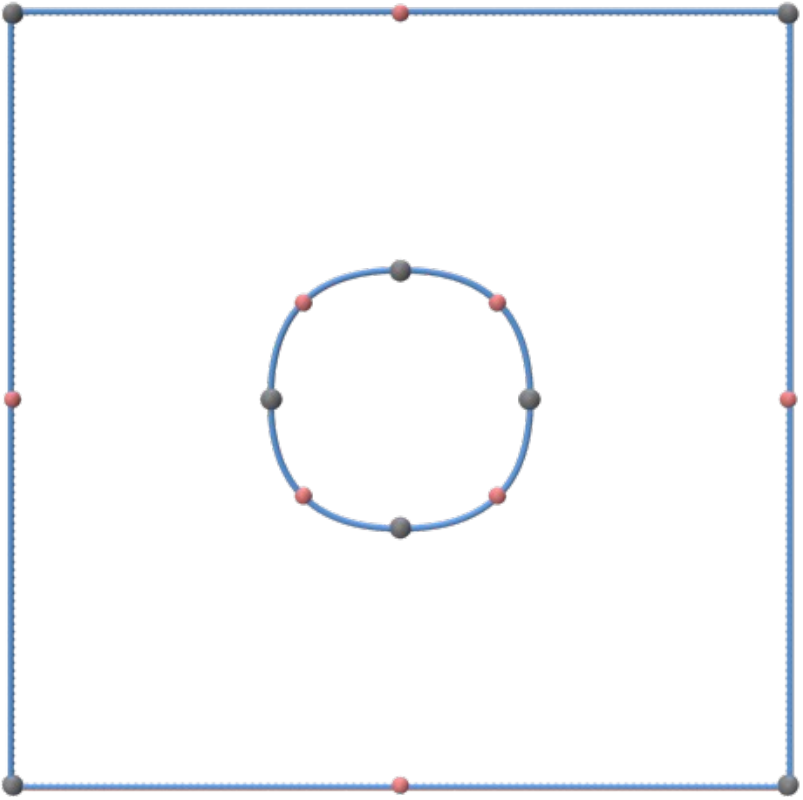}
  \includegraphics[width=.23\textwidth]{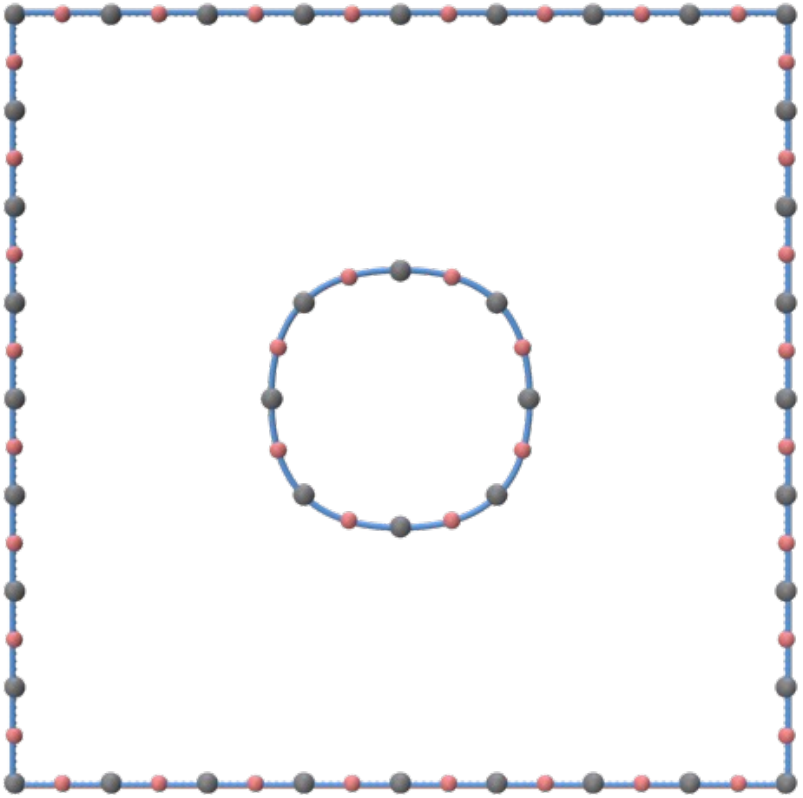}
  \includegraphics[width=.23\textwidth]{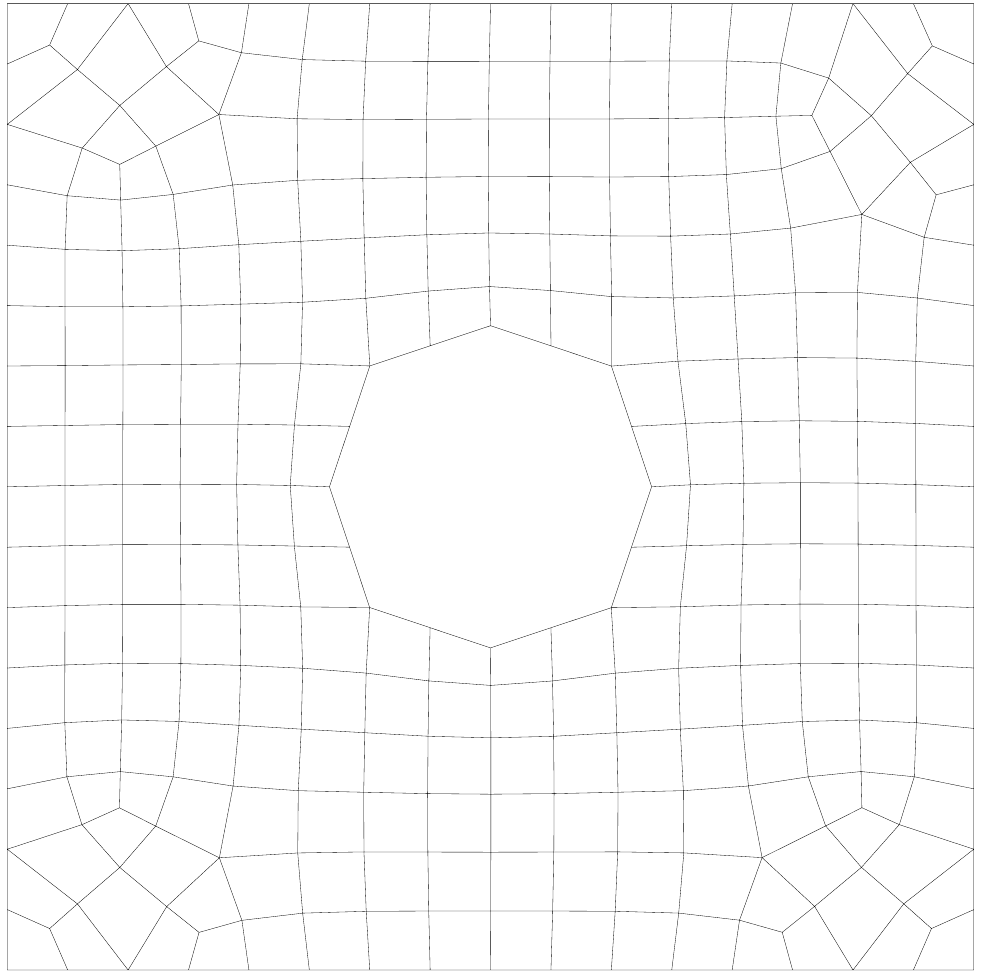}
    \includegraphics[width=.23\textwidth]{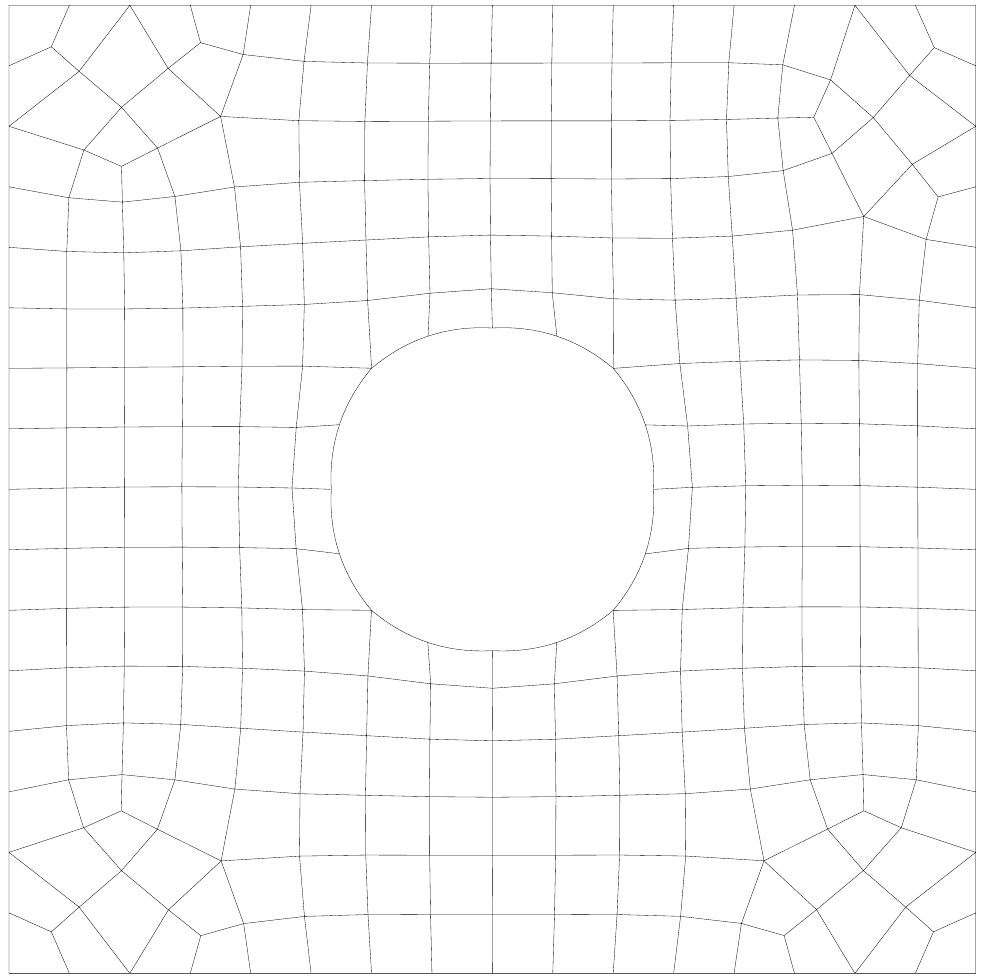}
   \parbox[t]{.23\textwidth}{\centering
            (a)
           }
  \parbox[t]{.23\textwidth}{\centering
           (b)
           }
  \parbox[t]{.23\textwidth}{\centering
           (c)
           }
  \parbox[t]{.23\textwidth}{\centering
           (d)
           }
  \caption{\label{pipeline}
           Algorithm pipeline. (a) Input curve; (b) curve reconstruction; (c) linear quadrilateral mesh; (d) high-order quadrilateral mesh. The gray points in (a) and (b) are the endpoints of the input curve, and the red points are the high-order control points.}
\end{figure}

\subsection{Geometric error-bounded curve reconstruction}\label{section_resampling}
In this section, we develop a geometric error-bounded curve reconstruction method for high-order mesh generation. It comprises three stages: (1) error-bounded curve approximation to keep the geometric deviation between the final curved mesh and the input boundaries and interfaces within a threshold $\epsilon$, (2) adaptive refinement to satisfy segment size, angular, and feature constraints, and (3) optimizing the distribution of endpoints of reconstructed curve segments introduced during adaptive refinement to enhance uniformity and smoothness of the resulting curved meshes.

\textbf{Error-bounded geometric approximation.}
First, we split $\mathcal{C}_{\text{ori}}$ at corner points with discontinuous tangents to obtain a refined segmentation $\mathcal{C}_{\text{cor}} = \{\mathbf{c}_i\}_{i=0}^j$, where each $\mathbf{c}_i$ is a degree-$n$ Lagrange curve satisfying $\mathbf{c}_i(1) = \mathbf{c}_{i+1}(0)$. $\mathcal{C}_{\text{cor}}$ and $\mathcal{C}_{\text{ori}}$ represent the same geometry, differing only in segmentation. For two points $\mathbf{P}_0,\mathbf{P}_1$ on $\mathcal{C}_{\text{ori}}$, $\mathbf{P}_0\mathbf{P}_1|_{\mathcal{C}_{\text{ori}}}$ denotes the arc of $\mathcal{C}_{\text{ori}}$ connecting them. 
If $\mathbf{c}_i(1)$ is not a corner, we construct a curve $\mathbb{H}(\mathbf{c}_i(0), \mathbf{c}_{i+1}(1))$ that approximates the arc $\mathbf{c}_i(0)\mathbf{c}_{i+1}(1)|_{\mathcal{C}_{\text{ori}}}$ by interpolating $\mathbf{c}_i(0)$ and $\mathbf{c}_{i+1}(1)$ as endpoints and uniformly redistributing the high-order interpolation points along this arc.
For each consecutive pair $\mathbf{c}_i, \mathbf{c}_{i+1} \in \mathcal{C}_{\text{cor}}$, the Hausdorff distance upper bound $\hat{b}_k(\mathbb{H}(\mathbf{c}_i(0),\mathbf{c}_{i+1}(1)), \mathbf{c}_i(0)\mathbf{c}_{i+1}(1)|_{\mathcal{C}_{\text{ori}}})$ is computed, with a maximum refinement iteration $k_0 = 20$ in Algorithm~\ref{alg:Framwork}. If $\hat{b}_k(\mathbb{H}(\mathbf{c}_i(0),\mathbf{c}_{i+1}(1)), \mathbf{c}_i(0)\mathbf{c}_{i+1}(1)|_{\mathcal{C}_{\text{ori}}}) < \epsilon$, $\mathbf{c}_i$ and $\mathbf{c}_{i+1}$ are replaced in $\mathcal{C}_{\text{cor}}$ by $\mathbb{H}(\mathbf{c}_i(0),\mathbf{c}_{i+1}(1))$. This process repeats until no further replacements occur.
 % This operation is applied throughout the subsequent approximation, refinement, and optimization steps to ensure high-quality meshes.
\begin{figure}
  \centering
\includegraphics[width=.2\textwidth]{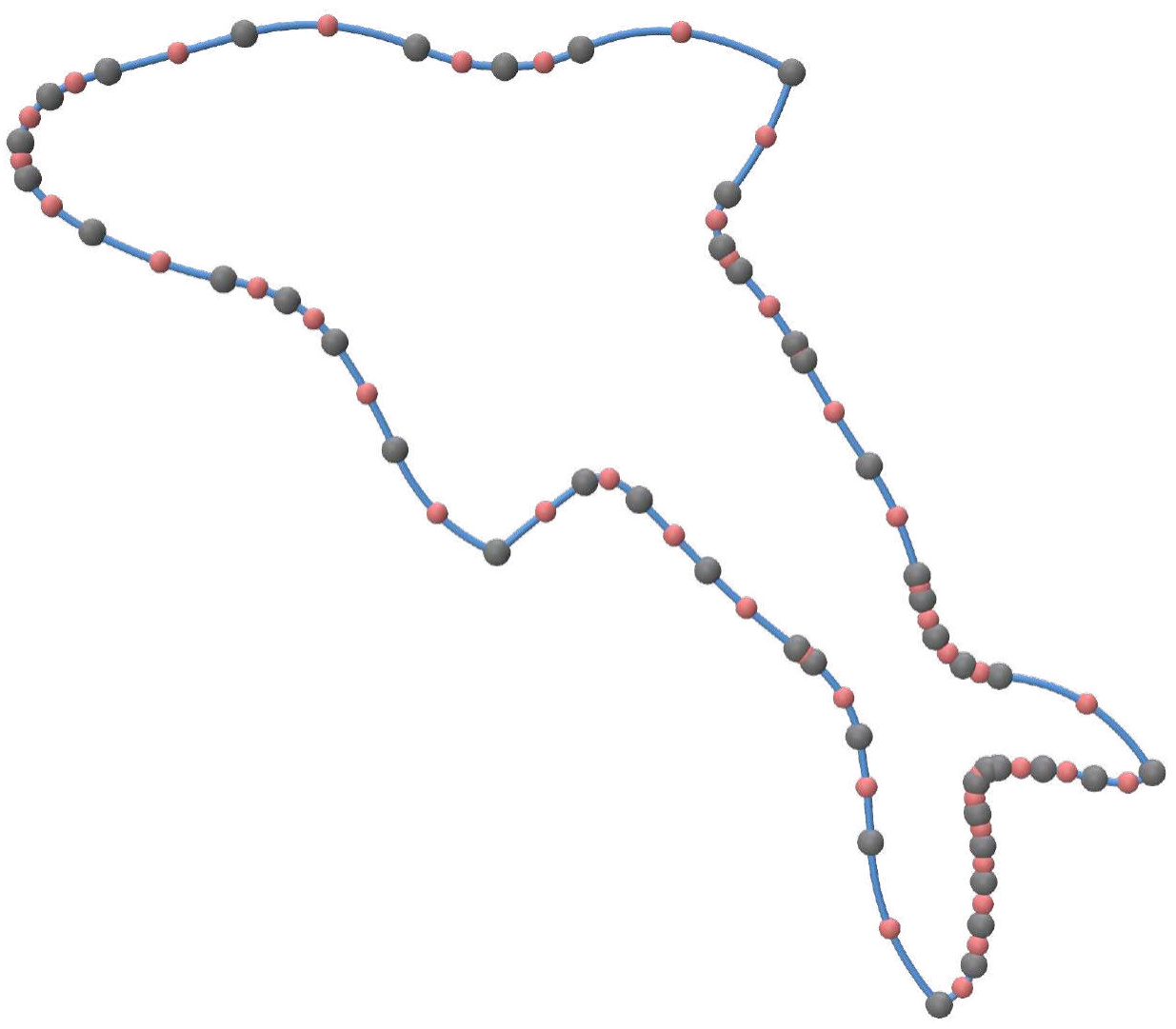}
\includegraphics[width=.2\textwidth]{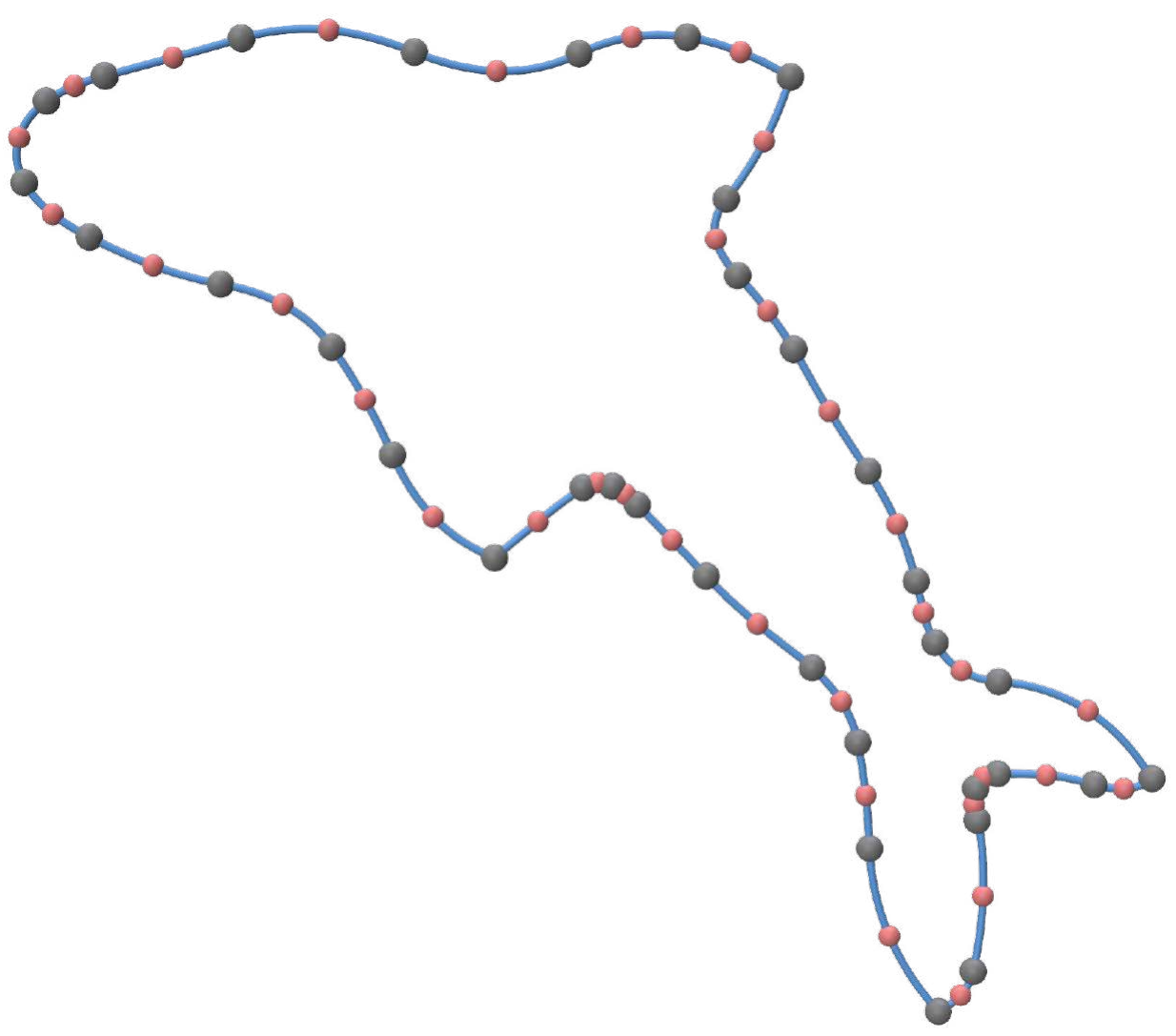} 
\includegraphics[width=.2\textwidth]{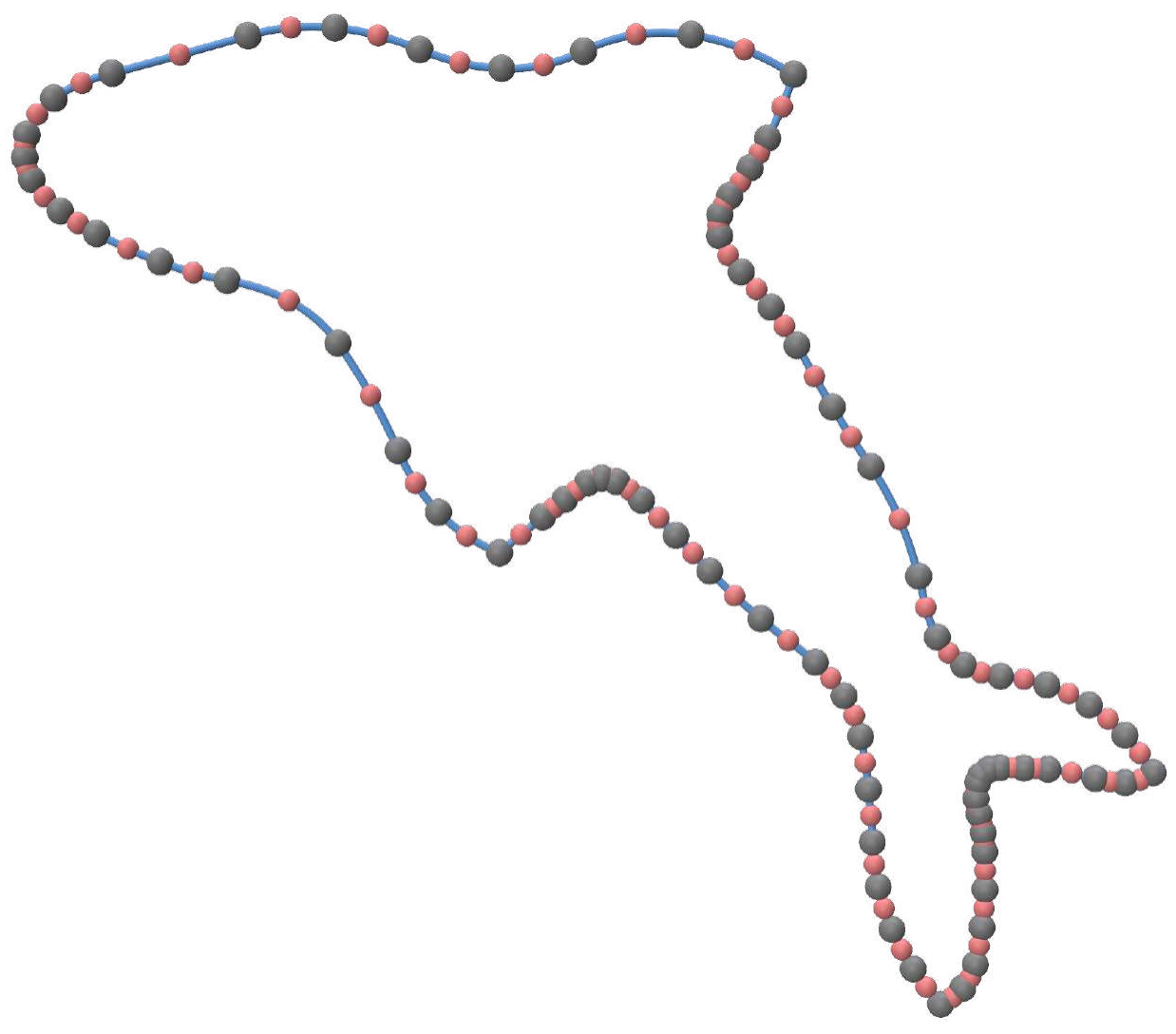}  
\includegraphics[width=.3\textwidth]{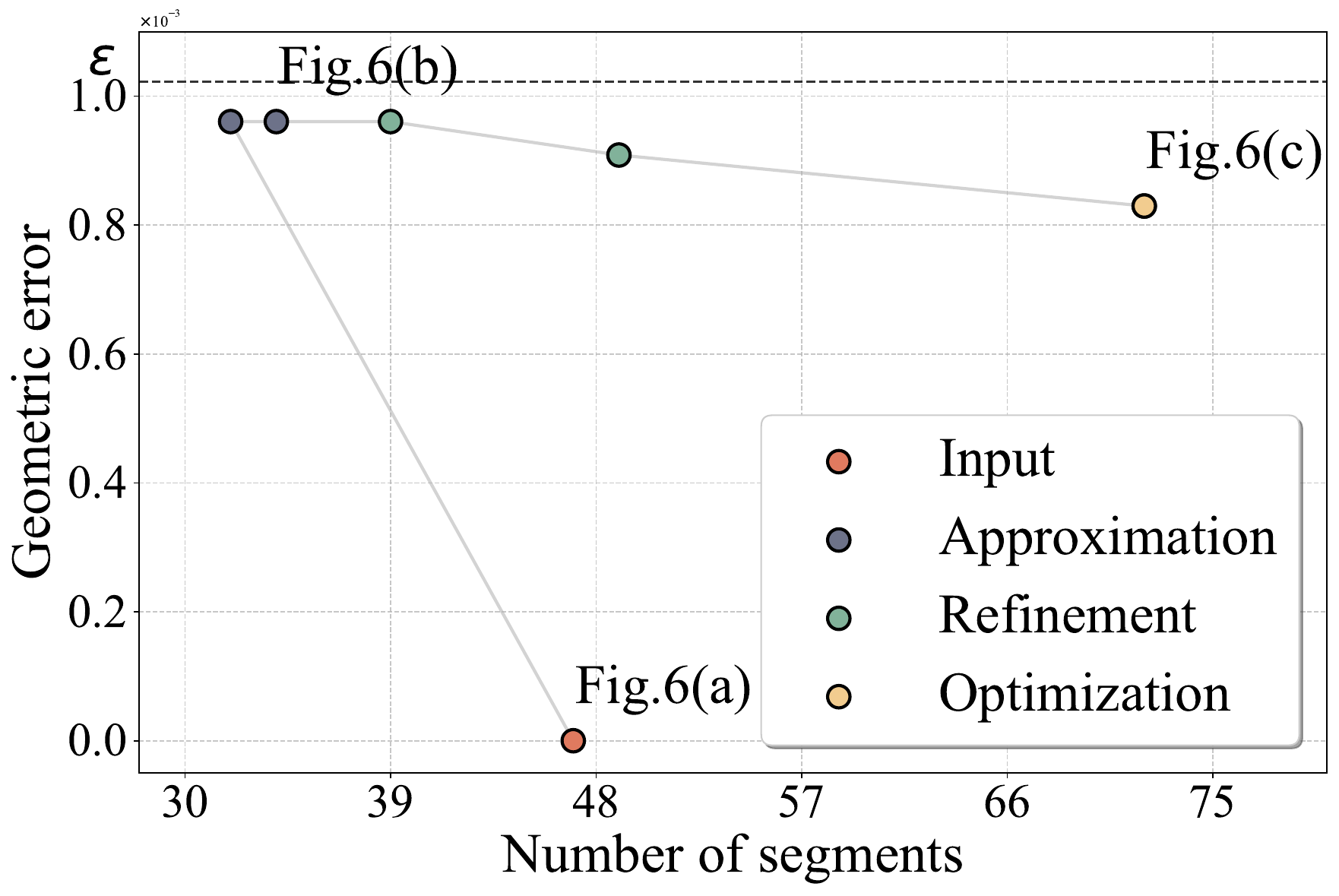}

  \parbox[t]{.2\textwidth}{\centering
            (a) 
           }
  \parbox[t]{.2\textwidth}{\centering
           (b) 
           }
    \parbox[t]{.2\textwidth}{\centering
       (c)            
       }
         \parbox[t]{.3\textwidth}{\centering
       (d)            
       }
           
  \caption{\label{resampling based geometric error}
  \textcolor{black}{Curve approximation with geometric error control and adaptive refinement. (a) Input curves (47 segments), (b) bounded-error approximation result (34 segments), (c) curve reconstruction after adaptive refinement (72 segments), and (d) approximation error versus segment number across three stages for the input curve in (a), showing bounded error, where the geometric tolerance is set to $0.001$ times the bounding-box diagonal.}
}
\end{figure}

For each $\mathbf{c} \in \mathcal{C}_{\text{cor}}$, if the geometric error $$\hat{b}_{k}(\mathbf{c}(0)\mathbf{c}(1)|_{\mathcal{C}_{\text{ori}}}, \mathbb{H}(\mathbf{c}(0),\mathbf{c}(1))) < \epsilon,$$ replace $\mathbf{c}$ with $\mathbb{H}(\mathbf{c}(0),\mathbf{c}(1))$; Otherwise, $\mathbf{c}$ is bisected into two segments, $\mathbb{H}(\mathbf{c}(0), \mathbf{P})$ and $\mathbb{H}(\mathbf{P}, \mathbf{c}(1))$, where $\mathbf{P}$ is the arc-length midpoint of the original arc $\mathbf{c}(0)\mathbf{c}(1)|_{\mathcal{C}_{\text{ori}}}$. This subdivision is applied recursively: we continue bisecting the segment $\mathbb{H}(\mathbf{c}(0),\mathbf{c}(1))$ until the approximation error satisfies the threshold, yielding a segment set $\mathcal{A}$ that replaces $\mathbf{c}$ in $\mathcal{C}_{\text{cor}}$. Successive bisection provides an asymptotically convergent approximation of $\mathcal{C}_{\text{ori}}$, ensuring the error eventually falls below $\epsilon$ and guaranteeing termination. As shown in Fig.~\ref{resampling based geometric error}(b), this bounded-error approximation reduces the number of short edges and irregular segments while respecting the prescribed geometric error tolerance.
%If the aforementioned operation is applied to the interface curve, the result from its first traversal is taken.

\textbf{Adaptive refinement.} 
 Let $\theta(\mathbf{c})$ be the angle of the sector spanned by the control vectors of $\mathbf{c} \in \mathcal{C}_{\text{cor}}$. By Theorem~\ref{valid high-order quad theorem}, larger $\theta(\mathbf{c})$ impedes valid high-order element construction; hence, we enforce an angular threshold $\tau = \frac{\pi}{4}$ to drive adaptive refinement and ensure inversion-free, high-quality meshes. For each $\mathbf{c} \in \mathcal{C}_{\text{cor}}$, we recursively bisect the curve until $\theta(\mathbf{c}) \leq \tau$ and its arc length is less than $2l_t$. This process is guaranteed to converge, as repeated refinement reduces the control vector angles asymptotically to zero. 
 % For high-order quadrilateral elements with boundary (interface) edges, we define the control vectors of boundary (interface) belong to the element’s 0-vector, i.e., $r \subseteq r_0$. By Theorem \ref{valid high-order quad theorem}, larger angular spans of $r$ hinder the construction of valid high-order elements.
 
% \begin{figure}[htb]
%   \centering
% \includegraphics[width=.45\textwidth]{figures/Figure6_merge.pdf}
% \includegraphics[width=.45\textwidth]{figures/figure7_resampling.pdf}     

%   \parbox[t]{.45\textwidth}{\centering
%             (a) 
%            }
%   \parbox[t]{.45\textwidth}{\centering
%            (b) 
%            }
           
%   \caption{\label{adaptive resampling}
%            Illustration of adaptive refinement. (a) Result of bounded-error curve approximation, and (b) curve reconstruction after adaptive refinement.}
% \end{figure}

\begin{figure}
  \centering
\includegraphics[width=.24\textwidth]{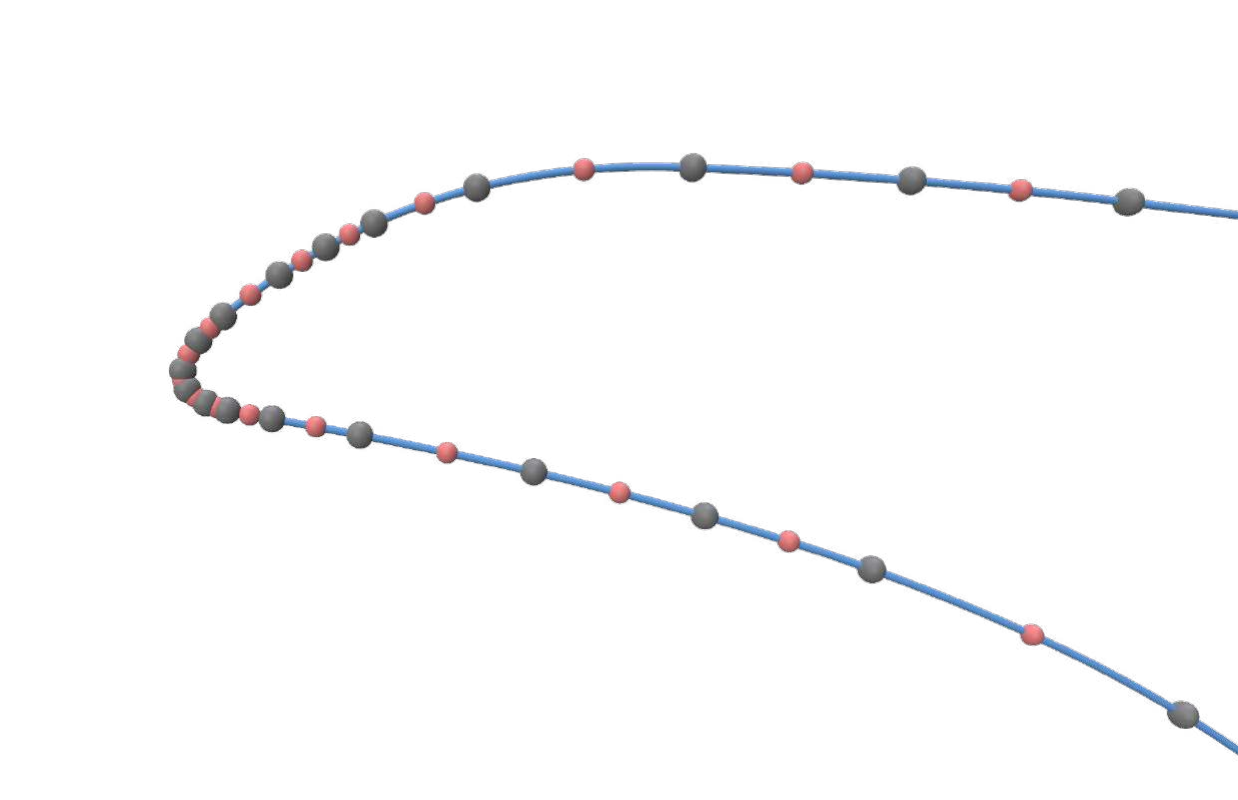}
\includegraphics[width=.24\textwidth]{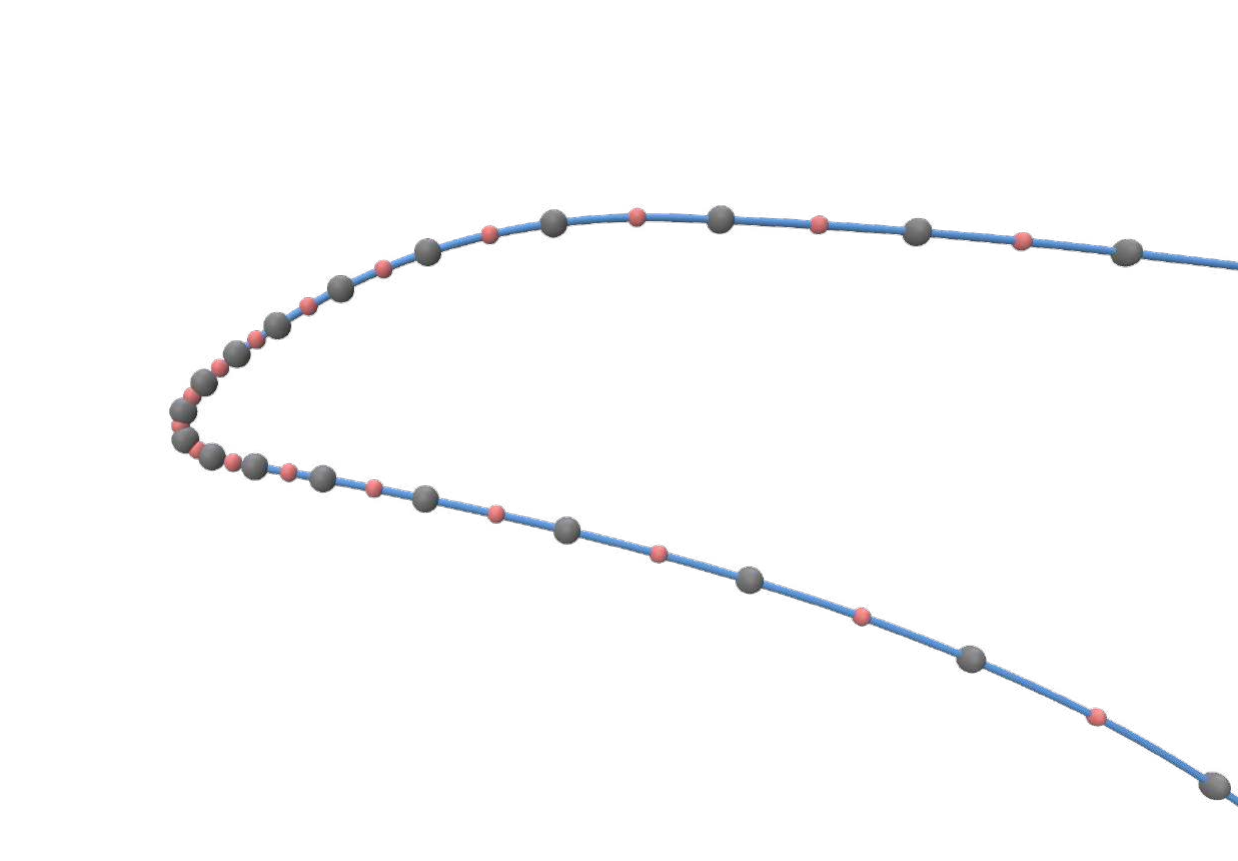}
\includegraphics[width=.22\textwidth]{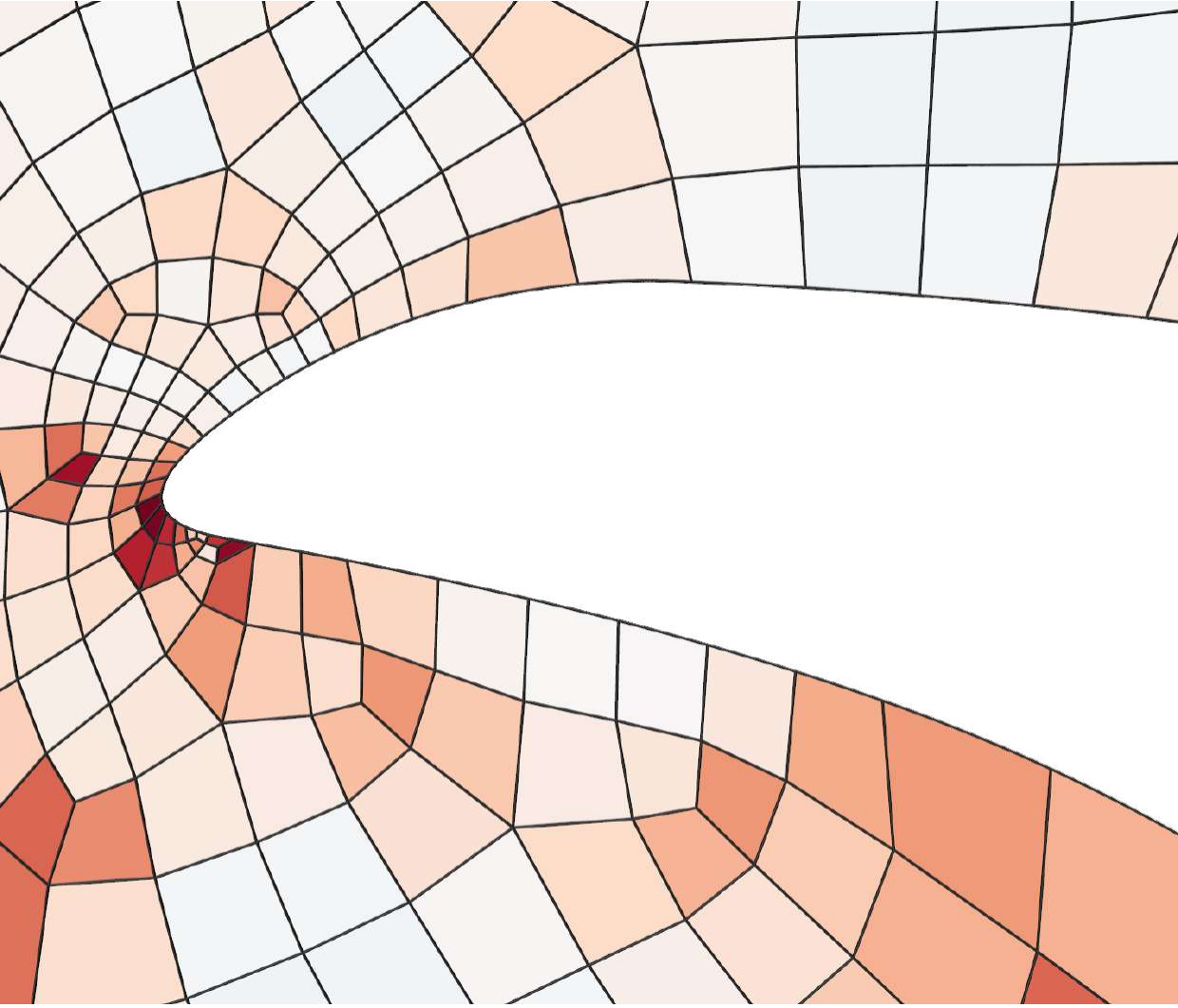}
\includegraphics[width=.258\textwidth]{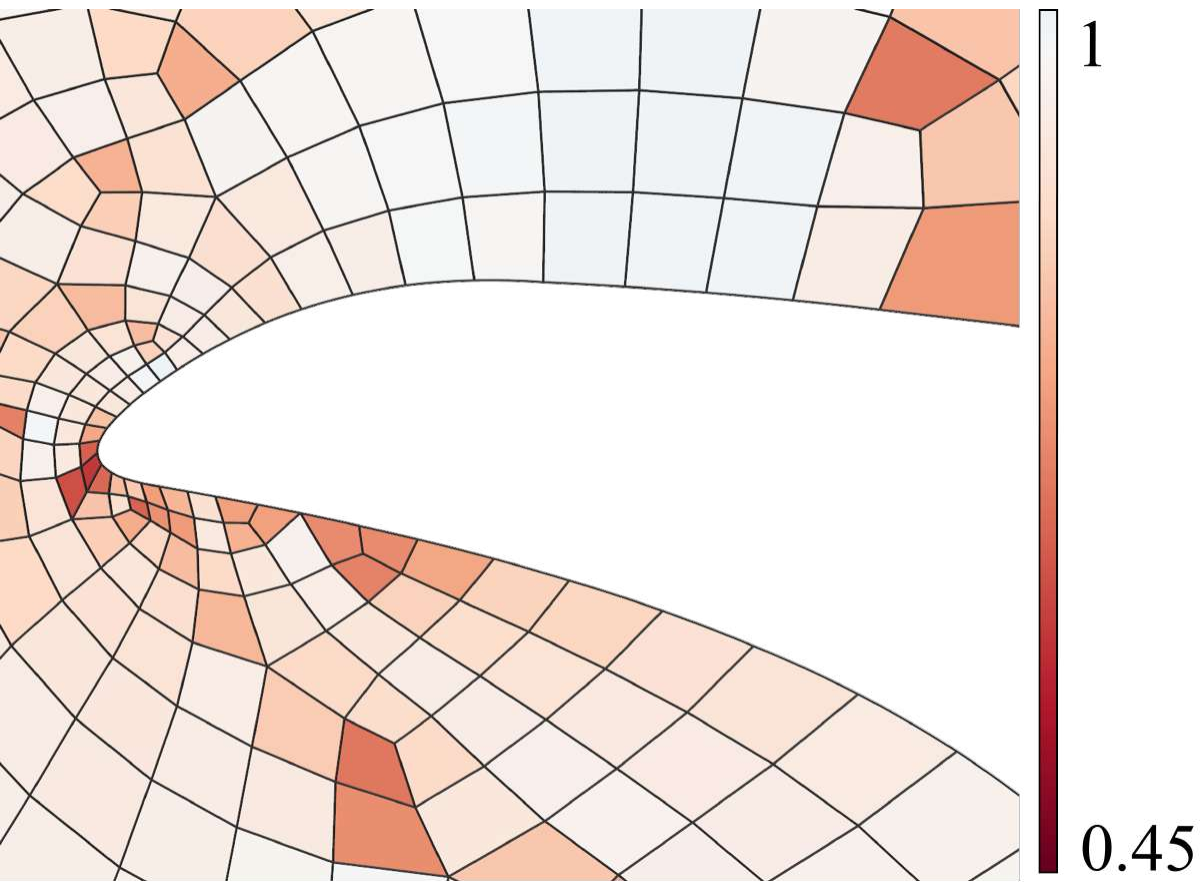}

     \parbox[t]{.24\textwidth}{\centering
            (a) 
           }
  \parbox[t]{.24\textwidth}{\centering
           (b) 
           }
     \parbox[t]{.24\textwidth}{\centering
            (c) 
           }
  \parbox[t]{.24\textwidth}{\centering
           (d) 
           }
           
  \caption{\label{opt resampling pipeline}
           Effect of optimization on curve reconstruction and mesh quality. (a,b) Reconstructed curves without and with optimization, respectively. (c,d) Corresponding meshes, where the color of each element represents its quality, indicated by the minimum shape measure $J_{m}$ (see Section~\ref{quality metrics}).}
\end{figure}

Refinement based solely on the previous criteria may produce large length disparities between adjacent curve segments and fails to adapt to geometric features such as narrow regions, often leading to poor-quality elements after high-order lifting. To address this limitation, we perform an additional adaptive refinement on each curve $\mathbf{c} \in \mathcal{C}_{\text{cor}}$. 
Let $\mathbf{c}_{k}$ be the nearest curve segment to $\mathbf{c}$ among $\mathcal{C}_{\text{cor}} \setminus \{ {\mathbf{c} \cup \mathbf{c}_{e}}\}$, where $\mathbf{c}_{e}$ denotes the set of segments adjacent to $\mathbf{c}$. Denote the arc length of $\mathbf{c}$ by $\|\mathbf{c}\|$. The parameter $\alpha$ controls segmentation density and mesh resolution, and is empirically set to $2.3$. Without loss of generality, assume $\|\mathbf{c}\| < \|\mathbf{c}_{k}\|$. Then $\mathbf{c}$ is bisected whenever $\frac{\|\mathbf{c}\|}{z(\mathbf{c})} > \alpha$, where $z(\mathbf{c})$ is defined as:
\begin{equation}
z(\mathbf{c}) =
\begin{cases} 
\alpha^{\gamma}\|\mathbf{c}_{k}\|, & d_{H}(\mathbf{c},\mathbf{c}_{k}) \geq \|\mathbf{c}_{k}\| \\
(1-\omega)\frac{\|\mathbf{c}\| + d_{H}(\mathbf{c},\mathbf{c}_{k})}{2} + \omega\frac{\|\mathbf{c}_{k}\| + \|\mathbf{c}\|}{2}, & \|\mathbf{c}\| < d_{H}(\mathbf{c},\mathbf{c}_{k}) < \|\mathbf{c}_{k}\| \\
d_{H}(\mathbf{c},\mathbf{c}_{k}), & d_{H}(\mathbf{c},\mathbf{c}_{k}) \leq \|\mathbf{c}\|,
\end{cases}
\end{equation}
with $\gamma$ being the largest integer satisfying $\dfrac{d_{H}(\mathbf{c},\mathbf{c}_{k})}{\|\mathbf{c}_{k}\|} > \sum_{k=0}^{\gamma} \alpha^{k}$, and weight $\omega = \frac{d_{H}(\mathbf{c},\mathbf{c}_{k}) - \|\mathbf{c}\|}{\|\mathbf{c}_{k}\| - \|\mathbf{c}\|}$. 

This adaptive refinement produces reconstructed curves whose segment lengths conform to geometric features while maintaining limited variation between neighboring segments, as illustrated in Fig.~\ref{resampling based geometric error}(c), thereby establishing a foundation for high-quality high-order quadrilateral mesh generation. Nevertheless, the refined curves may still exhibit non-uniform point transitions, as shown in Fig.~\ref{opt resampling pipeline}(a). 

\textbf{Endpoints distribution optimization.} To improve the uniformity of point distribution along the reconstructed curves, we employ a curve optimization step based on Centroidal Voronoi Tessellation (CVT) energy \cite{liu2009centroidal}, a standard tool in mesh generation for promoting uniform point placement. In our setting, Voronoi cells are defined along the curve, seed points are placed at the endpoints of the curve-segment set $\mathcal{C}_{\text{cor}}$, and arc length is used in place of Euclidean distance. To obtain a smoother transition in point distribution, we apply Lloyd's algorithm \cite{lloyd1982least}. At each iteration, the generator $\mathbf{P}_{i}$ (except those fixed as endpoints introduced during bounded-error approximation) is updated to the curvature-weighted centroid of its Voronoi cell, after which the optimized curve is reconstructed as $\mathbb{H}(\mathbf{P}_{i},\mathbf{P}_{i+1})$. The arc-length coordinate $s_{\rho}$ of the curvature-weighted centroid of a Voronoi cell $V$ is defined as:
\begin{equation}
    s_{\rho}=\frac{\int_V s\rho(s)ds}{\int_V \rho(s)ds},
\end{equation}
 where $s$ is the arc length, $\kappa(s)$ is the curvature of the curve segment at $s$, and $\rho(s)=1+\kappa(s)$. We perform a fixed number of $5$ Lloyd iterations to update the reconstructed curve set $\mathcal{C}_{\text{cor}}$. As illustrated in Fig.~\ref{opt resampling pipeline}, this procedure yields increasingly uniform point distributions and smoother transitions, making it well suited for high-quality high-order quadrilateral mesh generation.

Fig.~\ref{resampling based geometric error}(d) plots the approximation error of the reconstructed curves for the input curve in Fig.~\ref{resampling based geometric error}(a) across three stages. After stage~(1), the error is bounded by the reference tolerance and is further reduced through subsequent refinement and final optimization.

%  \begin{figure}[htb]
%   \centering
% \includegraphics[width=0.8\textwidth]{figures/error_vs_segments.pdf}    
%   \caption{\label{error vs segments}Approximation error versus segment number across three stages for the input curve in Fig.~\ref{resampling based geometric error}(a), showing bounded error, where the geometric tolerance is set to $0.001$ times the bounding-box diagonal.}
% \end{figure}

\subsection{Linear quadrilateral meshing}\label{section_modified_Blossom-quad}
The discrete segments formed by connecting the endpoints of each reconstructed curve segment in Section~\ref{section_resampling} are used as linear constraints for generating high-quality triangular meshes that preserve these boundaries~\cite{remacle2013frontal}. We then introduce a customized blossom-quad (BQ) algorithm to merge the triangles into quadrilaterals. The algorithm constructs the dual graph of the triangular mesh and applies minimum-weight perfect matching using the following cost function:
\begin{equation}
cost=\sum_{e\in E'} (1-\beta(q_{ij})) + \zeta(q_{ij}),
\end{equation}
where
\begin{equation}
\zeta(q_{ij}) =
\begin{cases}
10000, & \text{if $q_{ij}$ has two boundary/interface edges}, \\
0, & \text{otherwise}.
\end{cases}
\end{equation}
Here, $E'$ is the set of edges to be merged, $q_{ij}$ is the quadrilateral formed by triangles $t_i$ and $t_j$, and $\beta(q)$ is a linear quadrilateral quality metric.

The original BQ algorithm \cite{remacle2012blossom} uses only the quality term $1-\beta(q_{ij})$ and aggressively favors well-shaped quadrilaterals. However, on meshes containing $C^{1}$ boundary or interface chains, it often merges triangles into quads with two boundary/interface edges, which invariably degenerate after high-order boundary correction and produce elements with near-zero Jacobians (Fig.~\ref{bad element}(a)). By introducing a soft penalty $\zeta(q_{ij})$, our modified BQ algorithm avoids such unsafe mergers, effectively preventing degenerate elements during high-order deformation. Because the added constraint produce a small number of leftover triangles, we perform midpoint subdivision to obtain a pure quadrilateral mesh. During curve reconstruction, the target length is set to $2l_t$, and after subdivision the target becomes $l_t$, with all new midpoints projected onto the reconstructed curves. As shown in Fig.~\ref{bad element}(b), the modified BQ algorithm eliminates the degenerate elements produced by the original method and yields a topologically robust quadrilateral layout suitable for subsequent high-order meshing.

% {\color{orange}The discrete segments obtained by connecting the endpoints of each segment of the reconstructed curve from section \ref{section_resampling} serve as linear constraints to generate high-quality triangular meshes that maintain these constraints \cite{remacle2013frontal}. This section presents a customized blossom-quad (BQ) algorithm tailored for high-order mesh generation. This algorithm converts a triangular mesh into its dual graph and applies minimum weight perfect matching to merge triangles into quadrilaterals. In this algorithm, the cost function is defined as:
% \begin{equation}
% cost=\sum_{e\in E^{'}}(1-\beta(q_{ij}))+\zeta(q_{ij}),
% \end{equation}
% where
% \begin{equation}
% \zeta(q_{ij}) =
% \begin{cases}
% 10000, & \text{if $q_{ij}$ has two edges on the boundary}, \\
% 0, & \text{otherwise},
% \end{cases}
% \end{equation}
% where $E^{'}$ represents the set of edges being merged, $q_{ij}$ denotes the quadrilateral formed by merging triangle $t_{i}$ and triangle $t_{j}$, and $\beta(q)$ stands for a linear quality measure of the quadrilateral. Quadrilateral elements with adjacent $C^1$-continuous boundary edges inevitably produce degenerate, near-zero-quality elements after deformation (Fig.~\ref{bad element}). Unlike the standard BQ algorithm \cite{remacle2012blossom}, our tailored BQ for high-order quadrilateral mesh generation prevents such degeneracies by incorporating a soft constraint $\zeta(q_{ij})$ into the cost function to avoid elements with two boundary edges.

\begin{figure}
  \centering
\includegraphics[width=.6\textwidth]{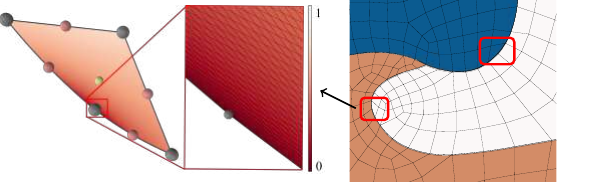} 
\includegraphics[width=.205\textwidth]{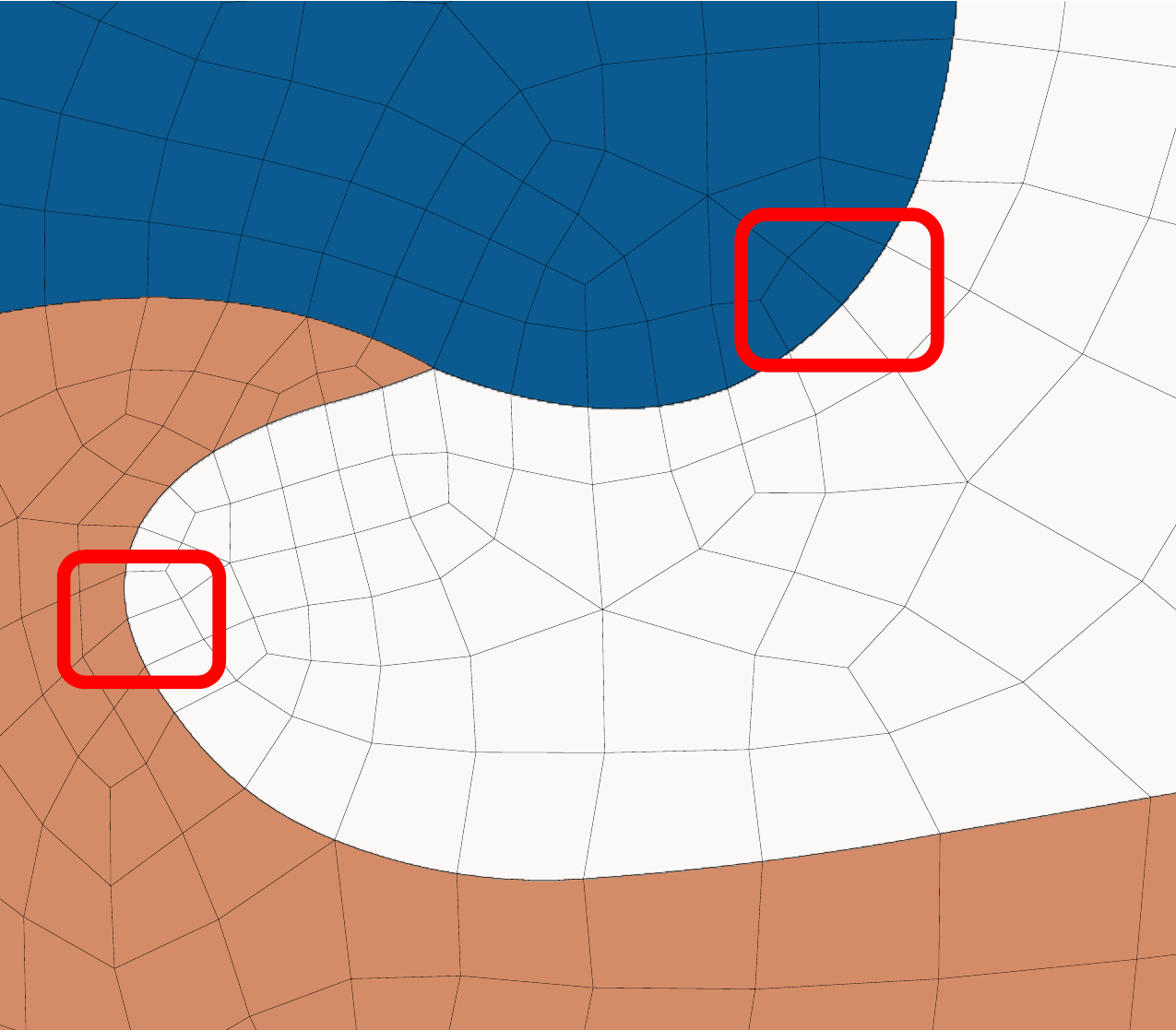}

     \parbox[t]{.6\textwidth}{\centering
            (a) 
           }
  \parbox[t]{.205\textwidth}{\centering
           (b) 
           }

  \caption{\label{bad element}
          High-order quadrilateral meshes ($n=2$) generated by the original (a) and modified (b) BQ methods. In (a), degenerate elements appear in regions highlighted by red boxes; a zoomed example is shown on the left, where color indicates $\det(J)$. In contrast, the modified BQ method prevents such elements during high-order deformation. Gray, red, and yellow points denote the endpoints of the input curve, high-order edge nodes, and high-order face nodes, respectively.}
\end{figure}

% Stricter constraints may yield a mixed triangle–quadrilateral mesh, so we apply midpoint subdivision to obtain a pure quadrilateral mesh. During curve reconstruction, the length threshold is $2l_t$; after subdivision, the upper bound of edge length becomes $l_t$, and new midpoints are projected onto the reconstructed curves. Fig.~\ref{modified blossom-quad} shows that our customized BQ method avoids degenerate elements, unlike the original BQ approach.
% }

% \begin{figure}
%   \centering
% \includegraphics[width=.2\textwidth]{figures/figure10_without_Modified_blossom.pdf}
% \includegraphics[width=.2\textwidth]{figures/figure10_Modified_blossom.pdf}

%      \parbox[t]{.2\textwidth}{\centering
%             (a) 
%            }
%   \parbox[t]{.2\textwidth}{\centering
%            (b) 
%            }
           
%   \caption{\label{modified blossom-quad}
%            High-order quadrilateral meshes ($n=2$) from the original (a) and tailored (b) BQ methods.}
% \end{figure}

\subsection{High-order quadrilateral meshing}
When upgrading the linear quadrilateral mesh generated in Section~\ref{section_modified_Blossom-quad} to higher orders, high-order nodes are first placed at degree-dependent equidistant positions on element edges or in the interior, as illustrated in Fig.~\ref{mvc}(a), ensuring consistency between the mesh degree and the input curve degree. 
Only elements whose edges are associated with boundary or interface curves require further deformation. 
For such boundary/interface elements, the high-order boundary interpolation points are moved to their corresponding positions on $\mathcal{C}_{\text{cor}}$, ensuring that the high-order mesh boundaries match $\mathcal{C}_{\text{cor}}$ exactly. 
However, adjusting boundary points alone may deteriorate element quality or even produce inverted elements, as shown in Fig.~\ref{mvc}(b), where shape measures below zero indicate inversion near the moved boundary point.

To prevent these degradations, we apply mean value coordinates (MVC) \cite{floater2003mean} to smoothly deform each affected high-order element. We first construct a polygon using the nodes $\{\mathbf{P}_i\}_{i=0}^{4n-1}$ on the four edges of the high-order quadrilateral element (Fig.~\ref{mvc}(a)), where $n$ denotes the degree of the input curve. 
The positions of interior node $\textbf{P}_{in}$ is defined using the MVC of this polygon:

\begin{equation}
\textbf{P}_{in}=\sum_{i=0}^{4n-1}\lambda_{i}\textbf{P}_{i},\sum_{i=0}^{4n-1}\lambda_{i}\equiv1,
\end{equation}
where,
\begin{equation}
\lambda_{i}=\frac{\tan(\frac{\angle \textbf{P}_{i-1}\textbf{P}_{in}\textbf{P}_{i}}{2})+\tan(\frac{\angle \textbf{P}_{i}\textbf{P}_{in}\textbf{P}_{i+1}}{2})}{||\textbf{P}_{in}-\textbf{P}_{i}||_{2}}.
\end{equation} Accordingly, when matching the boundary $\mathcal{C}_{\text{cor}}$, the interior nodes of high-order elements are updated using MVC, yielding high-quality high-order meshes. It is worth noting that a direct barycentric mapping may become non-injective, leading to element inversion. To avoid this issue, we employ a sequence of intermediate polygons so that the displacement between successive polygons remains small, ensuring a smooth and bijective composite mean value mapping~\cite{schneider2013bijective}. As shown in Fig.~\ref{mvc}(c), MVC-based deformation substantially improves the quality of high-order quadrilateral elements and prevents inversion, in contrast to simply projecting the high-order boundary points. %Thus, the MVC-based strategy enables robust deformation of high-order elements while preserving geometric fidelity and mesh quality.

\begin{figure}
  \centering
\includegraphics[width=.9\textwidth]{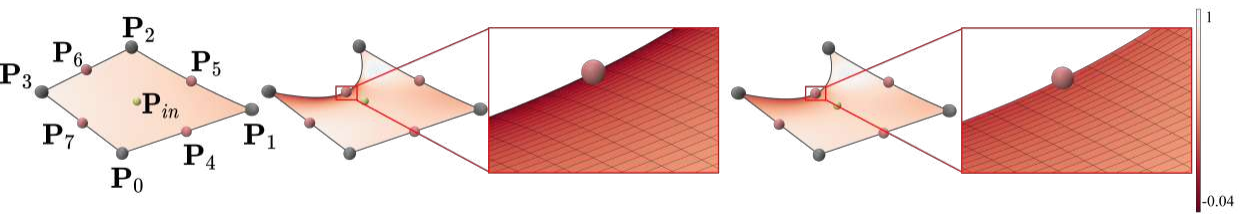}

     \parbox[t]{.2\textwidth}{\centering
            (a) 
           }
  \parbox[t]{.35\textwidth}{\centering
           (b) 
           }
    \parbox[t]{.35\textwidth}{\centering
           (c) 
           }
           
  \caption{\label{mvc}
           MVC-based deformation of a high-order quadrilateral element. (a) After degree elevation ($\min J_m=0.66$); (b) after boundary matching ($\min J_m=-0.04$); (c) final result after MVC deformation ($\min J_m=0.21$). Corner (gray), edge (red), and face (yellow) interpolation nodes are shown. The color of each point in the element is determined by the shape measure $J_m$ (see Section~\ref{quality metrics}).}
\end{figure}

%In summary, we present a method for generating high-quality, high-order unstructured quadrilateral meshes based on bounded-geometry-error curve reconstruction. Our approach transforms high-quality high-order mesh generation into a curve reconstruction problem with controlled geometric error, significantly improving efficiency. By employing a tailored Blossom-Quad algorithm and MVC-based high-order element deformation, the method maintains high mesh quality and avoids inverted or degenerate elements. Owing to the adoption of a meticulous curve reconstruction strategy and indirect mesh generation strategy naturally preserves boundary/interface features and ensures mesh conformity across interfaces.

\section{Experimental results}
In this section, we present the results of the proposed method for generating high-quality, high-order unstructured quadrilateral meshes. We compare our approach with existing high-order quadrilateral meshing techniques, including HOHQMesh~\cite{kopriva2024hohqmesh}, parametrization-based methods~\cite{xu2011parameterization,xu2013constructing,nian2016planar,pan2018low,ji2021constructing,wang2022tcb}, and the TMOP method~\cite{dobrev2019target}. Performance is evaluated in terms of the quality of the generated high-order quadrilateral meshes and computational runtime.

Our meshing pipeline is implemented in C++. Boundary- and interface-preserving triangular meshes are generated using the frontal-Delaunay for quads method~\cite{remacle2013frontal} integrated in Gmsh~\cite{geuzaine2009gmsh}. We employ a modified blossom-quad algorithm to prevent the formation of quadrilateral elements with two edges on boundaries or interfaces. All experiments and timing measurements are performed on a MacBook Air M2 with 4 performance cores and 8GB of memory.

Unless stated otherwise, the target element edge length is set to $l_{t}=0.05 \times \text{bbox}$, where $\text{bbox}$ is the diagonal length of the bounding box of the input curve. The geometric tolerance is set to $ \epsilon = 0.001 \times \text{bbox} $. In the adaptive refinement step, the control vector angle threshold is $\tau = \frac{\pi}{4}$ and segmentation density control parameter is $\alpha=2.3$. 

\begin{figure}
  \centering
  \includegraphics[width=1\textwidth]{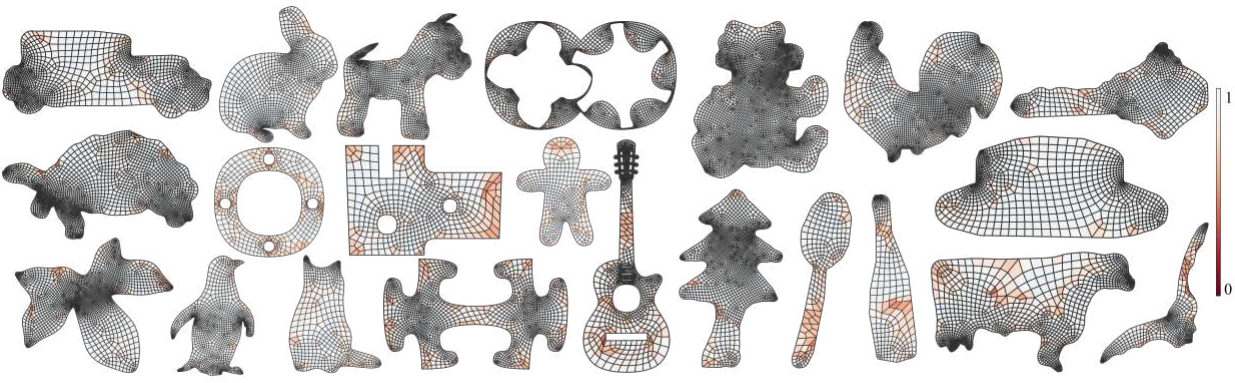}
  
  \caption{\label{data set without interface}
          Sample results from the high-order ($n=2,3$), high-quality mesh dataset without interfaces. Element color is indicated by the minimum shape measure $J_m$.}
\end{figure}

\begin{figure}
  \centering
  \includegraphics[width=1\textwidth]{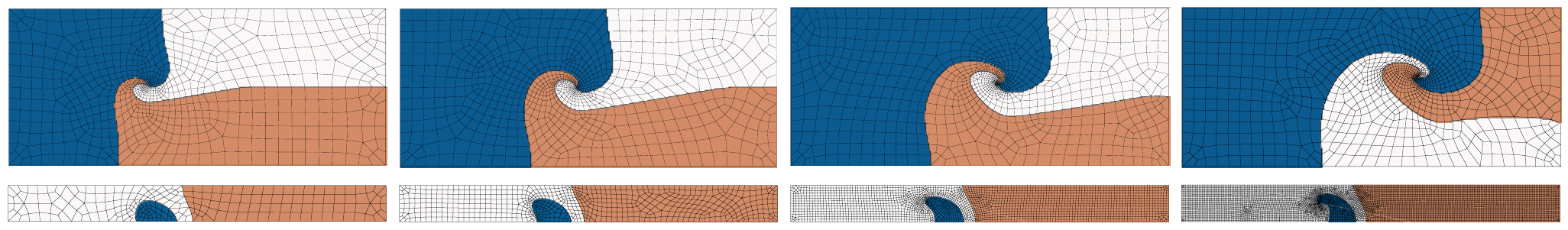}
  \caption{\label{data set with interface}
          High-order high-quality mesh dataset with material interfaces. Different colors indicate regions with distinct materials. Top: meshes of the triple-point problem at different time steps ($n=4$); bottom: meshes of the shock–bubble interaction at different time steps ($n=2$).}
\end{figure}

\subsection{Quality metrics}\label{quality metrics}
To quantitatively evaluate geometric fidelity and mesh quality of the generated high-order quadrilateral meshes, we use the following metrics: (i) the relative geometric error, and (ii) shape, skew, and singular-point measures~\cite{knupp2022geometric} for mesh quality.

\textbf{Relative geometric error $b_e$.} The relative geometric error between the original arcs $\mathcal{C}_\text{ori}$ and the reconstructed curve $\mathcal{C}_\text{cur}=\{\mathbf{c}_i\}$ is defined as
\[
  b_e = \frac{\max_i\{\hat{b}_{k}(\mathbf{c}_i,\mathbf{c}_i(0)\mathbf{c}_i(1)\!\mid_{\mathcal{C}_{\text{ori}}})\}}{\text{bbox}},
\]
where $\text{bbox}$ denotes the diagonal length of the bounding box of the input curve.

\textbf{Shape measure $J_m$.}
For a high-order quadrilateral element with Jacobian matrix $J$, the shape measure is defined as
\begin{equation}
    J_m = \frac{2 \det(J)}{\|J\|_F^2},
\end{equation}
where $\|J\|_F$ is the Frobenius norm of $J$.
The measure quantifies the deviation of the high-order element from the unit square, with values near $1$ indicating good shape quality (uniformity and orthogonality).

\textbf{Skew measure $J_k$.}
Let $\mathbf{j}_1$ and $\mathbf{j}_2$ denote the two column vectors of the Jacobian $J$.
The skew measure is defined as
\begin{equation}
    J_k = \frac{\det(J)}{\|\mathbf{j}_1\|_2 \, \|\mathbf{j}_2\|_2}.
\end{equation}
This metric evaluates only the orthogonality of the element, independent of edge-length disparity. 

Both $J_m$ and $J_k$ share the sign of $\det(J)$, and mesh validity requires $\det(J)>0$ everywhere.
We evaluate each element at $100$ Gauss-Lobatto points~\cite{karniadakis2013spectral} and report the minimum and average $J_m$ and $J_k$ values across all elements.

\textbf{Number of singular points $n_s$.} We report $n_s$, the number of singular vertices (those with valence not equal to four), to assess the topological regularity of the generated high-order quadrilateral meshes.

\subsection{Algorithm performance evaluation}
To evaluate the effectiveness of the proposed method, we first test it on a dataset of 110 models drawn from~\cite{zheng2019boundary,barroso2022efficient,liu2023topology}, including 90 boundary-only models and 20 models with internal interfaces. In our experiments, the generated elements used the same polynomial degree as the input geometry, focusing on 2nd- and 3rd-order models, which are the most commonly used in practice.

\begin{figure}
  \centering
\includegraphics[width=.24\textwidth]{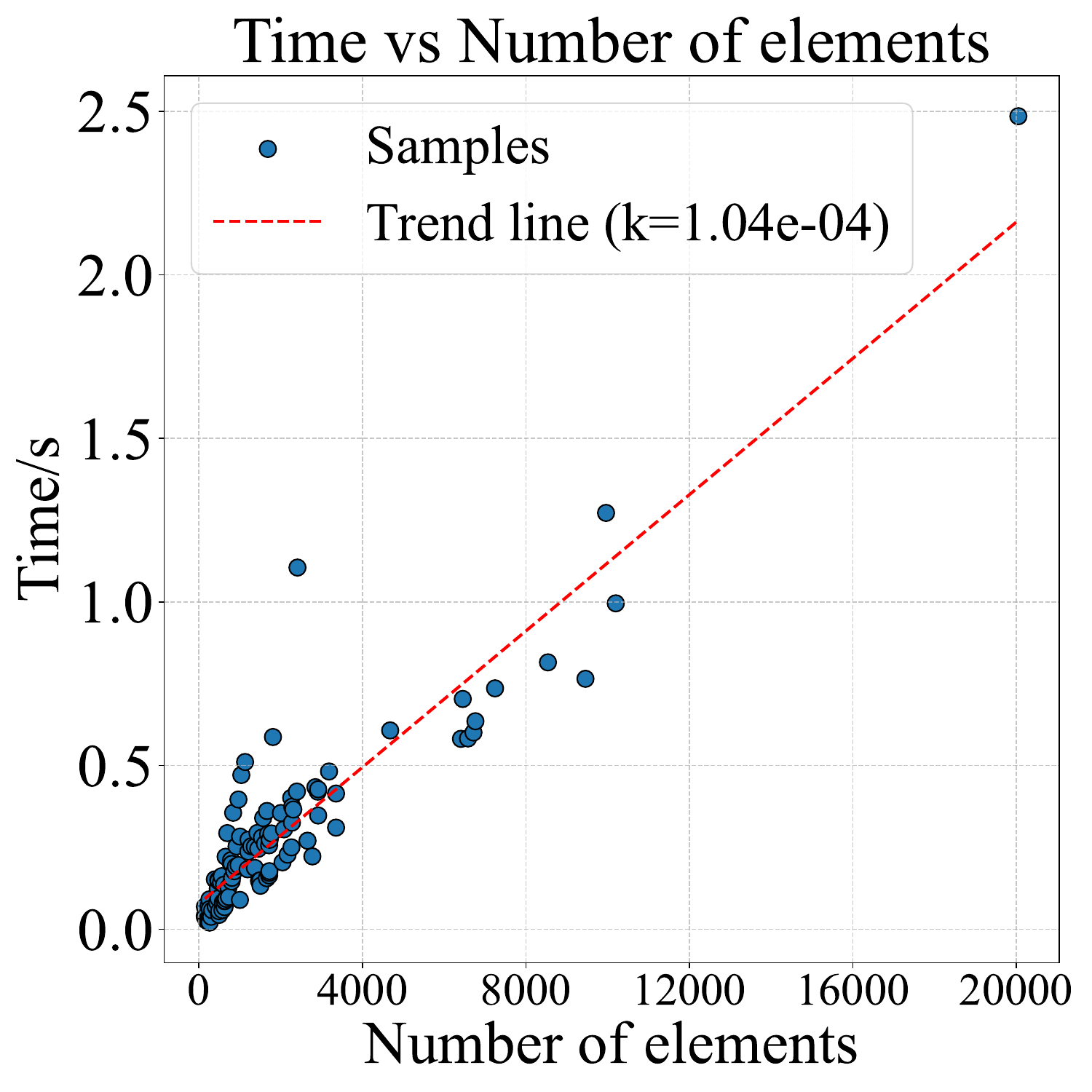}
\includegraphics[width=.24\textwidth]{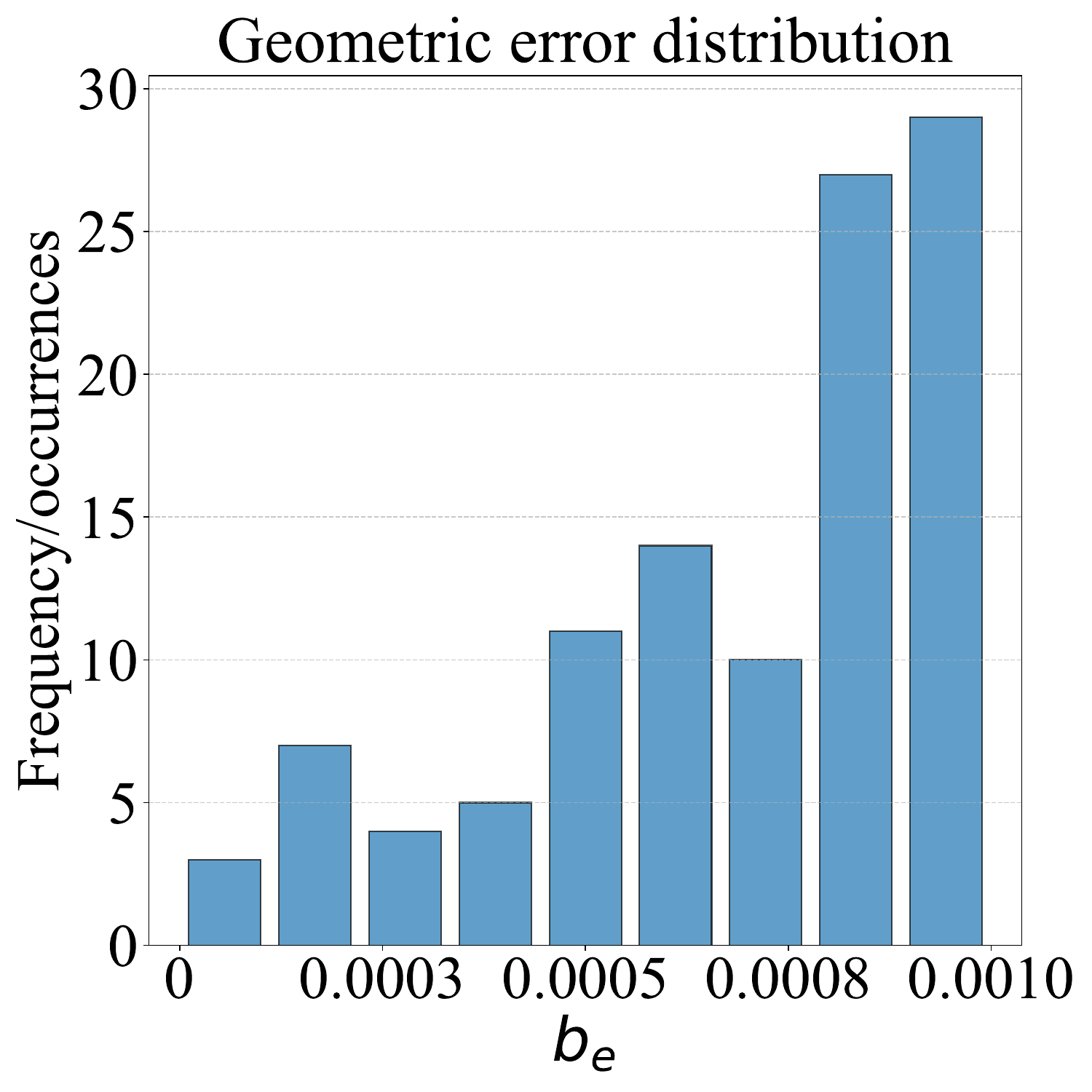}
\includegraphics[width=.24\textwidth]{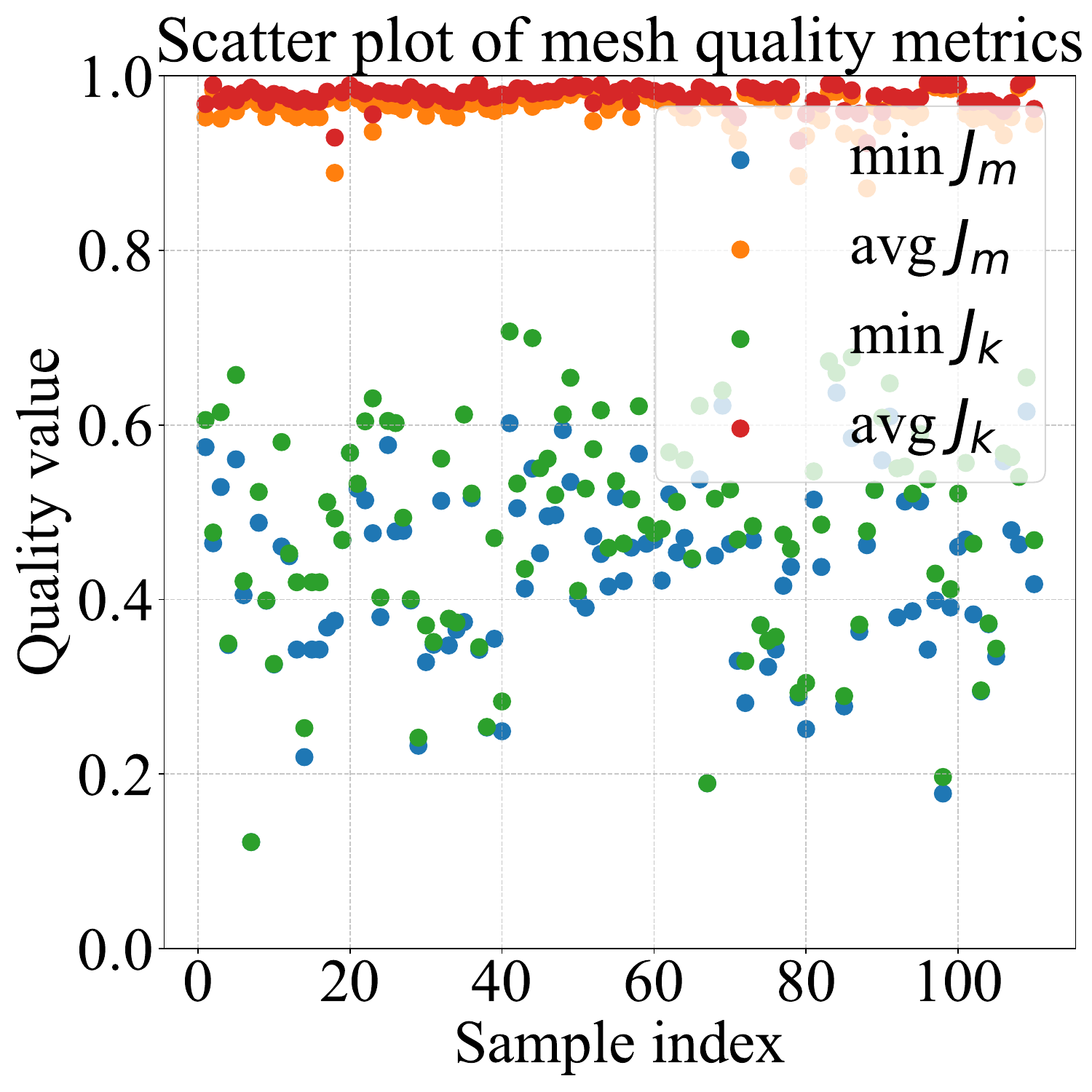}
\includegraphics[width=.24\textwidth]{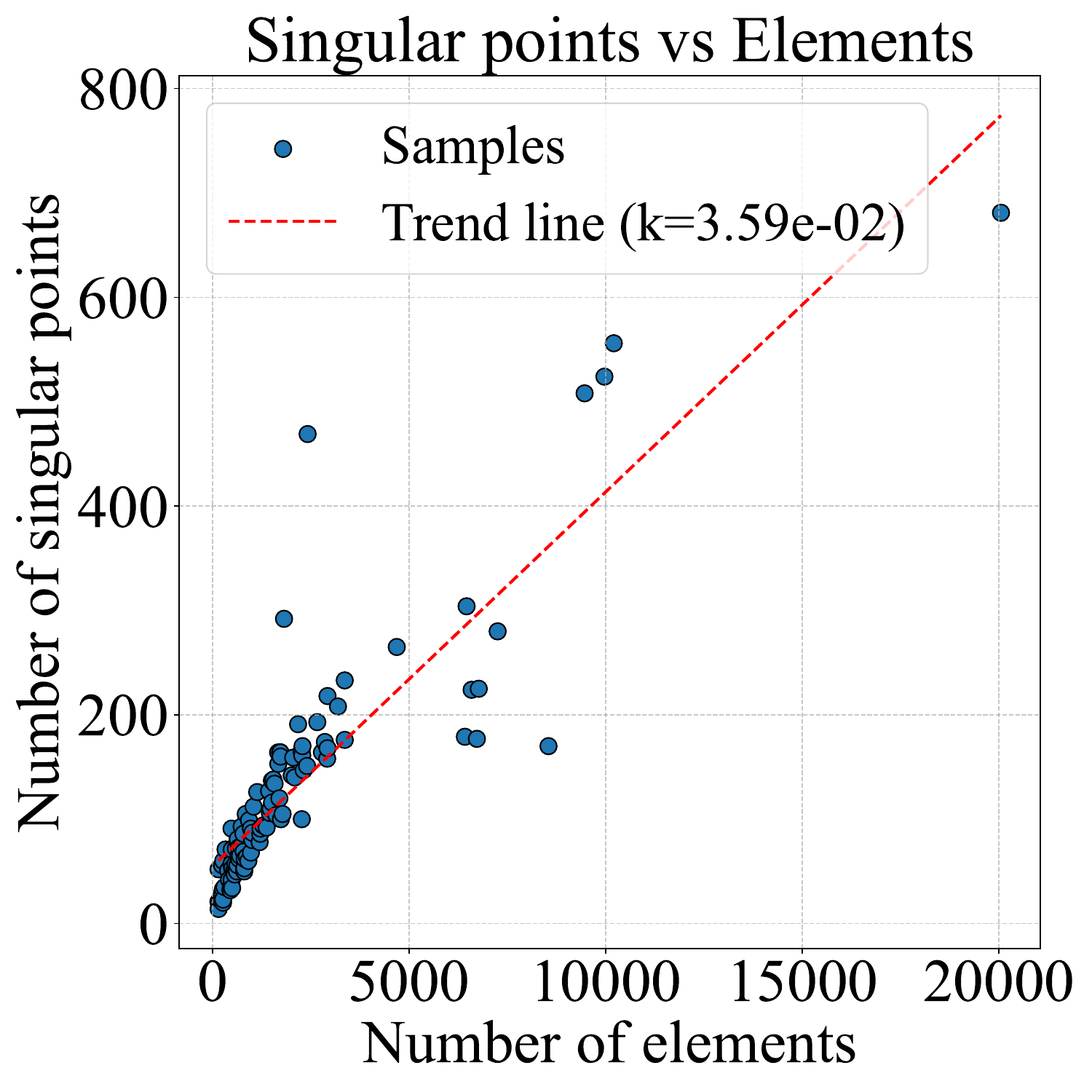}

     \parbox[t]{.24\textwidth}{\centering
            (a) 
           }
  \parbox[t]{.24\textwidth}{\centering
           (b) 
           }
   \parbox[t]{.24\textwidth}{\centering
           (c) 
           }
     \parbox[t]{.24\textwidth}{\centering
           (d) 
           }
           
  \caption{\label{data_set_data}
           Comprehensive statistics for 110 models. (a) Computation time vs. number of elements; (b) distribution of geometric error; (c) distribution of mesh quality; (d) number of singular points vs. number of elements.}
\end{figure}

\begin{figure}
  \centering
\includegraphics[width=.32\textwidth]{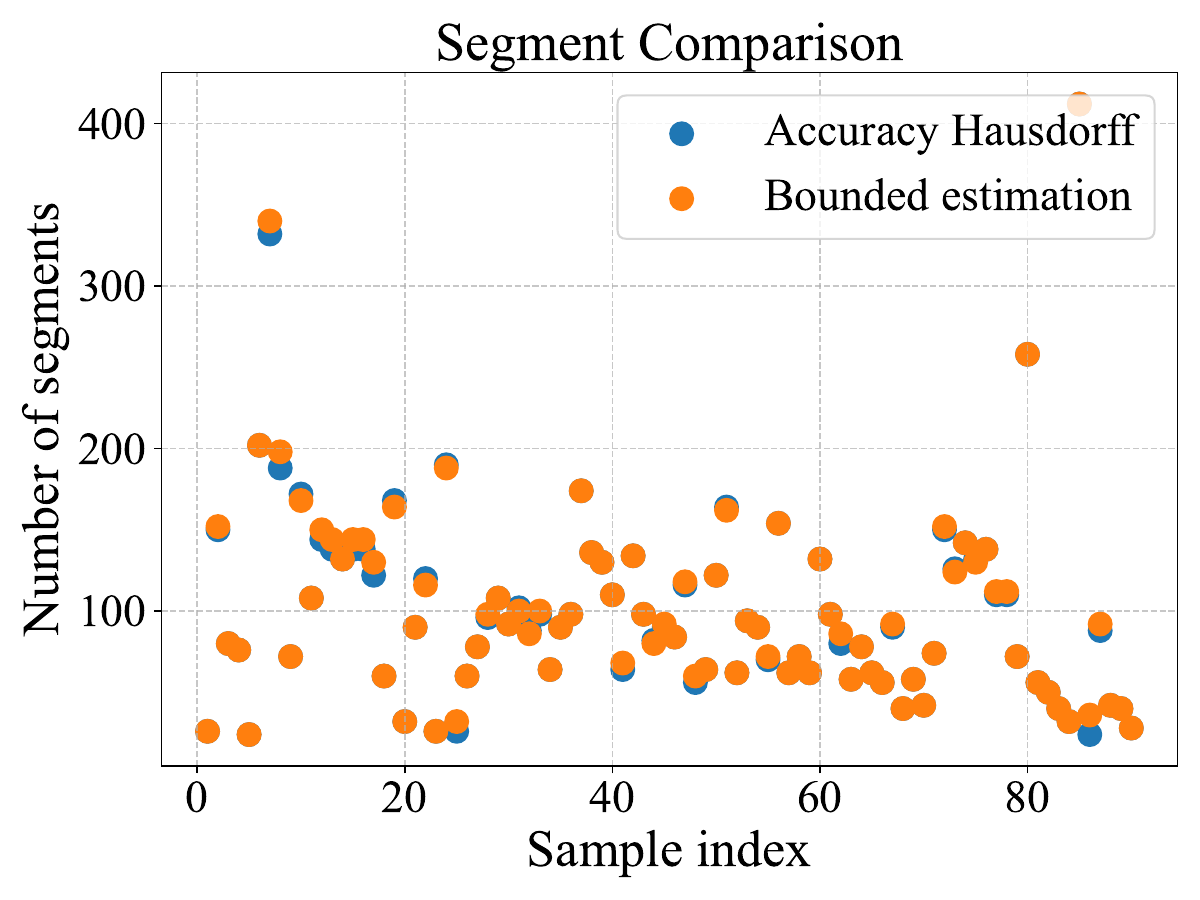}
\includegraphics[width=.32\textwidth]{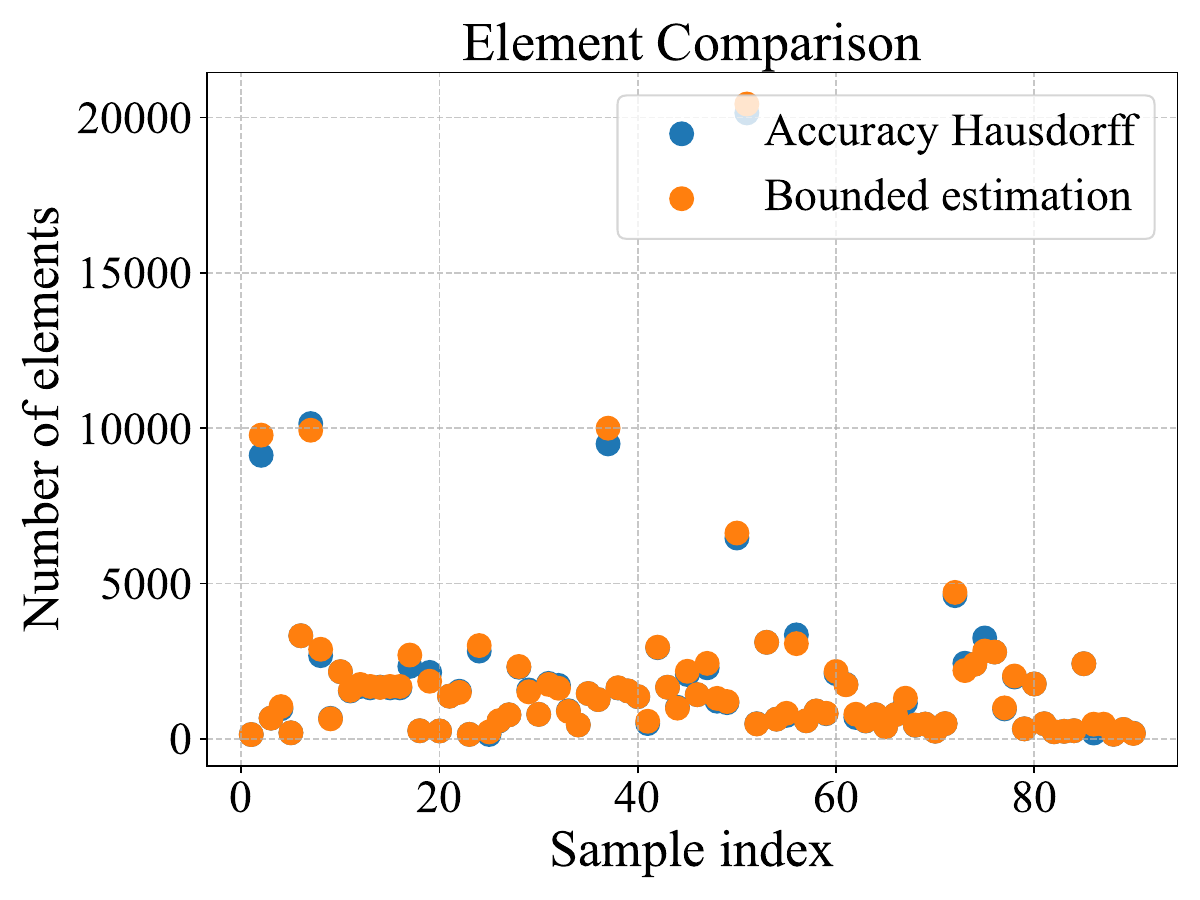}
\includegraphics[width=.32\textwidth]{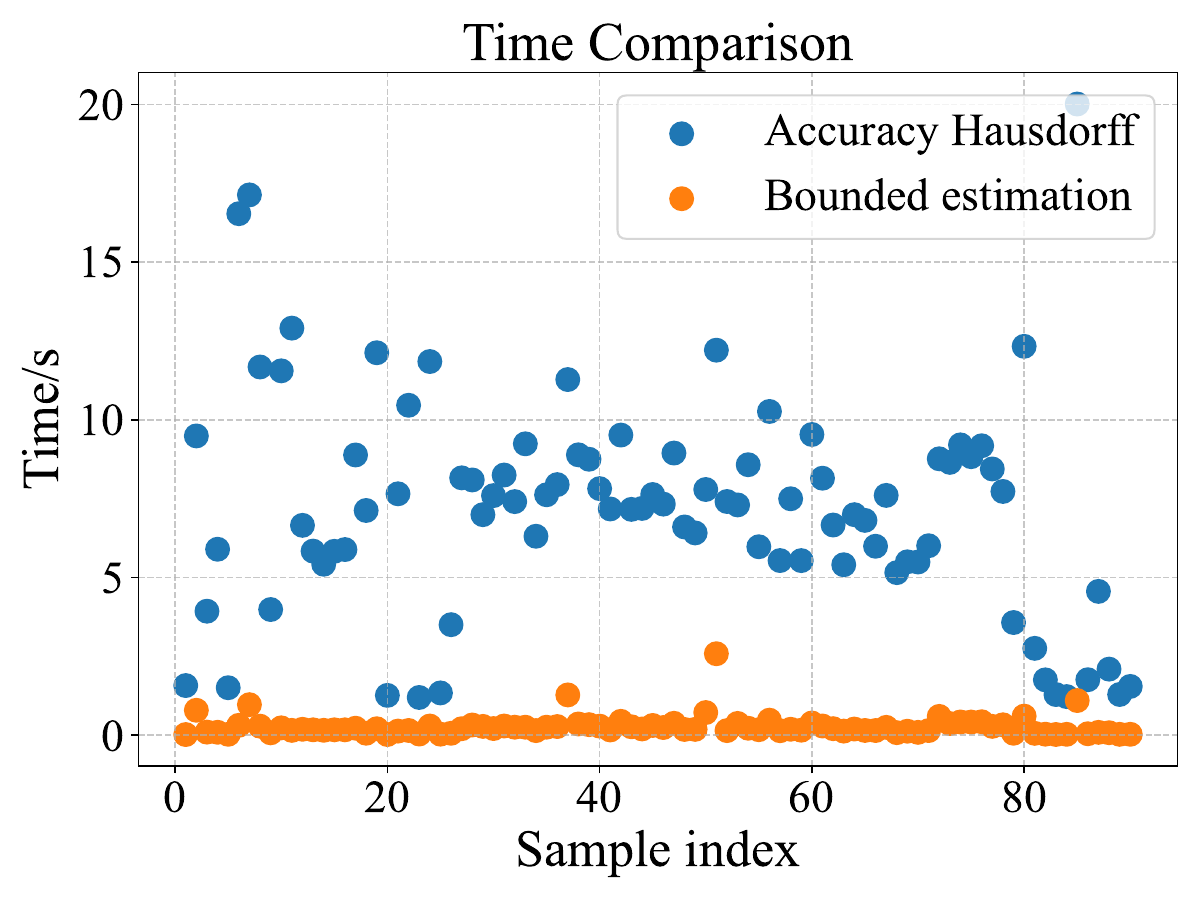}

     \parbox[t]{.32\textwidth}{\centering
            (a) 
           }
  \parbox[t]{.32\textwidth}{\centering
           (b) 
           }
     \parbox[t]{.32\textwidth}{\centering
            (c) 
           }
           
  \caption{\label{our bounded estimation vs accuracy}
           Effect of our bounded Hausdorff distance estimation versus exact Hausdorff distance estimation: (a) segment-level comparison, (b) element-level comparison, and (c) computational time.}
\end{figure}

Figs.~\ref{data set without interface} and~\ref{data set with interface} show quadratic and cubic high-order quadrilateral meshes generated by the proposed method for models without and with interfaces, illustrating accurate geometry preservation and interface conformity. Fig.~\ref{data_set_data}(a) shows that computation time scales nearly linearly with the number of high-order elements (fitting factor: \(1.04\times 10^{-4}\) s/element), demonstrating the efficiency of the optimization-based curve reconstruction. Fig.~\ref{data_set_data}(b) confirms bounded geometric errors, with \(b_e < 10^{-3}\) for all models. Fig.~\ref{data_set_data}(c) shows high mesh quality, with \(J_m\) and \(J_k\) averages near 1 and minima positive, ensuring no inverted or degenerate elements. Fig.~\ref{data_set_data}(d) reveals a near-linear increase of singular points with element count (fitting factor: \(3.59\times 10^{-2}\)/element).

We sample $1,000$ points on each $\mathbf{c}_i \in \mathcal{C}_{\text{cur}}$ and its original counterpart $\mathbf{c}_i(0)\mathbf{c}_i(1)|_{\mathcal{C}_{\text{ori}}}$, and compute the Hausdorff distance {\color{black}using efficient point-based computation method~\cite{taha2015efficient}} between these point sets as a proxy for the exact Hausdorff distance, replacing the upper-bound estimate proposed in this work. Fig.~\ref{our bounded estimation vs accuracy} compares results on $90$ boundary-only cases, showing that both the exact Hausdorff distance and the proposed upper bound produce similar numbers of final curve segments and high-order elements, while the proposed bound significantly improves computational efficiency.

The proposed method performs well on 2nd- to 4th-order meshes (Figs.~\ref{data set without interface} and~\ref{varies order}), which are commonly employed in high-order finite element computations. To demonstrate the scalability of our approach, we also show an example using elements of 10-th order, where the input consists of curves of degree 10, as illustrated in Figure~\ref{order10}. Such high-order elements may be required for high-accuracy applications, including spectral element methods~\cite{hafeez2023review}.

\begin{figure}
  \centering
\includegraphics[width=.75\textwidth]{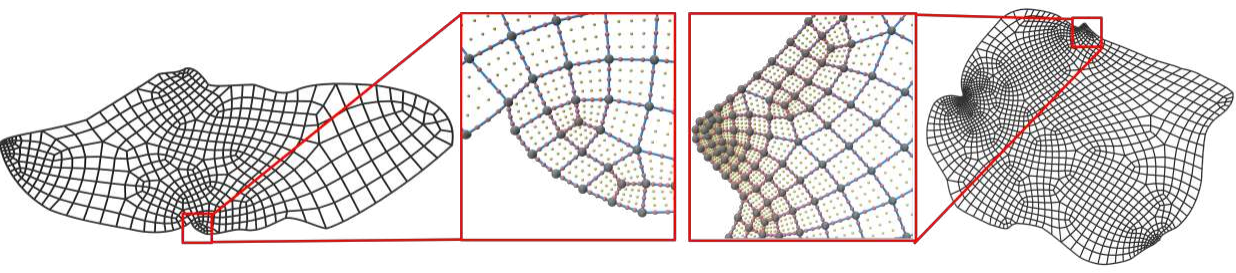}

  \caption{\label{varies order}
          \textcolor{black}{Illustration of a 4th-order quadrilateral mesh. Gray points denote element vertices, red points denote high-order edge nodes, and yellow points denote high-order face nodes.}}
\end{figure}

%In summary, extensive experiments demonstrate that the proposed method effectively generates high-quality, high-order quadrilateral meshes for complex geometries, preserving interface conformity and geometric features with bounded error, while producing inversion-free and non-degenerate elements. The optimization-based reconstruction strategy and the use of the Hausdorff distance upper bound significantly enhances computational efficiency.

\subsection{Parameter influence}

The proposed method introduces three parameters during curve reconstruction: the geometric error threshold $\epsilon$, the angular threshold $\tau$, and the segmentation density control parameter $\alpha$ for adaptive refinement. In the following, we analyze the influence of these parameters on the quality of the resulting high-order quadrilateral meshes through experiments.

\textbf{Geometric error threshold $\epsilon$.}
In the optimization-based curve reconstruction process, the geometric error threshold $\epsilon$ controls how closely the high-order mesh boundaries follow the input curve. Fig.~\ref{para epsilon} illustrates the input curve with fine geometric details, its reconstruction, and the resulting high-order quadrilateral meshes under different thresholds. When $\epsilon = 0$, the mesh strictly follows the input curve segmentation, resulting in unnecessarily dense elements in simple regions. For $\epsilon = 0.05 \times \text{bbox}$, the mesh is sparser, but the reconstructed curve deviates significantly from the original, failing to capture fine features. Using the default value $\epsilon = 0.001 \times \text{bbox}$ produces curves that closely match the input while maintaining a reasonable mesh density.

\begin{figure}
  \centering
\includegraphics[width=1.0\textwidth]{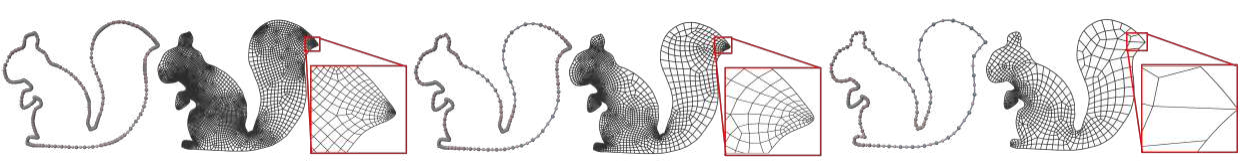}

\parbox[t]{.3\textwidth}{\centering
            (a) $\epsilon=0$
           }
\parbox[t]{.3\textwidth}{\centering
           (b) $\epsilon=0.001\times \text{bbox}$
           }
\parbox[t]{.3\textwidth}{\centering
           (c) $\epsilon=0.05\times \text{bbox}$
           }

  \caption{\label{para epsilon}
           Effect of given geometric error threshold $\epsilon$ ($n=2$).}
\end{figure}

\textbf{Angular threshold $\tau$.}
In the adaptive refinement process, the parameter $\tau$ controls the refinement level. According to Theorem~\ref{valid high-order quad theorem}, large $\tau$ may produce invalid high-order quadrilateral elements (Fig.~\ref{para R}(c)), whereas small $\tau$ can lead to over-refinement in curved regions, resulting in dense meshes (Fig.~\ref{para R}(a)) and reduced efficiency. Based on our extensive experiments, we adopt $\tau = \pi/4$, which balances validity and refinement (Fig.~\ref{para R}(b)) and is suitable for $n = 2, 3, 4$. For larger $n$, although $\tau = \pi/4$ remains valid, it tends to produce denser meshes (Fig.~\ref{order10}(a)), since higher-order elements can represent geometry more efficiently with fewer elements; thus, a larger value (e.g., $\tau = \pi/3$ for $n = 10$) is preferred. This is consistent with the degree elevation property of B\'{e}zier curves~\cite{farin2001curves}: as $n \to \infty$, the control polygon converges to the curve, causing the initial range $r_1$ to shrink toward the span of endpoint tangents while enlarging the admissible range for $r_0$. As shown in Fig.~\ref{order10}(b), increasing $\tau$ to $\pi/3$ yields a high-quality mesh with fewer elements while preserving geometric accuracy.

\begin{figure}
  \centering
\includegraphics[width=1.0\textwidth]{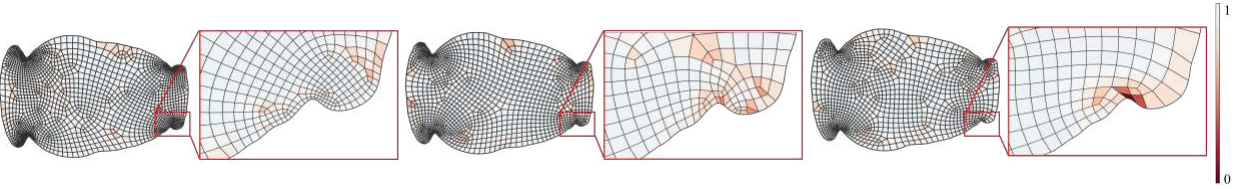}

\parbox[t]{.25\textwidth}{\centering
            (a) $\tau=\frac{\pi}{6}$
           }
\parbox[t]{.25\textwidth}{\centering
           (b) $\tau=\frac{\pi}{4}$
           }
\parbox[t]{.25\textwidth}{\centering
           (c) $\tau=\frac{\pi}{3}$
           }

  \caption{\label{para R}
           Effect of parameter $\tau$ in adaptive refinement ($n=3$). The color of element represents its minimum shape measure $J_m$. Inverted elements in (c) are highlighted in dark red.}
\end{figure}

\begin{figure}
  \centering
\includegraphics[width=.35\textwidth]{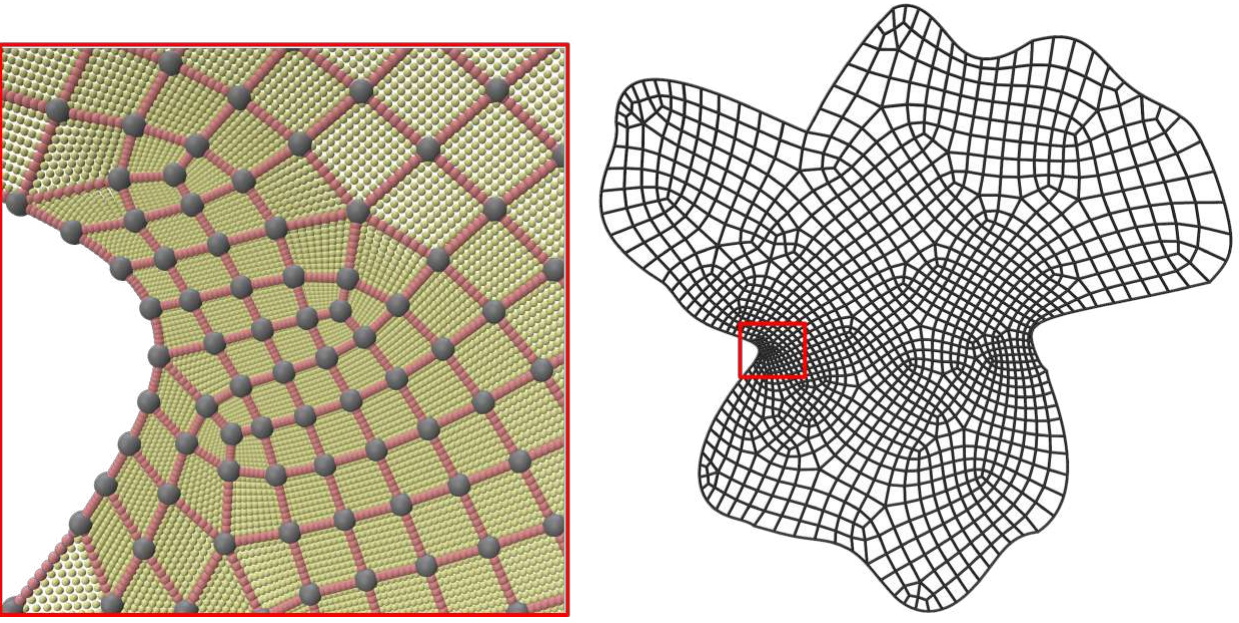}
\includegraphics[width=.35\textwidth]{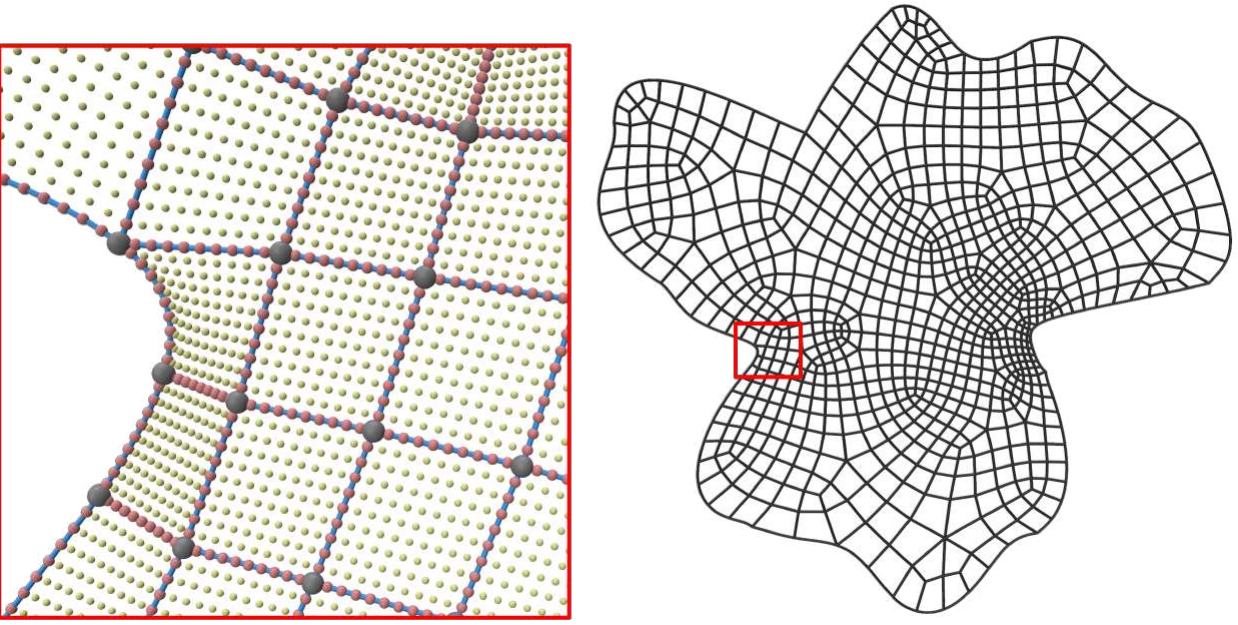}

\parbox[t]{.35\textwidth}{\centering
            (a) $\tau=\frac{\pi}{4}$
           }
\parbox[t]{.35\textwidth}{\centering
           (b) $\tau=\frac{\pi}{3}$
           }

  \caption{\label{order10}
          High-order meshes ($n=10$) with $\tau = \pi/4$ (left, $1426$ elements) and $\tau = \pi/3$ (right, $792$ elements). Zoom-ins highlight high-order nodes.}
\end{figure}

% \begin{figure}[htb]
%   \centering
% \includegraphics[width=.9\textwidth]{figures/figure16_para_l_0_3.pdf}
% \parbox[t]{.9\textwidth}{\centering
%             (a) $l_{t}=0.30$
%            }
           
% \includegraphics[width=.9\textwidth]{figures/figure16_para_l_bbox.pdf}
%  \parbox[t]{.9\textwidth}{\centering
%            (b) $l_{t}=0.62$
%            }
           
% \includegraphics[width=.9\textwidth]{figures/figure16_para_l_1_0.pdf}
%    \parbox[t]{.9\textwidth}{\centering
%            (c) $l_{t}=1.00$
%            }
           
%   \caption{\label{para l}
%            Effect of target element edge length $l_{t}$.}
% \end{figure}

% \textbf{Target element edge length $l_{t}$.} The target element edge length $l_{t}$ influences the final result of the generated high-order unstructured quadrilateral mesh. Figure~\ref{para l} illustrates the effect of different $l_{t}$ values on the mesh quality, where $l_{t} = 0.62$ corresponds to the default value $0.05 \times \text{bbox}$. As shown in the figure, using the default value produces geometrically adaptive meshes that are suitable for general models.

\textbf{Segmentation density control parameter $\alpha$.} The parameter $\alpha$ controls the density of curve segmentation, thereby influencing the final high-order quadrilateral mesh. Fig.~\ref{parameter alpha} illustrates its effect: smaller $\alpha$ values produce overly dense meshes, while larger values result in sparser meshes with abrupt element size transitions. The default value $\alpha = 2.3$ achieves a geometrically adaptive mesh with balanced density and smooth gradation.

\begin{figure}
  \centering
\includegraphics[width=.98\textwidth]{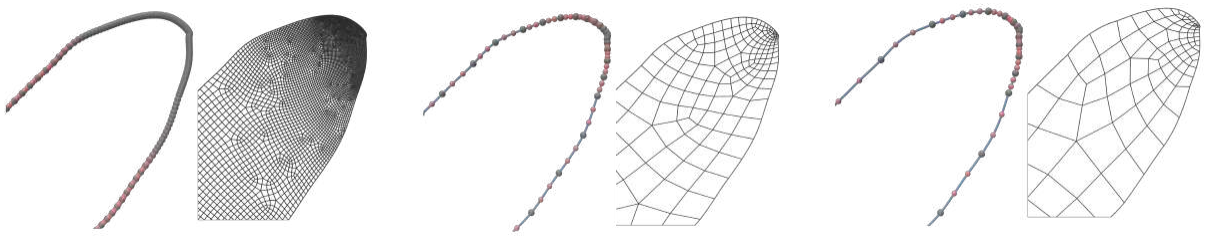}

\parbox[t]{.31\textwidth}{\centering
            (a) $\alpha=1.8$
           }
\parbox[t]{.31\textwidth}{\centering
           (b) $\alpha=2.3$
           }
\parbox[t]{.31\textwidth}{\centering
           (c) $\alpha=2.8$
           }

  \caption{\label{parameter alpha}
           Effect of parameter $\alpha$ in adaptive refinement ($n=3$).}
\end{figure}

%-------------------------------------------------------------------------
\subsection{Ablation study}
Ablation studies were conducted to evaluate the contribution of each step in the geometric error-bounded curve reconstruction. As shown in Fig.~\ref{ablation study figure}(a), skipping the error-bounded approximation produces high-quality meshes but with excessive element counts, reducing the efficiency of high-order representations. Omitting adaptive refinement and optimization, as in Fig.~\ref{ablation study figure}(b), yields the fewest elements but poor-quality or invalid meshes due to insufficient resolution. Performing only the first two steps without optimization, shown in Fig.~\ref{ablation study figure}(c), compromises mesh quality due to non-uniform endpoint distribution. In contrast, the full pipeline in Fig.~\ref{ablation study figure}(d) achieves a balanced combination of element count, high mesh quality, and bounded geometric error, demonstrating the necessity of each component.

\begin{figure}
  \centering
    \includegraphics[width=.2\textwidth]{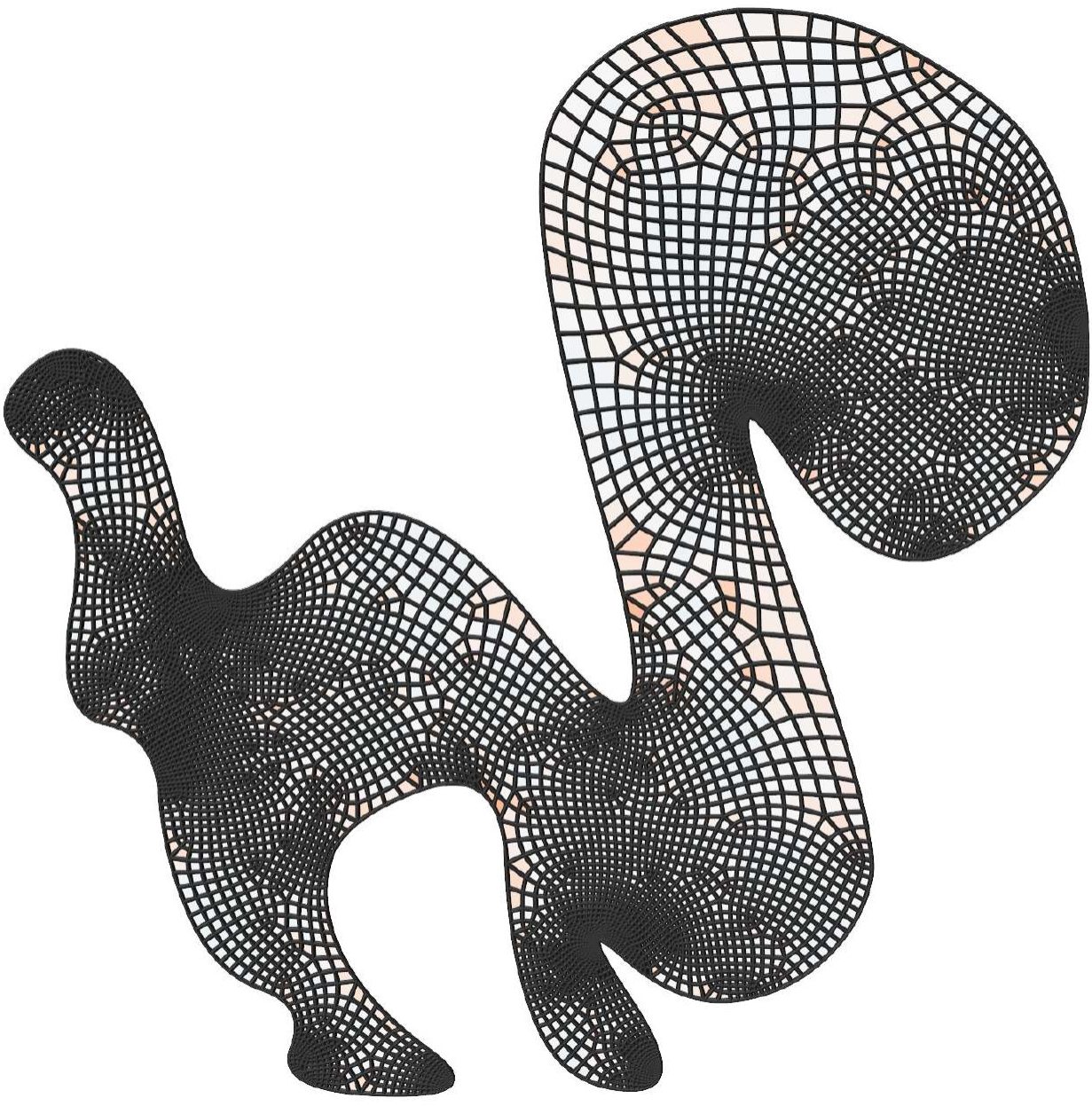}
  \includegraphics[width=.2\textwidth]{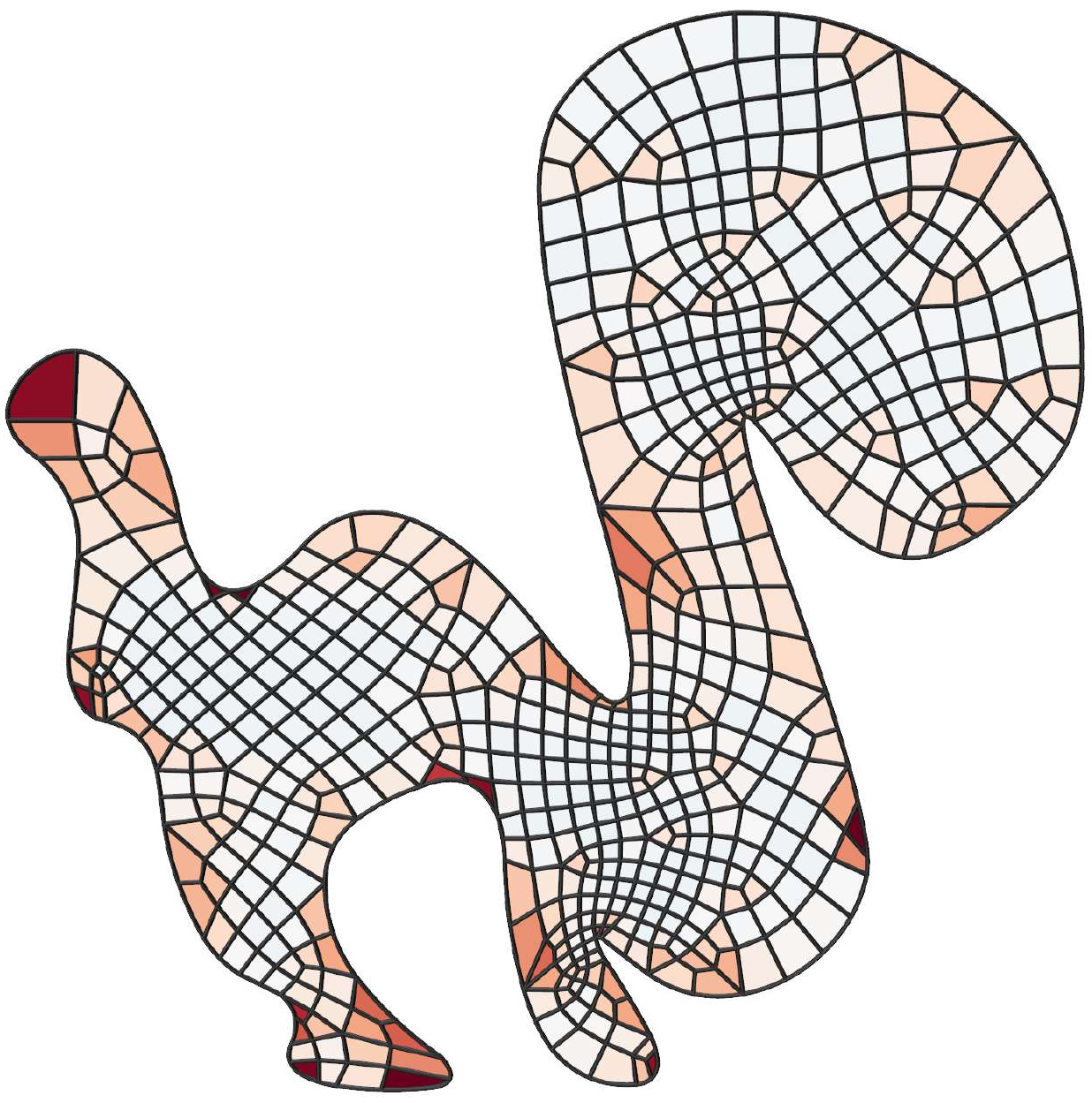}
  \includegraphics[width=.2\textwidth]{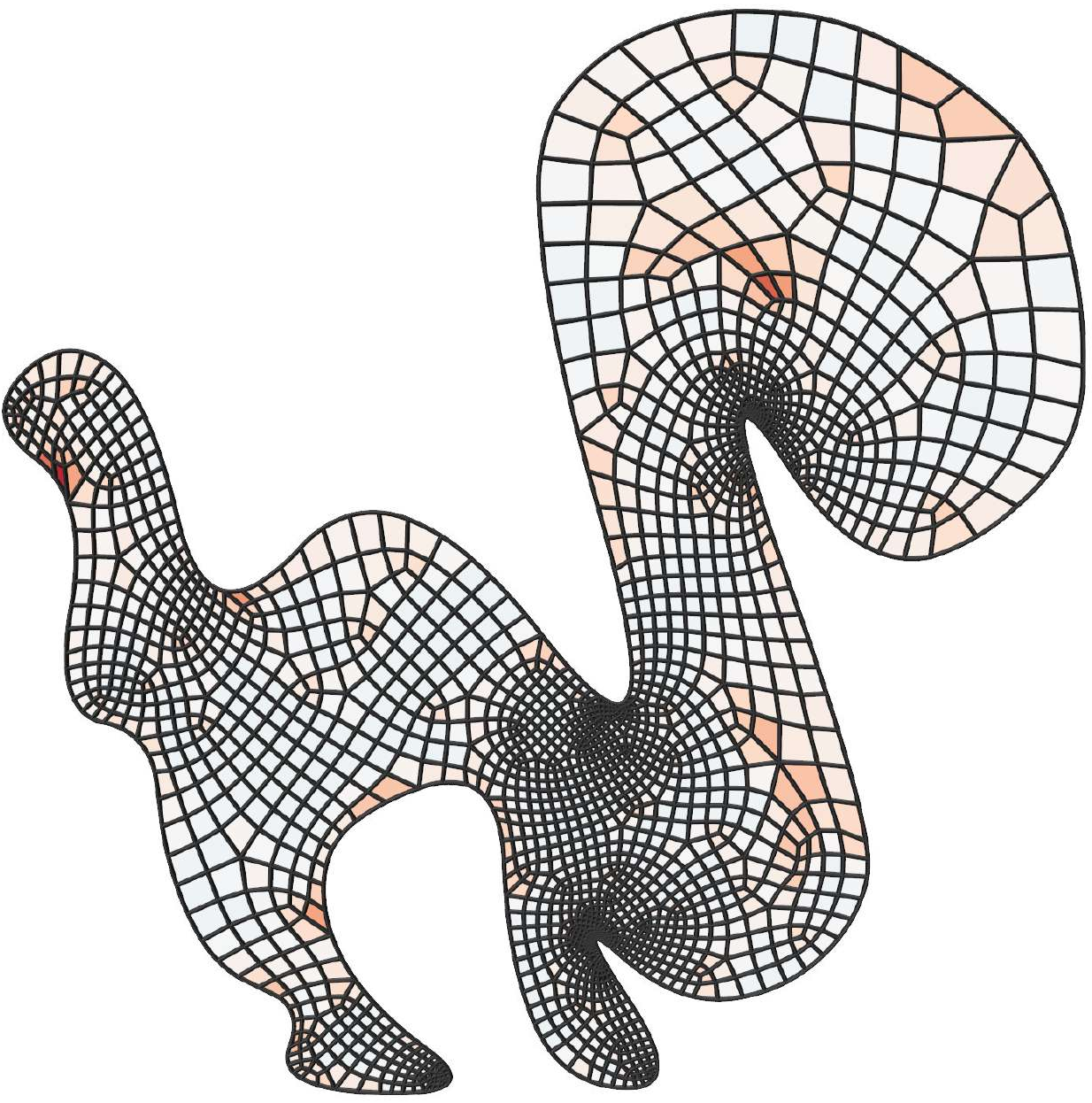}
    \includegraphics[width=.23\textwidth]{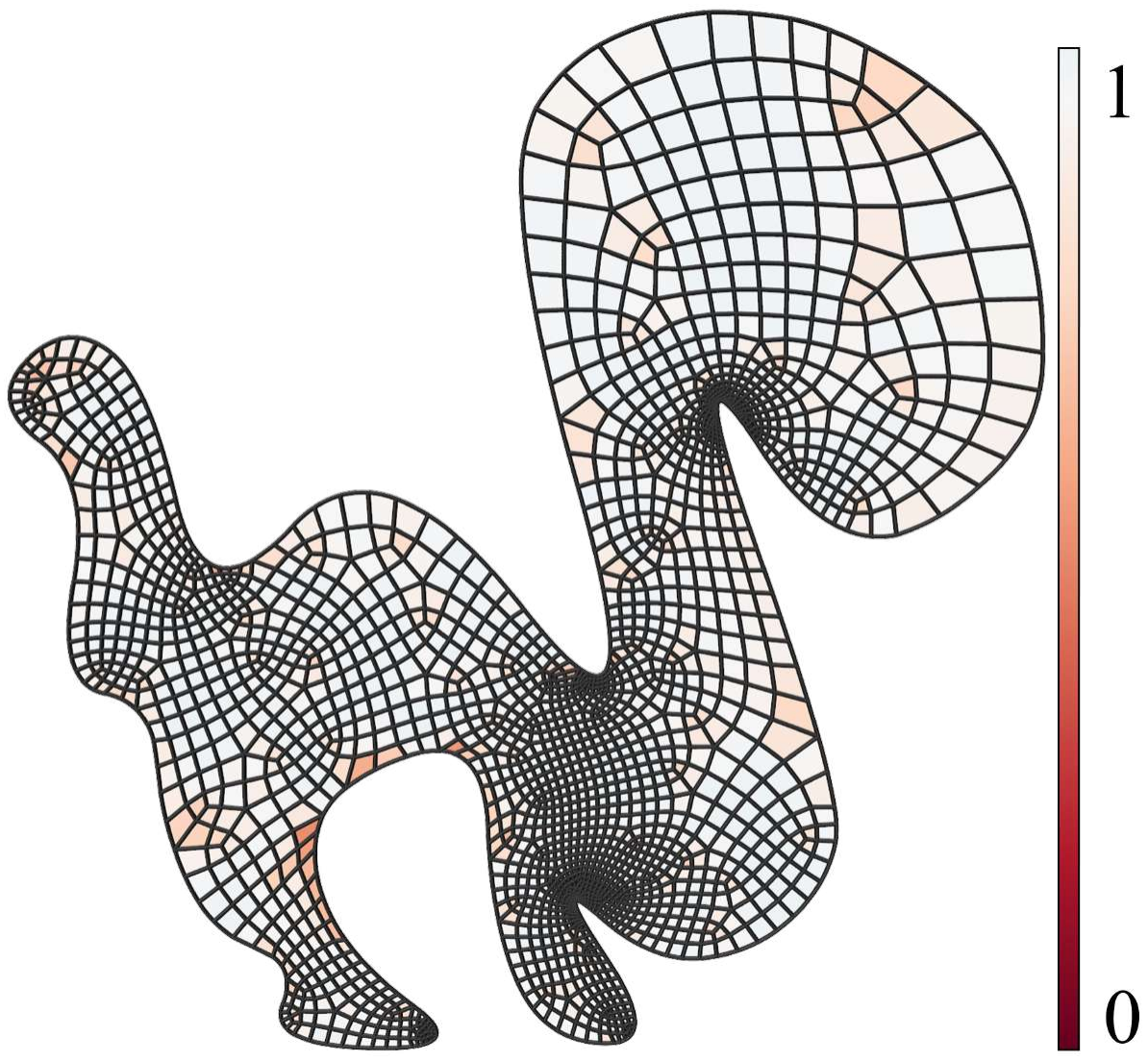}

    \parbox[t]{.2\textwidth}{\centering
            (a)
           }
  \parbox[t]{.2\textwidth}{\centering
           (b) 
           }
  \parbox[t]{.2\textwidth}{\centering
           (c)
           }
  \parbox[t]{.25\textwidth}{\centering
           (d) 
           }
           
  \caption{\label{ablation study figure}
           Ablation study ($n=2$). The color of element represents its minimum shape measure $J_m$. Inverted elements in (b) are highlighted in dark red.}
\end{figure}

\subsection{Comparative experiments}
We evaluate our method against linear quadrilateral meshing, the quadtree-based HOHQMesh~\cite{kopriva2024hohqmesh}, parametrization-based methods~\cite{xu2011parameterization,xu2013constructing,nian2016planar,pan2018low,ji2021constructing,wang2022tcb}, and the high-order mesh optimization approach TMOP~\cite{dobrev2019target}, demonstrating its advantages in mesh quality, interface conformity, and computational efficiency.

\textbf{Comparison with linear quadrilateral meshes.} We compare adaptive linear and high-order quadrilateral meshes under comparable geometric accuracy. The Hausdorff distances to the input curves are $0.0273$ and $0.0267$, respectively. As shown in Fig.~\ref{vs_ho_linear}(a), the linear mesh contains $3,713$ elements and $3,859$ nodes, whereas the high-order mesh in Fig.~\ref{vs_ho_linear}(b) has $536$ elements and $2,259$ nodes. This demonstrates that high-order quadrilateral meshes achieve superior geometric approximation with substantially fewer elements.

\begin{figure}
  \centering
\includegraphics[width=.3\textwidth]{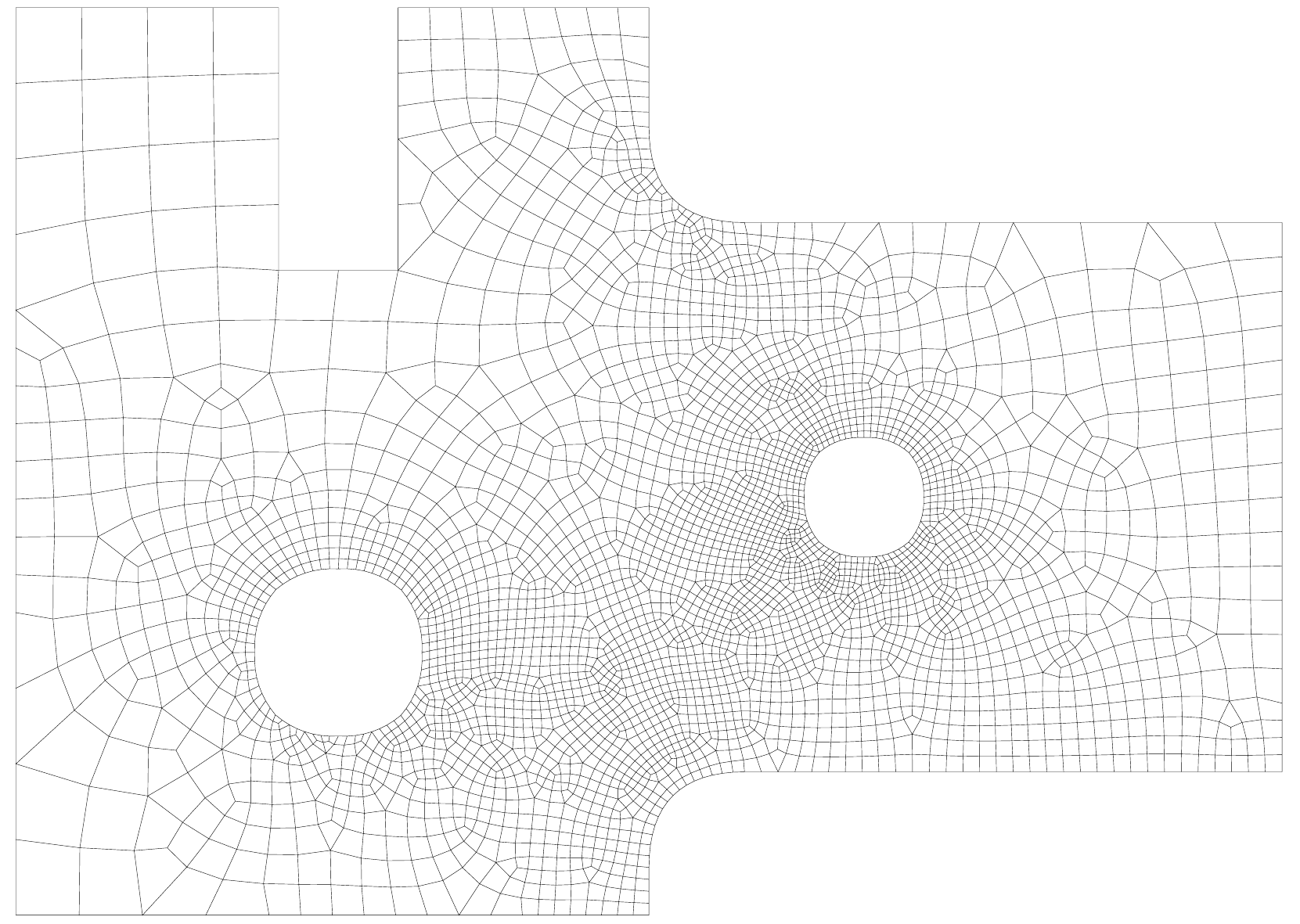}
\includegraphics[width=.3\textwidth]{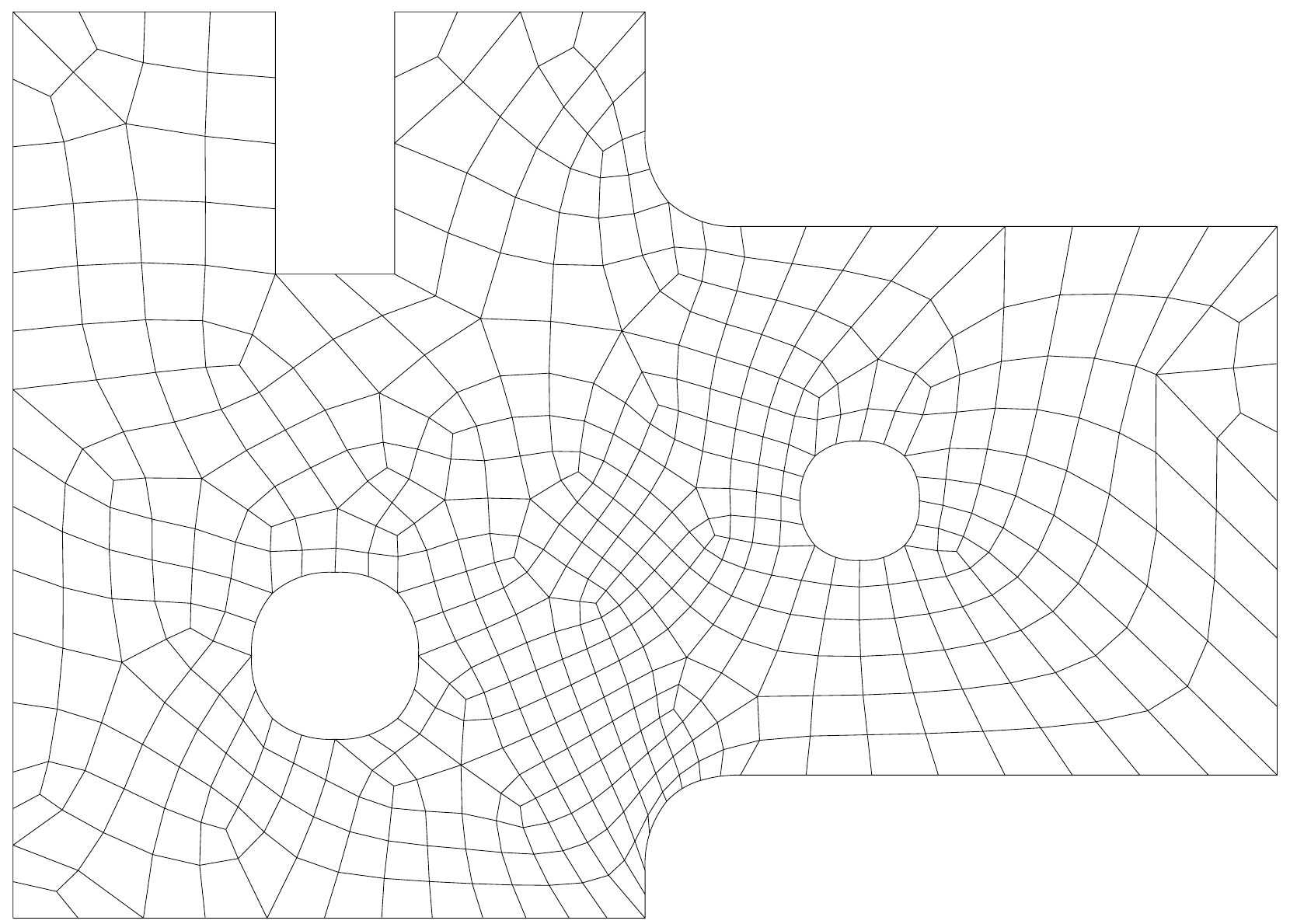}
  \caption{\label{vs_ho_linear}
           Comparison between linear (left) and quadratic (right) quadrilateral meshes.}
\end{figure}

\textbf{Comparison with HOHQMesh~\cite{kopriva2024hohqmesh}.} HOHQMesh generates unstructured all-quadrilateral and hexahedral meshes with high-order boundary information using a quadtree-based template subdivision. While effective for simple geometries, it struggles with complex models containing interfaces, often producing low-quality elements near interface regions and currently cannot handle non-closed interfaces (Fig.~\ref{data set with interface}). Fig.~\ref{vs_hohq_our} compares interface meshes generated by HOHQMesh and our method. HOHQMesh produces lower-quality elements near interfaces due to quadtree refinement, whereas our method maintains higher quality. Table~\ref{vs_HOHQMesh_data_quality} quantitatively confirms that our approach achieves superior mesh quality at comparable computational cost.

\begin{figure}
  \centering
\includegraphics[width=.95\textwidth]{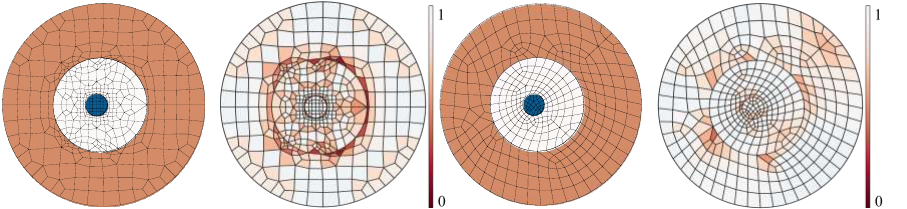}

\parbox[t]{.24\textwidth}{\centering
            (a) HOHQMesh (interface)
           }
  \parbox[t]{.24\textwidth}{\centering
           (b) HOHQMesh (quality)
           }
  \parbox[t]{.24\textwidth}{\centering
           (c) Our method (interface)
           }
  \parbox[t]{.24\textwidth}{\centering
           (d) Our method (quality)
           }
  \caption{\label{vs_hohq_our}
          Comparison between our method and HOHQMesh on an interface model ($n=2$): (a,c) interface meshes (colors denote materials) and (b,d) corresponding quality visualizations, with element color indicating the minimum shape measure $J_m$.}
\end{figure}

\begin{table} % 使用 table* 让表格跨两列
  \centering
  \caption{The mesh quality data of the example in Fig. \ref{vs_hohq_our} is presented.}
  \small
  \begin{tabular}{cccccccc}
    \toprule
Method & Number of elements & Times$/s$ & $\min J_{m}$&$\text{avg } J_{m}$ & $\min J_{k}$ & $\text{avg } J_{k}$ & $n_s$ \\
    \midrule
    HOHQMesh&464 &\textbf{0.05}
 & 0.18& 0.89 & 0.19& 0.95 & 163 \\
    Our method& \textbf{404}&0.06&\textbf{0.42}&\textbf{0.94}&\textbf{0.47}&\textbf{0.96} & \textbf{42} \\
    \bottomrule
    \multicolumn{8}{@{}l@{}}{\footnotesize Results in bold indicate superior performance.} \\
  \end{tabular}
  \label{vs_HOHQMesh_data_quality}
\end{table}

\textbf{Comparison with parametric methods.} Parametric methods for isogeometric analysis \cite{xu2011parameterization,xu2013constructing,nian2016planar,pan2018low,ji2021constructing,wang2022tcb} generate high-order quadrilateral meshes by constructing smooth mappings from physical domains to parametric domains. These methods typically enforce inter-element smoothness (e.g., $C^1$ continuity), which can hinder performance on complex geometries. In contrast, our approach avoids such constraints, producing high-quality, inversion-free, purely quadrilateral meshes that robustly handle complex boundaries and topologies.  

\begin{figure}
  \centering
\includegraphics[width=.16\textwidth]{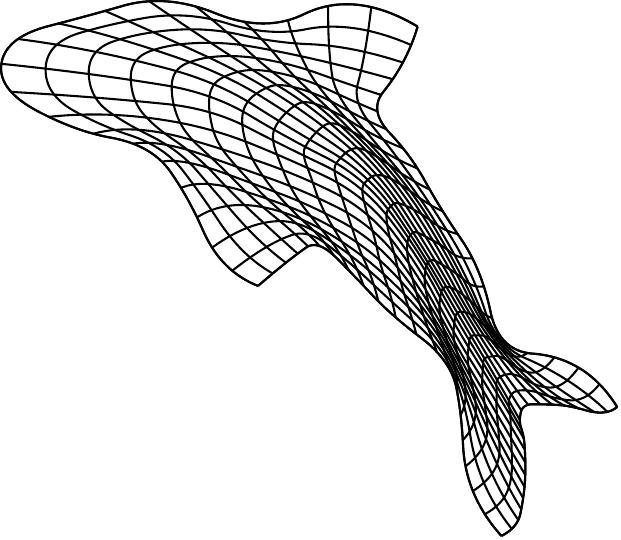}
\includegraphics[width=.16\textwidth]{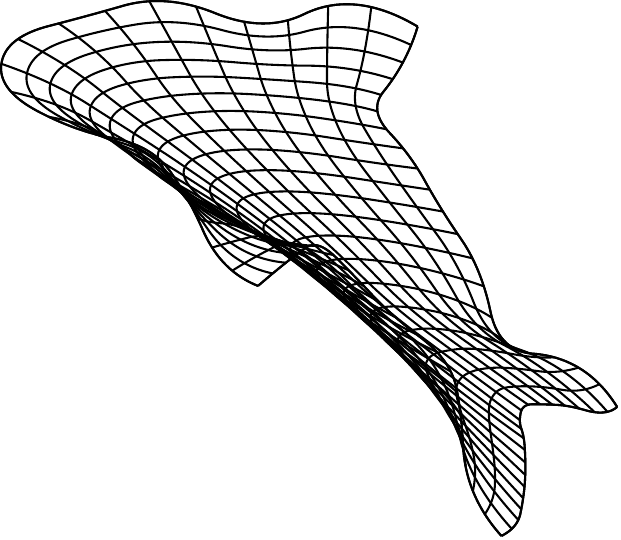}
\includegraphics[width=.16\textwidth]{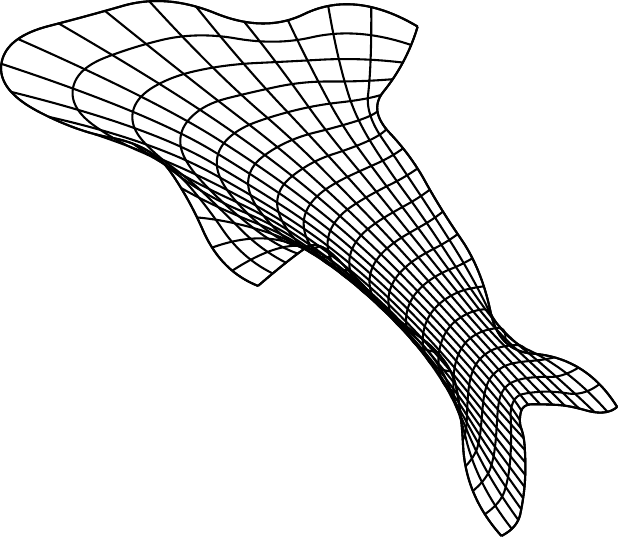}
\includegraphics[width=.16\textwidth]{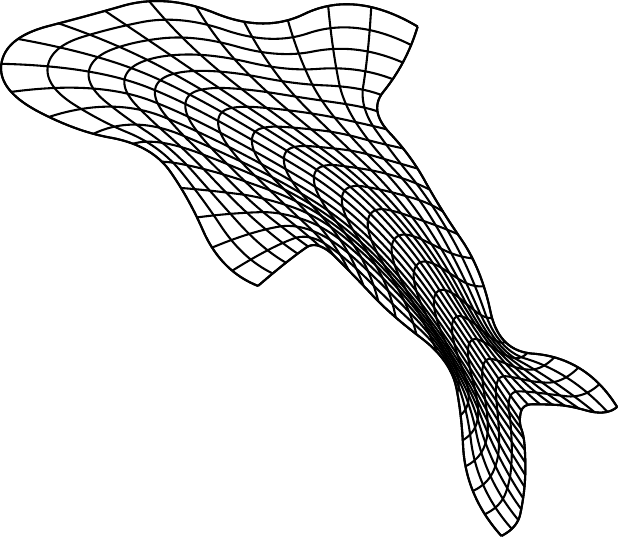}
\includegraphics[width=.16\textwidth]{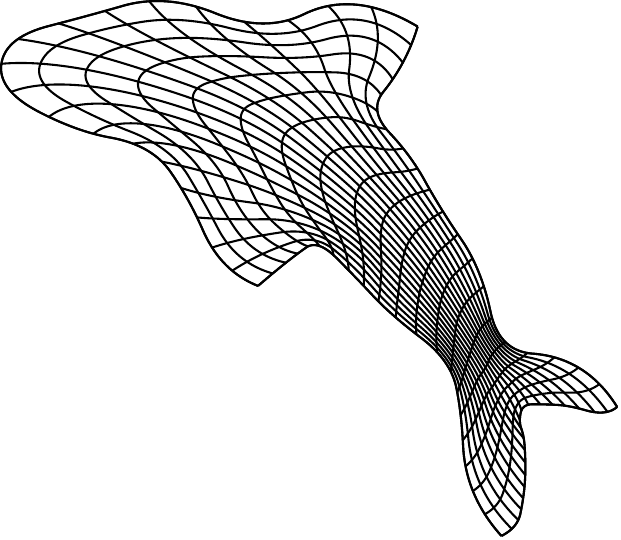}
\includegraphics[width=.16\textwidth]{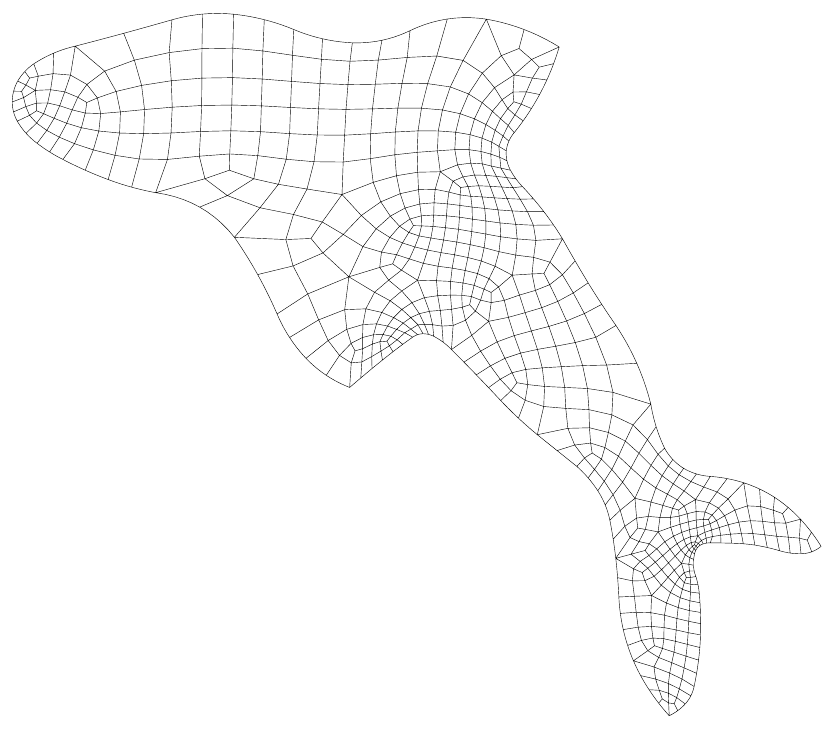}

     \parbox[t]{.16\textwidth}{\centering
            (a) 
           }
  \parbox[t]{.16\textwidth}{\centering
           (b) 
           }
   \parbox[t]{.16\textwidth}{\centering
           (c) 
           }
     \parbox[t]{.16\textwidth}{\centering
            (d) 
           }
  \parbox[t]{.16\textwidth}{\centering
           (e) 
           }
   \parbox[t]{.16\textwidth}{\centering
           (f) 
           }
           
  \caption{\label{vs_parametric_our}
          Comparison on the dolphin model with the parametric methods: (a) NCO~\cite{xu2011parameterization}, (b) VH~\cite{xu2013constructing}, (c) T-Map~\cite{nian2016planar}, (d) LRQC~\cite{pan2018low}, (e) CWA~\cite{ji2021constructing}, (f) our method ($n=2$).}
\end{figure}

\begin{table} % 使用 table* 让表格跨两列
  \centering
  \caption{Display of mesh quality data for dolphin example.}
  \small
  \begin{tabular}{ccccccc}
    \toprule
     & NCO~\cite{xu2011parameterization} & VH~\cite{xu2013constructing} & T-map~\cite{nian2016planar} & LRQC~\cite{pan2018low} & CWA~\cite{ji2021constructing} & Our \\
    \midrule
 $\text{min }J_{m}$&0.09 &-0.65
 & -0.78 & 0.11 & 0.18& \textbf{0.40} \\
    $\text{avg }J_{m}$& 0.42& 0.36&0.45&0.43&0.46& \textbf{0.95} \\
    \bottomrule
    \multicolumn{7}{@{}l@{}}{\footnotesize Results in bold indicate superior performance.} \\
  \end{tabular}
  \label{vs_parameterize_method}
\end{table}

\begin{figure}
  \centering
\includegraphics[width=.35\textwidth]{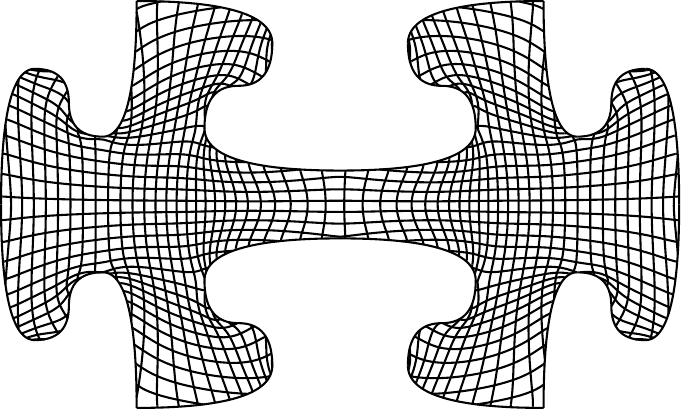}
\includegraphics[width=.35\textwidth]{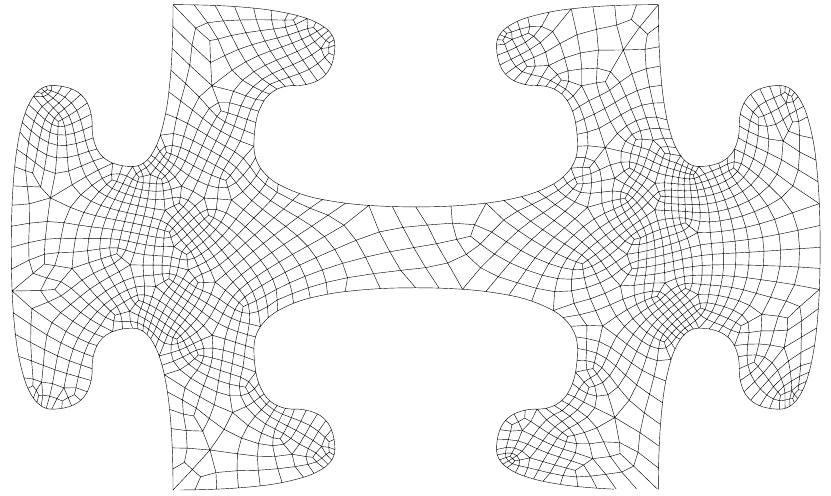}

     \parbox[t]{.35\textwidth}{\centering
            (a) 
           }
  \parbox[t]{.35\textwidth}{\centering
           (b) 
           }
           
  \caption{\label{vs_tcb_our}
           Comparison with TCB-spline-based parameterization. (a) TCB-spline result; (b) our method ($n=2$).
}
\end{figure}

\begin{table} % 使用 table* 让表格跨两列
  \centering
  \caption{Mesh quality comparison with TCB-spline-based method on 10 models.}
  \small
  \begin{tabular}{ccccccccccccc}
    \toprule
      & Method & Butterfly & Dragon & Jigsaw & Penguin & Dog & Scepter & Squirrel & Dolphin & Rabbit & Leaf\\
    \midrule
    $\text{min }J_{m}$ & TCB-spline & 0.16 & 0.13 & 0.12 & 0.31 & 0.12 & 0.20 & 0.19 & \textbf{0.40} & 0.13 & 0.10\\
                    & Our method & \textbf{0.46} & \textbf{0.33} & \textbf{0.45} & \textbf{0.47} & \textbf{0.49} & \textbf{0.51}& \textbf{0.38} &  \textbf{0.40} & \textbf{0.53} & \textbf{0.37}\\
    \midrule
    $\text{avg }J_{m}$ & TCB-spline & 0.97 & 0.93 & 0.89 & \textbf{0.98} & \textbf{0.97} & 0.93 & 0.95 & \textbf{0.98} & 0.96 & 0.95\\
                    & Our method & \textbf{0.98} & \textbf{0.97} & \textbf{0.96} & 0.97 & \textbf{0.97} & \textbf{0.97} & \textbf{0.97} & 0.95 & \textbf{0.97} & \textbf{0.97}\\
     \midrule
     $\text{min }J_{k}$ & TCB-spline & 0.37 & 0.23 & \textbf{0.55} & \textbf{0.47} & 0.23 & 0.22 & 0.23 & \textbf{0.46} & 0.17 & 0.20\\
                    & Our method & \textbf{0.48} & \textbf{0.33} & 0.45 & \textbf{0.47} & \textbf{0.52} & \textbf{0.60}& \textbf{0.40} &  0.40 & \textbf{0.53} & \textbf{0.51}\\
    \midrule
    $\text{avg }J_{k}$ & TCB-spline & 0.98 & 0.96 & 0.94 & \textbf{0.99} & \textbf{0.98} & 0.96 & \textbf{0.98} & \textbf{0.99} & \textbf{0.98} & 0.97\\
                    & Our method & \textbf{0.99} & \textbf{0.98} & \textbf{0.97} & 0.98 & \textbf{0.98} & \textbf{0.98} & \textbf{0.98} & 0.97 & \textbf{0.98} & \textbf{0.98}\\
    \bottomrule
    \multicolumn{12}{@{}l@{}}{\footnotesize Results in bold indicate superior performance.} \\
  \end{tabular}
  \label{vs_tcb_data}
\end{table}

Fig.~\ref{vs_parametric_our} compares meshes generated by our method and various B-spline-based parameterization methods \cite{xu2011parameterization,xu2013constructing,nian2016planar,pan2018low,ji2021constructing}, and Table~\ref{vs_parameterize_method} reports the corresponding shape quality $J_m$. Our method achieves higher minimum and average $J_m$, indicating significantly lower distortion. The TCB-spline-based method \cite{wang2022tcb} supports arbitrary polygonal domains, improving mapping flexibility. However, as shown in Fig.~\ref{vs_tcb_our}, it introduces numerous non-quadrilateral elements along physical boundaries, whereas our method maintains a pure quadrilateral topology. Table~\ref{vs_tcb_data} summarizes experiments on 10 representative models, demonstrating superior mesh quality and robustness for complex boundaries, with only a few simple models showing comparable performance.

\textbf{Comparison with Gmsh.} \textcolor{black}{Fig.~\ref{using_our_vs_gmsh} shows a comparison between the open-source software Gmsh~\cite{geuzaine2009gmsh} and the proposed algorithm for high-order quadrilateral mesh generation. During mesh generation, Gmsh does not explicitly consider the relationship between boundary segmentation and high-order element quality, resulting in a minimum element quality $\min J_m=0.015$. In contrast, the proposed algorithm accounts for the interplay among curve partitioning, high-order element quality, and bounded geometric error during boundary segmentation, leading to a significantly improved mesh with a minimum element quality of $J_m=0.47$.}

\begin{figure}
  \centering
\includegraphics[width=.9\textwidth]{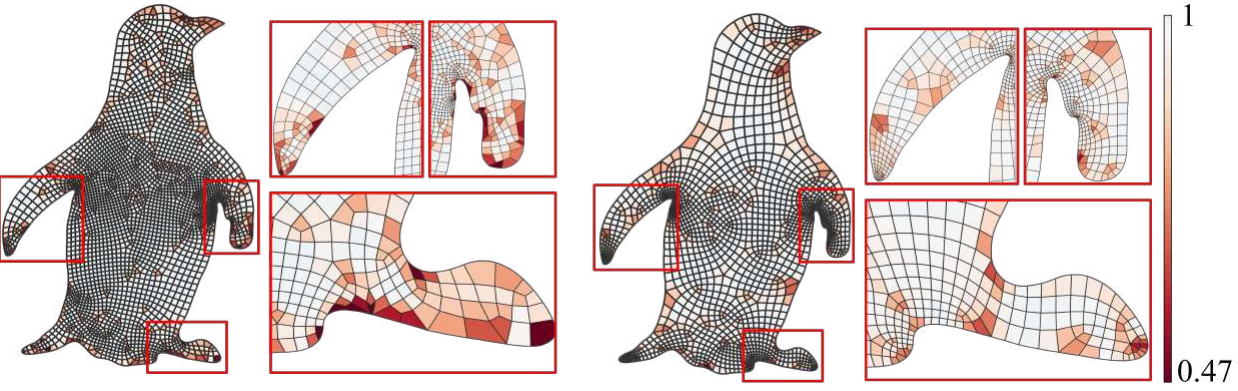}
           
  \caption{\label{using_our_vs_gmsh}
           Comparison of remeshing results ($n=2$) using Gmsh (left, 3,302 elements, $\min J_m = 0.015$) and our method (right, 2,046 elements, $\min J_m = 0.47$). The minimum shape measure $J_m$ is color-coded, with values below $0.47$ highlighted in dark red.}
\end{figure}

\textbf{Comparison with TMOP method.} TMOP~\cite{dobrev2019target} is a high-order mesh optimization method that improves mesh quality by minimizing an energy functional measuring the deviation between target and current element matrices. Fig.~\ref{remeshing_using_our} compares applying TMOP directly to a given high-order quadrilateral mesh versus remeshing it with our boundary/interface-aware approach. For the input mesh (Figs.~\ref{remeshing_using_our}(a–b)), TMOP preserves boundary/interface vertices and only relocates interior vertices without altering topology, yielding limited quality improvement for complex boundaries (Fig.~\ref{remeshing_using_our}(c)). In contrast, our geometry-preserving curve reconstruction-based remeshing significantly enhances mesh quality (Fig.~\ref{remeshing_using_our}(d)). By reformulating high-order optimization as boundary curve reconstruction, we reduce computational complexity—achieving a runtime of $0.11$s versus TMOP’s $1.18$s. Moreover, as shown in Fig.~\ref{using_our_to_opt}, further optimizing our result with TMOP yields negligible gains: Table~\ref{using_our_mesh_to_opt_quality} shows only a slight increase in average $J_m$ but a decrease in minimum $J_m$. These results demonstrate that our method produces higher-quality high-order quadrilateral meshes more effectively and efficiently than TMOP.

\begin{figure}
  \centering
\includegraphics[width=.24\textwidth]{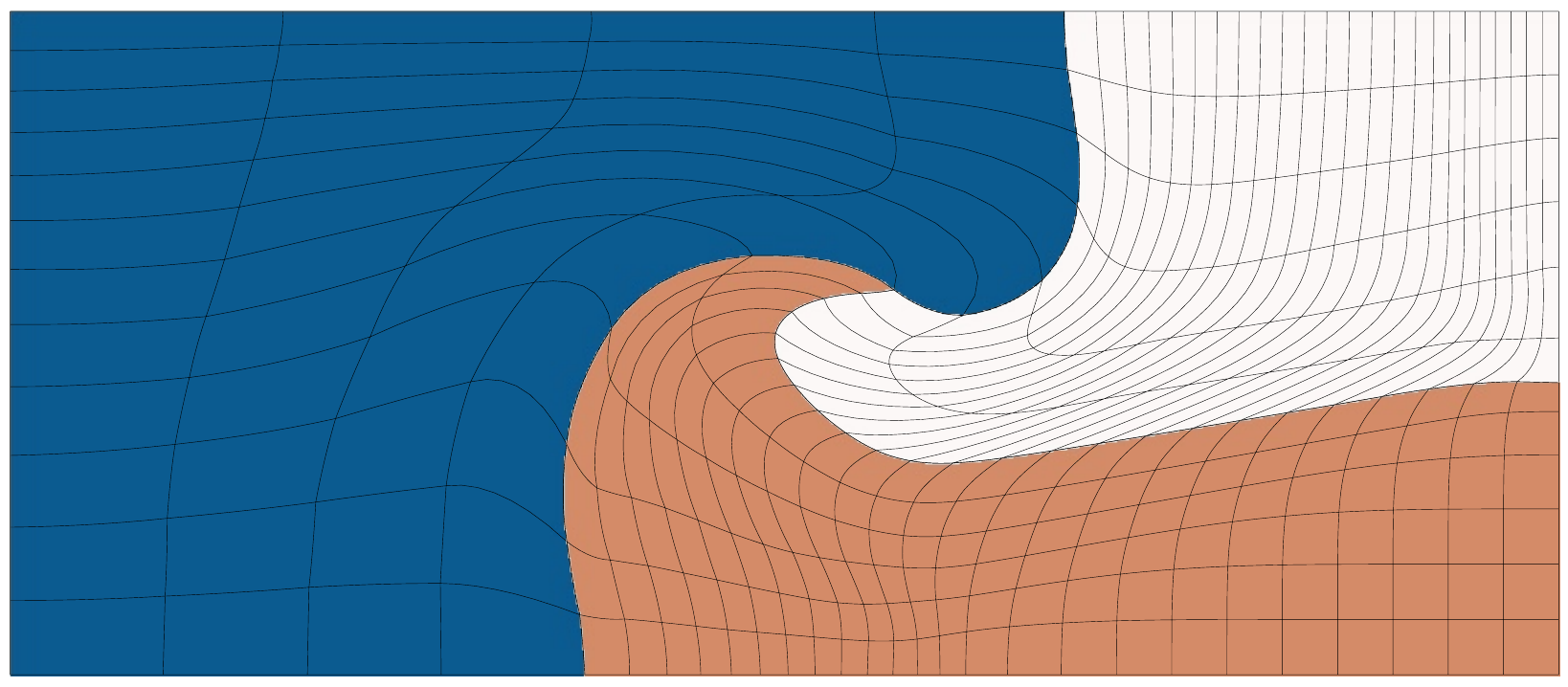}
\includegraphics[width=.24\textwidth]{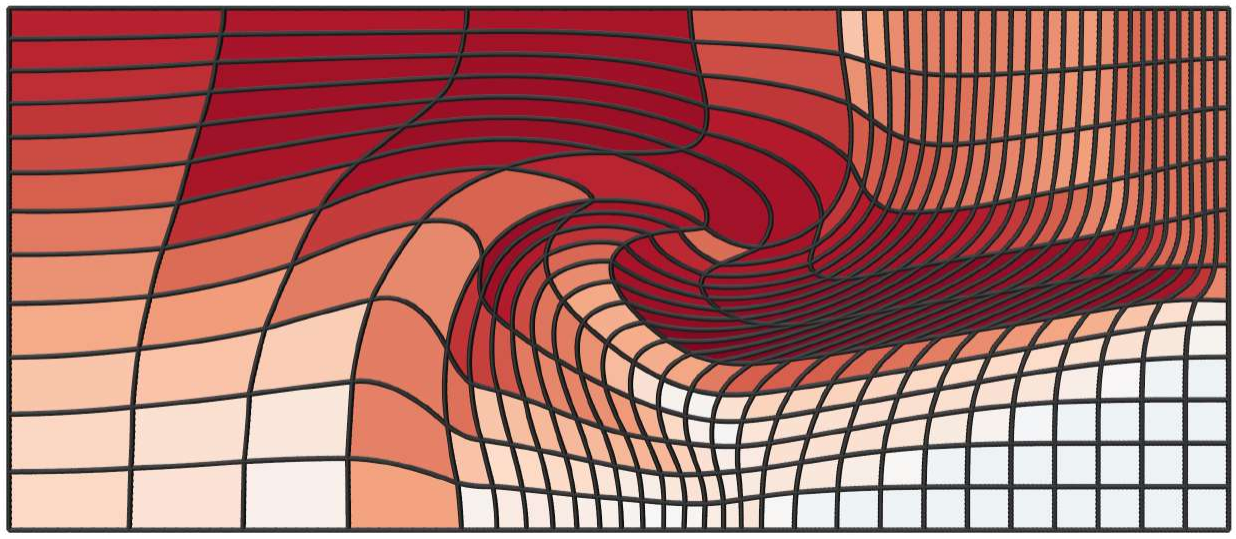}
\includegraphics[width=.24\textwidth]{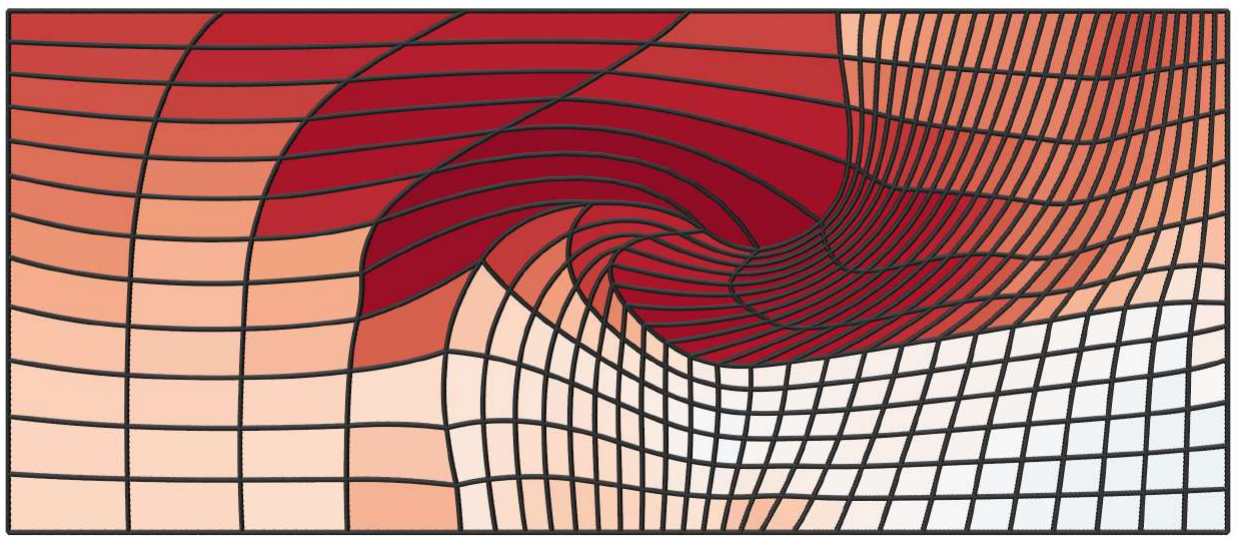}
\includegraphics[width=.257\textwidth]{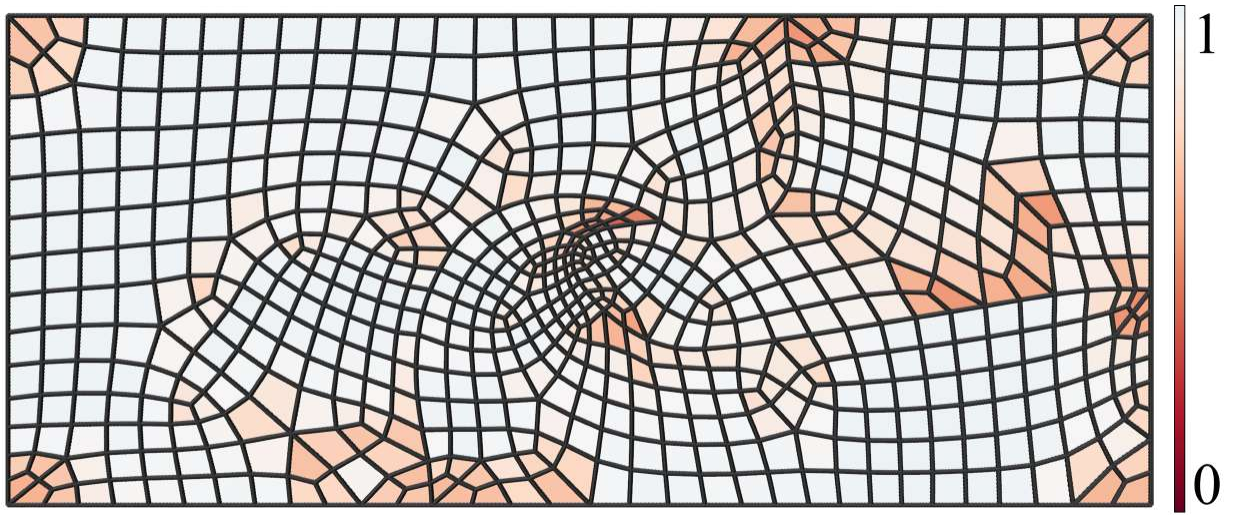}

     \parbox[t]{.24\textwidth}{\centering
            (a) 
           }
  \parbox[t]{.24\textwidth}{\centering
           (b) 
           }
        \parbox[t]{.24\textwidth}{\centering
            (c) 
           }
  \parbox[t]{.25\textwidth}{\centering
           (d) 
           }
           
   \caption{\label{remeshing_using_our}
           Comparison of mesh quality between our method and TMOP. (a) Input quadrilateral mesh colored by material; (b–d) meshes colored by minimum shape measure $J_m$: (b) input mesh, (c) TMOP result, (d) our result ($n=2$).}
\end{figure}

\begin{figure}
  \centering
  % \adjustbox{scale=0.9}{
\includegraphics[width=.30\textwidth]{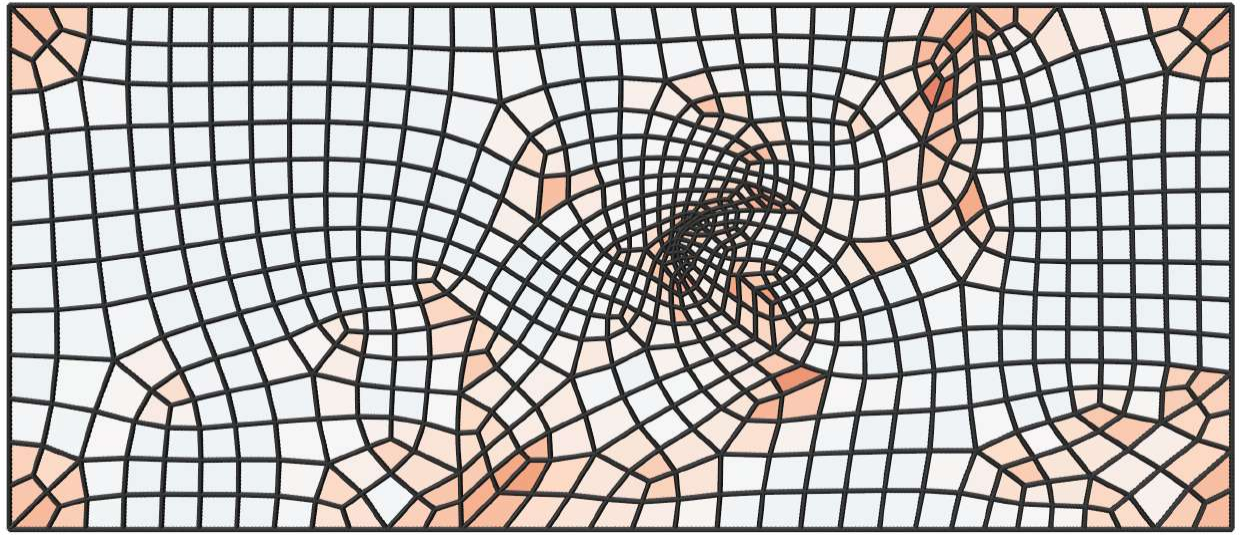}
\includegraphics[width=.315\textwidth]{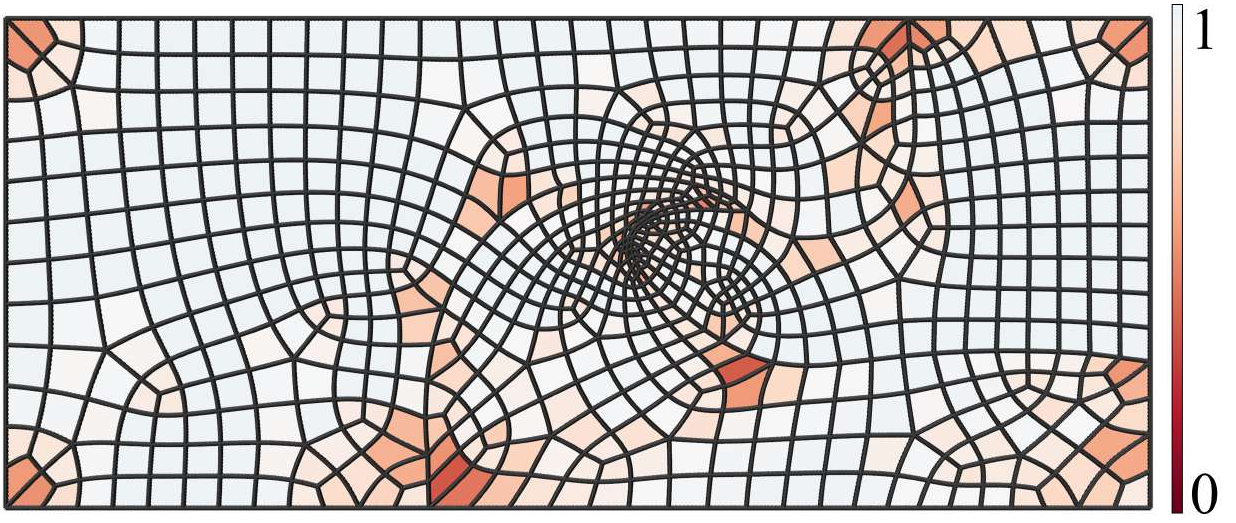}

     \parbox[t]{.30\textwidth}{\centering
            (a) 
           }
  \parbox[t]{.321\textwidth}{\centering
           (b) 
           }
           % }
  \caption{\label{using_our_to_opt}
           Comparison of the remeshing result ($n=2$) using our method (a) and the result after further optimization with TMOP (b). The color of element represents its minimum shape measure $J_m$.}
\end{figure}

\begin{table} % 使用 table* 让表格跨两列
  \centering
  \caption{The mesh quality data of the model in Fig. \ref{using_our_to_opt}.}
  \small
  \begin{tabular}{cccc}
    \toprule
      & Times$/s$ & $\min J_{m}$&$\text{avg } J_{m}$\\
    \midrule
    Our method& -&\textbf{0.48}&0.95 \\
    TMOP & 0.30 & 0.36 & \textbf{0.97}\\
    \bottomrule
    \multicolumn{4}{@{}l@{}}{\footnotesize Results in bold indicate superior performance.} \\
  \end{tabular}
  \label{using_our_mesh_to_opt_quality}
\end{table}

\section{Application in ALE method}

The high-order, high-quality quadrilateral meshes generated by our method are broadly useful for downstream simulations. As an example, in this section we integrate them with a high-order ALE scheme to enable robust, interface-preserving simulations of multi-material large-deformation problems.

High-order Lagrangian method~\cite{dobrev2012high} are a classical approach for multi-material large-deformation hydrodynamics, where the mesh moves with the material—naturally preserving clear interfaces and tracking motion accurately. However, severe mesh distortion or inversion under large deformations can restrict time steps, reduce efficiency, or cause failure. High-order ALE avoids this by decoupling mesh motion from material motion through an independent mesh update. The robustness of high-order ALE critically depends on interface-preserving, high-quality meshes. Existing high-order mesh generation approaches fall into two categories, neither of which fully satisfies this requirement: topology-preserving methods such as TMOP~\cite{dobrev2019target}, which maintain mesh connectivity but inevitably introduce low-quality elements when preserving complex interfaces (see Fig.~\ref{remeshing_using_our}); and topology-changing methods such as hr-adaptivity, which improve flexibility and mesh quality but often blur material interfaces~\cite{dobrev2022hr}. In contrast, our method achieves topology-changing remeshing while preserving sharp interfaces, enhancing the robustness of high-order ALE in multi-material large-deformation simulations.

\begin{figure}
  \centering
\includegraphics[width=1\textwidth]{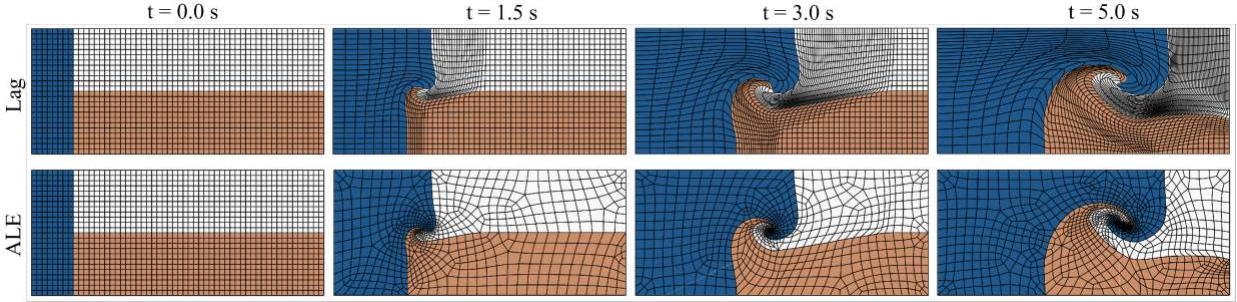}
\caption{\label{ale_application}
           Triple-point problem ($n=2$) solved using the high-order Lagrangian (Lag) method (top) and the high-order ALE method combined with our meshing approach (bottom). The Lag method is computed on the evolving mesh from the beginning, whereas the ALE method remeshes when mesh quality drops, achieving superior results with fewer elements, larger time steps, and shorter simulation times.}
\end{figure}

% ...

\begin{table}
\centering
\caption{Simulation and final mesh quality data for Lagrangian and ALE methods (see Fig.~\ref{ale_application}).}
\label{ale_application_table}
\begin{threeparttable}
\small
\begin{tabular}{ccccccc}
\toprule
         & \multirow{2}{*}{Step} & \multirow{2}{*}{$\min dt$} & \multirow{2}{*}{Time/s} & \multicolumn{3}{c}{Mesh at $t=5$} \\
         &                       &                            &                         & Number of elements & $\min J_m$ & $\text{avg } J_m$ \\
\midrule
Lagrange & 59,986                 & 0.000042                   & 65251.28                & 1,344              & 0.043      & 0.473      \\
ALE      & \textbf{10,269}        & \textbf{0.000279}          & \textbf{25056.66}       & \textbf{1,218}     & \textbf{0.204} & \textbf{0.815} \\
\bottomrule
\end{tabular}
\begin{tablenotes}
\footnotesize
\item Step: total number of time integration steps performed during the simulation; $\min dt$: the smallest time step size encountered over the entire simulation; Time/s: total simulation time. Results in bold indicate superior performance.
\end{tablenotes}
\end{threeparttable}
\end{table}
%\item[]

% \begin{table*}[htb]
% \centering
% \caption{Simulation and final mesh quality data for Lagrangian and ALE methods (see Fig.~\ref{ale_application}).}
%   \small
% \begin{tabular}{ccccccc}
% \toprule
%          & \multirow{2}{*}{Step} & \multirow{2}{*}{$\min dt$} & \multirow{2}{*}{Time/s} & \multicolumn{3}{c}{Mesh at $t=5$}              \\
%          &                       &                            &                         & Number of element & $\min J_m$ & $\text{avg } J_m$ \\
%          \midrule
% Lagrange & 59986                 & 0.000042                   & 65251.28                & 1344              & 0.043      & 0.473      \\
% ALE      & \textbf{10269}                 & \textbf{0.000279}                   & \textbf{25056.66}                & \textbf{1218}              & \textbf{0.204}      & \textbf{0.815}   \\  
% \bottomrule
% \multicolumn{7}{@{}p{\dimexpr\linewidth-2\tabcolsep}@{}}{%
%   \footnotesize
%   {\color{orange}Step: total number of time integration steps performed during the simulation; 
%   $\min \Delta t$: the smallest time step size encountered over the entire simulation; }
%   Time/s: total simulation time. Results in bold indicate superior performance.
% }\\
% \end{tabular}
% \label{ale_application_table}
% \end{table*}

We validate our method on the triple-point problem~\cite{galera2010two,kucharik2014conservative}, a two-material, three-state Riemann problem that generates vorticity and serves as a classic 2D multi-material hydrodynamics benchmark. This problem tests mesh distortion handling by producing fine-scale vortex structures that become increasingly intricate with mesh refinement~\cite{kenamond2021intersection}. Fig.~\ref{ale_application} shows that high-order ALE coupled with our mesh generation method preserves clear interfaces and resolves finer vortices compared to the pure Lagrangian approach. Table~\ref{ale_application_table} further demonstrates that high-order ALE with our adaptive meshes achieves superior results with fewer elements, larger time steps, fewer steps, and shorter simulation times.

\section{Conclusion}

In this paper, we presents an efficient method for generating high-quality, high-order quadrilateral meshes with bounded geometric error, based on curve reconstruction. By analyzing the relationship between boundary curve distribution and mesh quality, we propose a curve reconstruction strategy that preserves geometric accuracy, transforming the challenging mesh optimization problem into a more tractable boundary reconstruction task, thereby significantly improving efficiency. Experimental results demonstrate that the proposed method efficiently generates high-quality, high-order quadrilateral meshes with guaranteed geometric accuracy and supports complex models with interfaces. We further apply the our meshing method to ALE simulations, demonstrating enhanced robustness of high-order ALE and its ability to avoid numerical errors caused by inverted or low-quality elements in pure Lagrangian methods.

\textbf{Limitation and future work.} Proposed method is limited to cases where the input 2D curve degree matches the output mesh order, and future work will extend to generating high-order meshes with geometric error bounds for mismatched degrees, as well as to 3D domains. Moreover, although adaptive refinement and endpoint optimization rarely increase geometric error, we ensure the error never exceeds the prescribed threshold by checking and bisecting any curve segment that violates it after each step. Our method produces high-quality, inversion-free high-order quadrilateral meshes for all models in our dataset, but lacks rigorous theoretical guarantees against inverted or low-quality elements. Future work will focus on generating provably valid, high-quality quadrilateral meshes. Furthermore, our method has been extensively tested for moderate orders ($n \leq 4$) and its scalability has been demonstrated for higher orders (e.g., $n = 10$). However, the use of equidistant nodes in the current formulation may lead to numerical instability for very high-order meshes. In future work, we plan to adopt non-equidistant node distributions, such as Chebyshev nodes, to improve numerical stability in very-high-order mesh generation.




%% For numbered reference style
% \bibitem{egbibsample}

%   \textit{\LaTeX: a document preparation system},
%   Addison Wesley, Massachusetts,

% \end{thebibliography}

\appendix
\section{Proof of Lemma \ref{lemma 1}}\label{Proof of Lemma 1}
\begin{proof}
By the definition of Hausdorff distance,  
\begin{equation}
d_H(\mathbf{l}_1, \mathbf{l}_2) = \max\left\{\adjustlimits\sup_{\mathbf{x} \in \mathbf{l}_1} \inf_{\mathbf{y} \in \mathbf{l}_2} \|\mathbf{x} - \mathbf{y}\|_2,\,\adjustlimits\sup_{\mathbf{y} \in \mathbf{l}_2} \inf_{\mathbf{x} \in \mathbf{l}_1} \|\mathbf{x} - \mathbf{y}\|_2 \right\},
\end{equation}
which implies  
\begin{equation}
d_H(\mathbf{l}_1, \mathbf{l}_2) \leq \max_t \|\mathbf{l}_1(t) - \mathbf{l}_2(t)\|_2.
\end{equation} 
Expressing the B\'ezier curves in terms of Bernstein basis functions $B_i^n(t)$, we have  
\begin{equation}
\mathbf{l}_j(t) = \sum_{i=0}^n \mathbf{Q}_{ij} B_i^n(t), \quad j = 1, 2.
\end{equation} 
Substituting this into the norm yields  
\begin{align}
\max_t \|\mathbf{l}_1(t) - \mathbf{l}_2(t)\|_2 
&= \max_t \left\| \sum_{i=0}^n (\mathbf{Q}_{i1} - \mathbf{Q}_{i2}) B_i^n(t) \right\|_2 \\
&\leq \max_t \sum_{i=0}^n B_i^n(t) \|\mathbf{Q}_{i1} - \mathbf{Q}_{i2}\|_2 \\
&\leq \max_i \|\mathbf{Q}_{i1} - \mathbf{Q}_{i2}\|_2,
\end{align}
where the inequality follows from the convexity of the norm and the fact that $\sum_{i=0}^n B_i^n(t) = 1$, $B_i^n(t) \geq 0$.  
Therefore, $d_H(\mathbf{l}_1, \mathbf{l}_2) \leq d_{m}(\mathbf{l}_{1},\mathbf{l}_{2})$, as required. 
\end{proof}

\section{Proof of Lemma \ref{lemma 2}}\label{Proof of Lemma 2}
\begin{proof}
We observe that $\mathbf{l}_2^1 \subseteq \mathbf{l}_2$ implies  
\begin{equation}
\adjustlimits\sup_{\mathbf{x} \in \mathbf{l}_1^1} \inf_{\mathbf{y} \in \mathbf{l}_2} d(\mathbf{x}, \mathbf{y}) \leq \adjustlimits\sup_{\mathbf{x} \in \mathbf{l}_1^1} \inf_{\mathbf{y} \in \mathbf{l}_2^1} d(\mathbf{x}, \mathbf{y}).
\end{equation} 
% \begin{equation}
% \sup_{x \in l_1^1} \inf_{y \in l_2} d(x, y) = \min\{\sup_{x \in l_1^1} \inf_{y \in l_2^1} d(x, y),\sup_{x \in l_1^1} \inf_{y \in l_2^2} d(x, y)\}.
% \end{equation} 
Similarly, $\mathbf{l}_2^2 \subseteq \mathbf{l}_2$ yields  
\begin{equation}
\adjustlimits\sup_{\mathbf{x} \in \mathbf{l}_1^2} \inf_{\mathbf{y} \in \mathbf{l}_2} d(\mathbf{x}, \mathbf{y}) \leq \adjustlimits\sup_{\mathbf{x} \in \mathbf{l}_1^2} \inf_{\mathbf{y} \in \mathbf{l}_2^2} d(\mathbf{x}, \mathbf{y}).
\end{equation} 
% \begin{equation}
% \sup_{x \in l_1^2} \inf_{y \in l_2} d(x, y) =\min\{ \sup_{x \in l_1^2} \inf_{y \in l_2^1} d(x, y),\sup_{x \in l_1^2} \inf_{y \in l_2^2} d(x, y)\}.
% \end{equation}
Therefore,  
\begin{align}
\adjustlimits\sup_{\mathbf{x} \in \mathbf{l}_1} \inf_{\mathbf{y} \in \mathbf{l}_2} d(\mathbf{x}, \mathbf{y}) 
&= \max\left\{ \adjustlimits\sup_{\mathbf{x} \in \mathbf{l}_1^1} \inf_{\mathbf{y} \in \mathbf{l}_2} d(\mathbf{x}, \mathbf{y}),\, \adjustlimits\sup_{\mathbf{x} \in \mathbf{l}_1^2} \inf_{\mathbf{y} \in \mathbf{l}_2} d(\mathbf{x}, \mathbf{y}) \right\} \\
&\leq \max\left\{\adjustlimits\sup_{\mathbf{x} \in \mathbf{l}_1^1} \inf_{\mathbf{y} \in \mathbf{l}_2^1} d(\mathbf{x}, \mathbf{y}),\, \adjustlimits\sup_{\mathbf{x} \in \mathbf{l}_1^2} \inf_{\mathbf{y} \in \mathbf{l}_2^2} d(\mathbf{x}, \mathbf{y})\} \right\}.
\end{align}
Analogously,  
\begin{equation}
\adjustlimits\sup_{\mathbf{y} \in \mathbf{l}_2} \inf_{\mathbf{x} \in \mathbf{l}_1} d(\mathbf{x}, \mathbf{y}) \leq \max\left\{ \adjustlimits\sup_{\mathbf{y} \in \mathbf{l}_2^1} \inf_{\mathbf{x} \in \mathbf{l}_1^1} d(\mathbf{x}, \mathbf{y}),\, \adjustlimits\sup_{\mathbf{y} \in \mathbf{l}_2^2} \inf_{\mathbf{x} \in \mathbf{l}_1^2} d(\mathbf{x}, \mathbf{y}) \right\}.
\end{equation} 
Combining both bounds, we conclude  
\begin{equation}
d_H(\mathbf{l}_1, \mathbf{l}_2) \leq \max\left\{ d_H(\mathbf{l}_1^1, \mathbf{l}_2^1),\, d_H(\mathbf{l}_1^2, \mathbf{l}_2^2) \right\}.
\end{equation} 
Based on the splitting scheme in Algorithm 1 and Lemma \ref{lemma 1}, we can reuse the above proof process to obtain
\begin{align}
d_H(\mathbf{l}_1, \mathbf{l}_2) 
&\leq \max\left\{ d_H(\mathbf{l}_1^1, \mathbf{l}_2^1),\, d_H(\mathbf{l}_1^2, \mathbf{l}_2^2) \right\} \\
&\leq \max_{i}\left\{ d_H(\mathbf{l}_1^i, \mathbf{l}_2^i) \right\} \label{equation in lemma 2}\\
&\leq \max_{i}\left\{ d_m(\mathbf{l}_1^i, \mathbf{l}_2^i) \right\}.
\end{align}

% Expressing the B\'ezier curves in terms of Bernstein basis functions $B_i^n(t)$, we have  
% \begin{equation}
% \mathbf{l}_j^{k}(t) = \sum_{i=0}^n \mathbf{Q}_{ij}^{k} B_i^n(t), \quad j = 1, 2.
% \end{equation} 

% By the de Casteljau algorithm, the control vertices after splitting are convex combinations of the original curve's control points, 
% \begin{equation}
% \mathbf{Q}_{i_{0}j}^{k} = \sum_{i=0}^{i_0}\alpha_{ij}^{k} \mathbf{Q}_{ij}, \quad \alpha_{ij}^{k}\geq 0 ,\quad\sum_{i=0}^{i_0}\alpha_{ij}^{k}=1.
% \end{equation} 
% Thus, the distance between corresponding control vertices after splitting is less than or equal to that before splitting,
% \begin{align}
% \|\mathbf{Q}_{i_{0}1}^{k}-\mathbf{Q}_{i_{0}2}^{k}\|_{2}
% &=\|\sum_{i=0}^{i_0}\alpha_{i1}^{k} \mathbf{Q}_{i1}-\sum_{i=0}^{i_0}\alpha_{i2}^{k} \mathbf{Q}_{i2}\|_{2}\\
% &\leq \sum_{i=0}^{i_0}\sum_{i=0}^{i_0}\alpha_{i1}^{k}\alpha_{i2}^{k}\|\mathbf{Q}_{i1}-\mathbf{Q}_{i2}\|_{2}\\
% &\leq \max_i{\|\mathbf{Q}_{i1}-\mathbf{Q}_{i2}\|_{2}}.
% \end{align}

Finally, as the number of subdivisions increases ($k \to \infty$), the maximum distance between corresponding control points, $b_k = \min\{b_{k-1},\max\limits_i\{d_m(\mathbf{l}_1^i, \mathbf{l}_2^i)\}\}$, decreases monotonically and is bounded below by $d_H(\mathbf{l}_1, \mathbf{l}_2)$. 
% Hence, by the monotone convergence theorem, the sequence $\{b_k\}$ converges.
\end{proof}

\section{Proof of Theorem \ref{hausdorff geometric error theorem}}\label{proof of theorem hausdorff geometric error}
\begin{proof}
According to Algorithm \ref{alg:Framwork}, we have $ d_{H}(\mathbf{l}_{1},\mathbf{l}_{2}) < b_{k} $ and $ d_{H}(\mathbf{l}_{1},\mathbf{l}_{2}) < b'_{k} $, from which it follows directly that $ d_{H}(\mathbf{l}_{1},\mathbf{l}_{2}) < \hat{b}_{k} $. 
\end{proof}

\section{Proof of Theorem \ref{hausdorff geometric error theorem for multicurve}}\label{proof of theorem hausdorff geometric error}
\begin{proof}
According to the Eq.(\ref{equation in lemma 2}) of Lemma \ref{lemma 2} and Theorem \ref{hausdorff geometric error theorem}, it follows that $ d_H(\mathcal{A}, \mathbf{l}) \leq \max\limits_{i}\{d_{H}(\mathbf{l}_i, \mathbf{l}^*_i)\} \leq \max\limits_i \, \hat{b}_k(\mathbf{l}_i, \mathbf{l}^*_i) $. 
\end{proof}

\section{Proof of Theorem \ref{valid high-order quad theorem}}\label{proof of theorem valid high-order quad}
\begin{proof}
The Jacobian matrix of a high-order quadrilateral element $\psi$ is given by $J = \left( \frac{\partial\psi}{\partial\xi}, \frac{\partial\psi}{\partial\eta} \right)$. When $\psi$ is expressed in the Bernstein basis, the partial derivatives are:

\begin{equation}
\frac{\partial\psi}{\partial\xi} = n \sum_{j=0}^{n} \sum_{i=0}^{n-1} \Delta_{ij}^{0} B_{i}^{n-1}(\xi) B_{j}^{n}(\eta),
\end{equation}
\begin{equation}
\frac{\partial\psi}{\partial\eta} = n \sum_{i=0}^{n} \sum_{j=0}^{n-1} \Delta_{ij}^{1} B_{i}^{n}(\xi) B_{j}^{n-1}(\eta).
\end{equation}

 Since $\sum_{i=0}^{n} \sum_{j=0}^{m} B_{i}^{n}(\xi) B_{j}^{m}(\eta) \equiv 1$ and $B_{i}^{n}(\xi),B_{j}^{m}(\eta)\in[0,1]$, the derivatives $\frac{\partial\psi}{\partial\xi}$ and $\frac{\partial\psi}{\partial\eta}$ are convex combinations of the vectors $\{\Delta_{ij}^{0}\}$ and $\{\Delta_{ij}^{1}\}$, respectively. Hence, $\frac{\partial\psi}{\partial\xi} \in r_0$ and $\frac{\partial\psi}{\partial\eta} \in r_1$, where $r_0$ and $r_1$ are the sectors spanned by the 0-vectors and 1-vectors. If $r_0 \cap r_1 = \emptyset$ and both sectors lie strictly on the same side of a line, then $\frac{\partial\psi}{\partial\xi}$ and $\frac{\partial\psi}{\partial\eta}$ are linearly independent and form a right-handed system, which implies $J > 0$.

\end{proof}



% \textbf{Acknowledgements.} This work was supported by the National Key R&D Program of  

\bibliographystyle{model1-num-names}
% \bibliographystyle{cas-model2-names}

\bibliography{cas-refs}

\end{document}